\documentclass[12pt]{article}
\usepackage{amssymb}
\usepackage{amsfonts}
\usepackage{amsmath}
\usepackage{harvard}
\usepackage{chicago}
\usepackage{graphicx}
\usepackage{float}
\usepackage{amsthm}

\newtheorem{theorem}{Theorem}

\newtheorem{lemma}[theorem]{Lemma}

\newtheorem{corollary}[theorem]{Corollary}

\theoremstyle{definition}
\newtheorem{remark}{Remark}
\numberwithin{remark}{section}
\floatstyle{ruled}
\newfloat{algorithm}{tbp}{loa}
\providecommand{\algorithmname}{Algorithm}
\floatname{algorithm}{\protect\algorithmname}
\input{tcilatex}
\renewcommand{\cite}{\citeasnoun}
\newcommand{\citet}{\cite}
\renewcommand{\thepage}{}
\renewcommand{\thefootnote}{\fnsymbol{footnote}}
\allowdisplaybreaks
\renewcommand{\baselinestretch}{1.3}

\begin{document}

\title{Spatial Correlation Robust Inference\thanks{%
M{\"{u}}ller acknowledges financial support from the National Science
Foundation grant SES-191336.}}
\author{Ulrich K. M{\"{u}}ller and Mark W. Watson ~~ \\
\relax Department of Economics, Princeton University~~\\
\relax Princeton, NJ, 08544 \leavevmode \leavevmode }
\date{First Draft: December 2020\\
This Draft: February 2021}
\maketitle

\begin{abstract}
We propose a method for constructing confidence intervals that account for
many forms of spatial correlation. The interval has the familiar `estimator
plus and minus a standard error times a critical value' form, but we propose
new methods for constructing the standard error and the critical value. The
standard error is constructed using population principal components from a
given `worst-case' spatial covariance model. The critical value is chosen to
ensure coverage in a benchmark parametric model for the spatial
correlations. The method is shown to control coverage in large samples
whenever the spatial correlation is weak, i.e., with average pairwise
correlations that vanish as the sample size gets large. We also provide
results on correct coverage in a restricted but nonparametric class of
strong spatial correlations, as well as on the efficiency of the method. In
a design calibrated to match economic activity in U.S. states the method
outperforms previous suggestions for spatially robust inference about the
population mean.

Key Words: Confidence interval, HAR, HAC, Random field

JEL: C12, C20
\end{abstract}

\newpage \setcounter{page}{1} \renewcommand{\thepage}{\arabic{page}} %
\renewcommand{\thefootnote}{\arabic{footnote}} 
\renewcommand{\baselinestretch}{1.23} \small \normalsize%

\section{\label{sec:Introduction}Introduction}

Prompted by advances in both data availability and theory in economic
geography, international trade, urban economics, development and other
fields, empirical work using spatial data has become commonplace in
economics. These applications highlight the importance of econometric
methods that appropriately account for spatial correlation in real-world
settings. While important advances have been made, researchers arguably lack
practical methods that allow for reliable inference about parameters
estimated from spatial data for the wide-range spatial designs and
correlation patterns encountered in applied work.\footnote{\cite{Ibragimov10}%
, \cite{Sun_Kim_2012} and \cite{Bester_Conley_Hansen_Vogelsang_2016}, for
instance, find nontrivial size distortions of modern methods even in
arguably fairly benign designs, and \cite{Kelly2019} reports very large
distortions under spatial correlations calibrated to real-world data.} This
paper takes a step forward in this regard.

Specifically, we consider the problem of constructing a confidence interval
(or test of a hypothesized value) for the mean of a spatially-sampled random
variable. We propose a confidence interval constructed in the usual way,
i.e., as the sample mean plus and minus an estimate of its standard error
multiplied by a critical value. The novelty is that the standard error and
critical value are constructed so the resulting confidence interval has the
desired large-sample coverage probability (say, $95\%$) for a relatively
wide range of correlation patterns and spatial designs. The analysis is
described for the mean, but the required modifications for regression
coefficients or parameters in GMM settings follow from standard arguments.

To be more precise, suppose that a random variable $y$ is associated with a
location $s\in \mathcal{S}$, where $\mathcal{S}\subset \mathbb{R}^{d}$.
Figure \ref{fig:Three-Time-Series Designs} shows three one-dimensional ($d=1$%
) spatial designs. Panel (a) shows the familiar case of regularly spaced
locations, corresponding to the standard time series setting; panels (b) and
(c) show randomly selected locations drawn from a density $g$, where $g$ is
uniform in panel (b) and triangular in panel (c). Figure 2 shows two
geographic examples, so $d=2$, for the U.S.~state of Texas. In panel (a),
the locations are randomly selected from a uniform distribution, while in
panel (b) they are more likely to be sampled from areas with high economic
activity, here measured by light intensity as seen from space.\footnote{%
The light data are from \cite{Henderson_et_al_light}.} In much of our
analysis, we will assume that locations are i.i.d.~draws from a distribution
with density $g$, and so will encompass the irregularly spaced time series
and Texas examples.

\begin{figure}[t]
\begin{centering}
\includegraphics[scale=0.75]{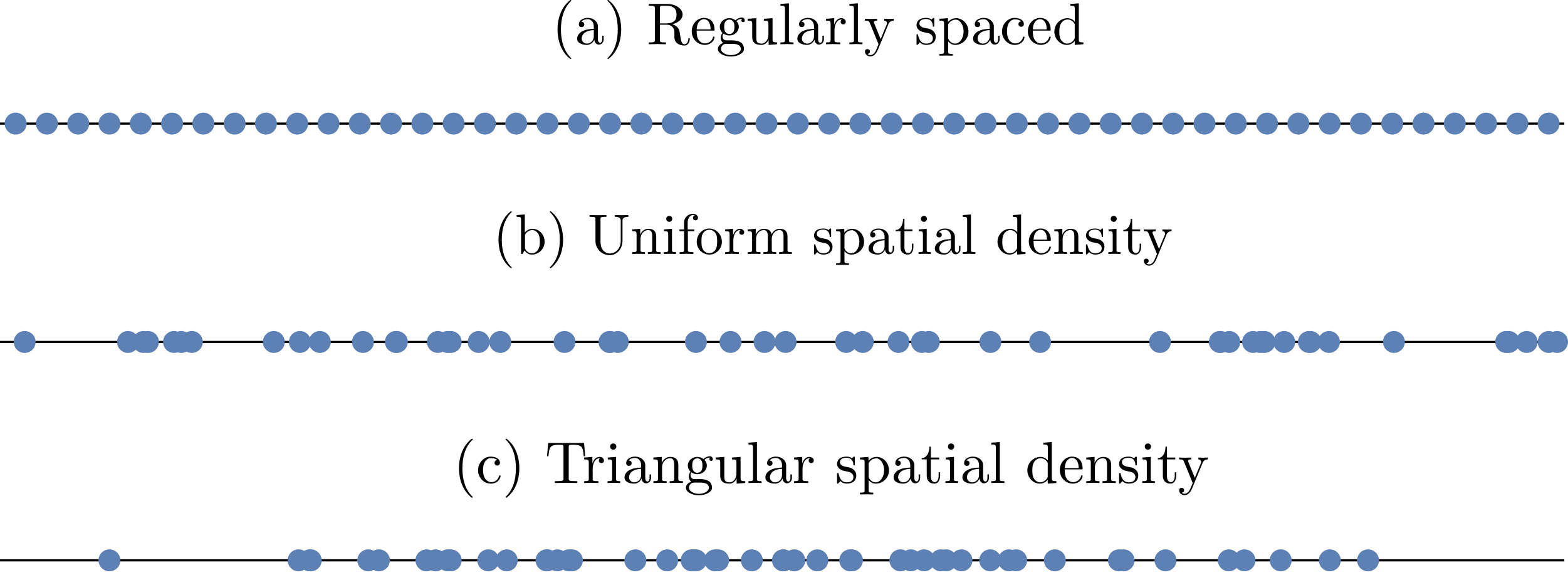} 
\par\end{centering}
\vspace{0.3in}
\caption{Three One-Dimensional Spatial Designs}
\label{fig:Three-Time-Series Designs}
\end{figure}

\begin{figure}[tbp]
\begin{centering}
\includegraphics[scale=0.65]{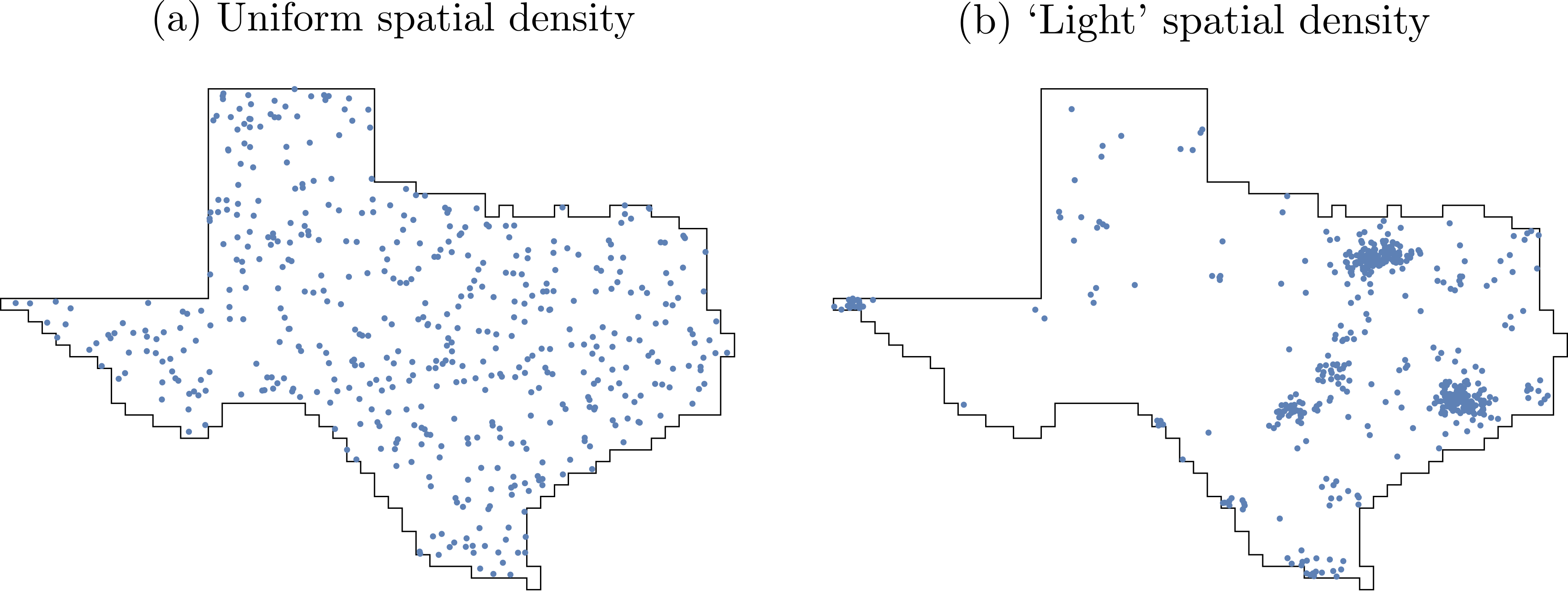}
\par\end{centering}
\caption{Two Geographic Spatial Designs}
\label{fig:Two-Spatial-Designs}
\end{figure}

Adding some notation, suppose 
\begin{equation}
y_{l}=\mu+u_{l}\text{ for $l=1,...,n$}
\label{eq:y=00003D00003D00003D00003Dmu+u}
\end{equation}
where $y_{l}$ is associated with the spatial location $s_{l}$, $\mu$ is the
mean of $y_{l}$, and $u_{l}$ is an unobserved error, assumed to be
covariance stationary with mean zero and covariance function $\mathbb{E}%
[u(r)u(s)]=\sigma_{u}(r-s)$. Let $\overline{y}$ denote the sample mean, and
consider the usual t-statistic 
\begin{equation*}
\tau=\frac{\sqrt{n}(\overline{y}-\mu_{0})}{\hat{\sigma}}
\end{equation*}
where $\hat{\sigma}^{2}$ is an estimator for the variance of $\sqrt{n}(%
\overline{y}-\mu)$. Tests of the null hypothesis $H_{0}:\mu=\mu_{0}$ reject
when $|\tau|>\func{cv}$, where $\func{cv}$ is the critical value, and the
corresponding confidence interval for $\mu$ has endpoints $\overline{y}\pm%
\func{cv}\hat{\sigma}/\sqrt{n}$. Inference methods in this class differ in
their choice of $\hat{\sigma}^{2}$ and critical value $\func{cv}$.

The case of regularly-spaced time series observations (panel (a) of Figure %
\ref{fig:Three-Time-Series Designs}) is the most well-studied version of
this problem. Here $\func{Var}(\sqrt{n}(\overline{y}-\mu ))$ is the long-run
variance of $y$. Classic choices for $\hat{\sigma}^{2}$ are kernel-based
consistent estimators such as those proposed in \cite{Newey87} and \cite%
{Andrews91}, and associated standard normal critical values. A more recent
literature initiated by \cite{Kiefer00} and \cite{Kiefer05} accounts for the
sampling uncertainty of kernel-based $\hat{\sigma}^{2}$ by considering
\textquotedblleft fixed-$b$\textquotedblright\ asymptotics where the
bandwidth is a fixed fraction of the sample size, which leads to a
corresponding upward adjustment of the critical value. Closely related are
projection estimators of $\hat{\sigma}^{2}$ where the number of projections
is treated as fixed in the asymptotics, as in M\"{u}ller (2004, 2007)\nocite%
{Muller04}\nocite{Muller07c}, \cite{Phillips05}, \cite{Sun13}, and others,
leading to Student-t critical values. These newer methods are found to
markedly improve size control under moderate serial correlation compared to
inference based on standard normal critical values.

In the general spatial case, the variance of $\overline{y}$ depends on the
correlation between all of the observations, and this in turn depends on two
distinct features of the problem. The first is the correlation between
observations at arbitrary locations (say $r$ and $s$); this is given by the
covariance function $\sigma_{u}(r-s)$. The second feature is which locations
in $\mathcal{S}$ are likely to be sampled; this is given by the spatial
density $g$. Only the first of these features is important in the
regularly-spaced time series example because the locations do not vary from
one application to the next.

Most existing suggestions for spatial inference are derived under the
assumption that the locations are (asymptotically) uniformly distributed,
corresponding to a constant density $g$: This includes the consistent
kernel-based estimator in \cite{Conley99}, the spatial analogue of the fixed-%
$b$ kernel approach analyzed in \cite{Bester_Conley_Hansen_Vogelsang_2016},
as well as the spatial projection-based estimator put forward in \cite%
{Sun_Kim_2012}. Exceptions include \cite{Kelejian2007}, who derive a
consistent kernel for $\hat{\sigma}^{2}$ under assumptions that can
accommodate arbitrary locations $s_{l}$, and the cluster approach suggested
by Ibragimov and M\"{u}ller (2010, 2015)\nocite{Ibragimov10}\nocite%
{Ibragimov15} and \cite{Bester11} (also see \cite{Cao2020}).

This paper makes progress over this literature by developing a method that
(i) accounts for sampling uncertainty in $\hat{\sigma}^{2}$ in a spatial
context while allowing for nonuniform spatial densities $g$; (ii) is valid
under generic weakly correlated $u_{l}$; (iii) also controls size under a
restricted but nonparametric form of strongly correlated $u_{l}$. The last
property sets it apart from all previously mentioned methods; in a time
series setting, \cite{Robinson05} and \cite{Muller14} derive inference under
parametric forms of strong dependence, and \cite{Dou_2019} derives optimal
inference under a non-parametric form of strong dependence under a
simplifying Whittle-type approximation to the implied covariance matrices.

Our method works as follows: First, a benchmark parametric model is
specified for the covariance function, say $\sigma_{u}^{0}(\cdot)=%
\sigma_{u}^{0}(\cdot|c)$, where $c$ is a persistence parameter with larger
values indicating less dependence. For a given lower bound on the
persistence parameter, say $c_{0}$, a hypothetical covariance matrix for $%
(y_{1},...,y_{n})^{\prime}$ is constructed using $\sigma_{u}^{0}(%
\cdot|c_{0}) $ evaluated at the actual sample locations $(s_{1},...,s_{n})$.
The eigenvectors of the demeaned version of this covariance matrix are the
(population) principal components of the residuals $\hat{u}_{l}=y_{l}-%
\overline{y}$ under $\sigma_{u}^{0}(\cdot|c_{0})$, and the sample variance
of $q$ of these principal components is the estimator $\hat{\sigma}^{2}$.
The critical value is chosen to ensure coverage for all $c\geq c_{0}$. The
number of principal components $q$ is chosen to minimize the expect length
of the confidence interval in the model where $u_{l}$ is i.i.d.~For
shorthand, we refer to the method as \emph{spatial correlation principal
components}, abbreviated SCPC.

Intuitively, variance estimators $\hat{\sigma}^{2}$ that are quadratic forms
in $\hat{u}$ are sums of squares of weighted averages of $\hat{u}$. Under
spatial correlation, most weighted averages are less variable than $%
\overline{y}$, leading to a downward biased $\hat{\sigma}^{2}$. SCPC selects
the linear combinations of $\hat{u}$ that are most variable, so that the
bias is as small as possible in the benchmark model with parameter $c_{0}$.

The remainder of the paper studies this method. Section 2 provides the
specifics for SCPC. These specifics raise a variety of issues that are the
focus of the remaining sections of the paper. In particular, Section 3 lays
out the analytic framework used to study the large-sample and finite-sample
Gaussian properties of spatial t-statistics. We use the framework to analyze
SCPC, but several of the results in Section 3 encompass other methods,
notably \textquotedblleft fixed-$b$\textquotedblright\ kernel-based methods,
and general projection estimators with a fixed number of basis functions. We
find that in contrast to the regularly spaced time series case, such
t-statistics with analogously adjusted critical values are \emph{not}
generically valid under weak correlation as soon as the spatial density
function is not uniform. We develop an alternative approach to the
construction of critical values that restores validity, and this is used for
SCPC inference. Section 4 thus shows that SCPC has the desired large-sample
coverage probability under generic weak correlation. Moreover, Section 4
provides a set of (easily verifiable) sufficient conditions that guarantee
coverage under arbitrary mixtures of a set of strong correlation patterns in
a finite-sample Gaussian setting. Section 4 also investigates the
finite-sample coverage probability of SCPC confidence sets when there is
heteroskedasticity across locations or measurement errors in locations ---
two problems faced in some applications. Section 5 addresses the question of
efficiency of SCPC by computing a lower bound on the expected length of
confidence intervals for \emph{any }inference method that controls coverage
in a particular class of spatial correlations. Comparing the expected length
of SCPC to this lower bound provides a measure of the efficiency of the
method. Section 6 compares the properties of SCPC to other methods that have
been proposed in the literature, and the results suggest that SCPC dominates
these methods over the range of covariance functions and spatial designs
considered. Section 7 discusses extensions and implementation issues. First,
it discusses how the results developed in the body of the paper for
inference about the population mean can be applied to inference problems
about regression coefficients or parameters in GMM models. It then discusses
two important computational issues involved in computing the critical value
and computing the required eigenvectors for the construction of SCPC in very
large-$n$ applications. Finally, Section 7 provides a sketch of the
generalization of the SCPC method to multivariate (F-test) settings. Proofs
are collected in the appendix.

\section{\label{sec:SCPC}Spatial Correlation Principal Components}

This section provides details for computing the SCPC t-statistic, critical
value and associated confidence interval. The construction of SCPC raises a
variety of questions about its properties, many of which are posed here and
discussed in detail in the remaining sections of the paper.

The construction of the SCPC t-test and confidence interval involves, among
other things, various covariance matrices and probability calculations. We
stress at the outset that these are used to describe the required
calculations, and they are not assumptions about the probability
distribution of the data under study. Those assumptions will be listed in
Section 3 and, it will turn out, are significantly more general than what
would follow from the description in this section.

Let $\mathbf{y}=(y_{1},y_{2},...,y_{n})^{\prime }$ and similarly for $%
\mathbf{s}=(s_{1},s_{2},...,s_{n})^{\prime },$ $\mathbf{u}%
=(u_{1},u_{2},...,u_{n})^{\prime }$ and the vector of residuals $\mathbf{%
\hat{u}}=(\hat{u}_{1},\hat{u}_{2},...,\hat{u}_{n})^{\prime }$. Let $\mathbf{l%
}$ denote an $n\times 1$ vector of $1$s, and $\mathbf{M}=\mathbf{I}-\mathbf{l%
}(\mathbf{l}^{\prime }\mathbf{l})^{-1}\mathbf{l}^{\prime }$. Consider a
benchmark model for $u_{l}$ with a parametric covariance function $\func{Cov}%
(u(r),u(s))=\sigma _{u}^{0}(r-s|c)$, where smaller values of the scalar
parameter $c$ indicate stronger correlations. In the following, we focus on
the simple Gaussian exponential (`AR(1)') model where $\sigma
_{u}^{0}(r-s|c)=\exp (-c||r-s||)$ for $c>0$. Let $\mathbf{\Sigma }(c)$
denote the $n\times n$ covariance matrix with $\mathbf{\Sigma }(c)_{ij}=\exp
(-c||s_{i}-s_{j}||)$, so that $\mathbf{\Sigma }(c)$ is the covariance matrix
of $u(s)$ evaluated at the sample locations $\mathbf{s}.$ Let $c_{0}$ denote
a pre-determined value of $c$ that is meant to capture an upper bound on the
spatial persistence in the data. (The choice of $c_{0}$ is discussed below).
Let $\mathbf{r}_{1},\mathbf{r}_{2},...,\mathbf{r}_{n}$ denote the
eigenvectors of $\mathbf{M}\mathbf{\Sigma }(c_{0})\mathbf{M}$ corresponding
to the eigenvalues ordered from largest to smallest, and normalized so that $%
n^{-1}\mathbf{r}_{j}^{\prime }\mathbf{r}_{j}=1$ for all $j$. The scalar
variable $n^{-1/2}\mathbf{r}_{j}^{\prime }\mathbf{\hat{u}}$ has the
interpretation as the $j$th population principle component of \textbf{$%
\mathbf{\hat{u}|}$}$\mathbf{s}$ $\sim \mathcal{N}(\mathbf{0},\mathbf{M}%
\mathbf{\Sigma }(c_{0})\mathbf{M)}$. The SCPC estimator of $\sigma ^{2}$
based on the first $q$ of these principal components is 
\begin{equation}
\hat{\sigma}_{\text{SCPC}}^{2}(q)=q^{-1}\sum_{j=1}^{q}(n^{-1/2}\mathbf{r}%
_{j}^{\prime }\mathbf{\hat{u}})^{2},  \label{eq:sigma_hat(q) for SCPC}
\end{equation}%
and the corresponding SCPC t-statistic is 
\begin{equation}
\tau _{\text{SCPC}}(q)=\frac{\sqrt{n}(\overline{y}-\mu _{0})}{\hat{\sigma}_{%
\text{SCPC}}(q)}.  \label{eq:tau(q) for SCPC}
\end{equation}

The critical value $\func{cv}_{\text{SCPC}}(q)$ of the level-$\alpha$ SCPC
test is chosen so that size is equal to $\alpha$ under the Gaussian
benchmark model with $c\geq c_{0}$. That is, $\func{cv}_{\text{SCPC}}(q)$
satisfies 
\begin{equation}
\sup_{c\geq c_{0}}\mathbb{P}_{\mathbf{\Sigma}(c)}^{0}(|\tau_{\text{SCPC}%
}(q)|>\func{cv}_{\text{SCPC}}(q)|\mathbf{s})=\alpha,  \label{eq:SCPC cvq}
\end{equation}
where $\mathbb{P}_{\mathbf{\Sigma}(c)}^{0}$ means that the probability is
computed in the benchmark model $\mathbf{y|s}\sim\mathcal{N}(\mu_{0}\mathbf{l%
},\mathbf{\Sigma}(c))$.

The final ingredient in the method is the choice of $q$. Let $\mathbb{E}%
^{1}[2\hat{\sigma}_{\text{SCPC}}(q)\func{cv}_{\text{SCPC}}(q)|\mathbf{s]}$
denote the expected length of the confidence interval constructed using $%
\tau _{\text{SCPC}}(q)$ under the Gaussian i.i.d.~model $\mathbf{y|s}\sim 
\mathcal{N}(\mathbf{l}\mu ,\mathbf{I})$. (The superscript \textquotedblleft
1\textquotedblright\ on $\mathbb{E}$ differentiates this from the benchmark
model with covariance matrix $\mathbf{\Sigma }(c)$.) SCPC chooses $q_{\func{%
SCPC}_{\text{{}}}}$ to make this length as small as possible, that is $q$
solves 
\begin{equation}
\min_{q\geq 1}\mathbb{E}^{1}[2\hat{\sigma}_{\text{SCPC}}(q)\func{cv}_{\text{%
SCPC}}(q)|\mathbf{s]=}\min_{q\geq 1}\sqrt{8}n^{-1/2}q^{-1/2}\func{cv}_{\text{%
SCPC}}(q)\frac{\Gamma ((q+1)/2)}{\Gamma (q/2)}  \label{eq:SCPC Length}
\end{equation}%
with the equality exploiting that $q\hat{\sigma}_{\text{SCPC}}^{2}(q)\mathbf{%
|s}\sim \chi _{q}^{2}$ in the Gaussian i.i.d.~model.

\begin{remark}
The primary concern in the construction of $\hat{\sigma}^{2}$ is downward
bias. Recall that the eigenvector $\mathbf{r}_{1}$ maximizes $\mathbf{h}%
^{\prime}\mathbf{M}\mathbf{\Sigma}(c_{0})\mathbf{Mh}$ among all vectors $%
\mathbf{h}$ of the same length, the second eigenvector $\mathbf{r}_{2}$
maximizes $\mathbf{h}^{\prime}\mathbf{M}\mathbf{\Sigma}(c_{0})\mathbf{Mh}$
subject to $\mathbf{h}^{\prime}\mathbf{r}_{1}=0$, and so forth, and for any $%
q\geq1$, the $n\times q$ matrix $(\mathbf{r}_{1},\ldots,\mathbf{r}_{q})$
maximizes $\limfunc{tr}\mathbf{H}^{\prime}\mathbf{M}\mathbf{\Sigma}(c_{0})%
\mathbf{MH}$ among all $n\times q$ matrices $\mathbf{H}$ with $n^{-1}\mathbf{%
H}^{\prime}\mathbf{H=I}_{q}$. Thus, the SCPC method selects the linear
combinations of $\mathbf{\hat{u}}$ in the estimator of $\sigma^{2}$ that
induce the smallest bias in the benchmark model with $c=c_{0}$, under the
constraint of being unbiased in the i.i.d.~model.

The choice of $q$ trades off the downward bias in $\hat{\sigma}_{\text{SCPC}%
}^{2}(q)$ that occurs when $q$ is large and its large variance when $q$ is
small. Both bias and variance lead to a large critical value, and (\ref%
{eq:SCPC Length}) leads to a choice of $q$ that optimally trades off these
two effects to obtain the shortest possible expected confidence interval
length in the i.i.d.~model.
\end{remark}

\begin{remark}
By construction, SCPC confidence intervals have correct coverage in Gaussian
models with a spatial exponential covariance function (`AR(1)' models) with
spatial persistence level less than or equal to the model with $c=c_{0}$.
Lemma \ref{lm:NormalLimit_q} in Section \ref{sec:An-Analytic-Framework}
provides a central limit result that rationalizes the normality assumption.
Theorem \ref{thm:cv_n} provides conditions on the choice of $c_{0}$ so that
the SCPC t-test controls size in large samples not just in the exponential
model, but under generic `weak correlation', as defined in Section \ref%
{sec:An-Analytic-Framework}. Theorem \ref{thm:robust} provides easily
verifiable sufficient conditions for size control under mixtures of
parametric small sample Gaussian models.
\end{remark}

\begin{remark}
SCPC requires that the researcher chooses a value for $c_{0}$ which
represents the highest degree of spatial correlation allowed by the method.
One way to calibrate $c_{0}$ is via the average pairwise correlation of the
spatial observations 
\begin{equation*}
\overline{\rho}=\frac{1}{n(n-1)}\sum_{l=1}^{n}\sum_{\ell\neq l}\func{Cor}%
\left(y_{l},y_{\ell}\left\vert \mathbf{s}_{n}\right.\right)
\end{equation*}
that is, we set $c_{0}$ so that it implies a given value $\overline{\rho}%
_{0} $ of $\bar{\rho}$. For example, $\overline{\rho}_{0}=0.001$ implies
very weak correlation, $\overline{\rho}_{0}=0.02$ stronger correlation, and $%
\overline{\rho}_{0}=0.10$ very strong correlation. In our examples, we
calibrate $c_{0}$ to these three values of $\overline{\rho}$.
\end{remark}

\begin{remark}
The SCPC method with $c_{0}$ calibrated in this way is invariant to the
scale of the locations $\{s_{l}\}_{l=1}^{n}\mapsto \{as_{l}\}_{l=1}^{n}$ for 
$a>0$, and (in contrast to Sun and Kim's (2012) suggestion) also to
arbitrary distance preserving transformations, such as rotations.
\end{remark}

\begin{figure}[tbp]
\centering{}\includegraphics[scale=0.72]{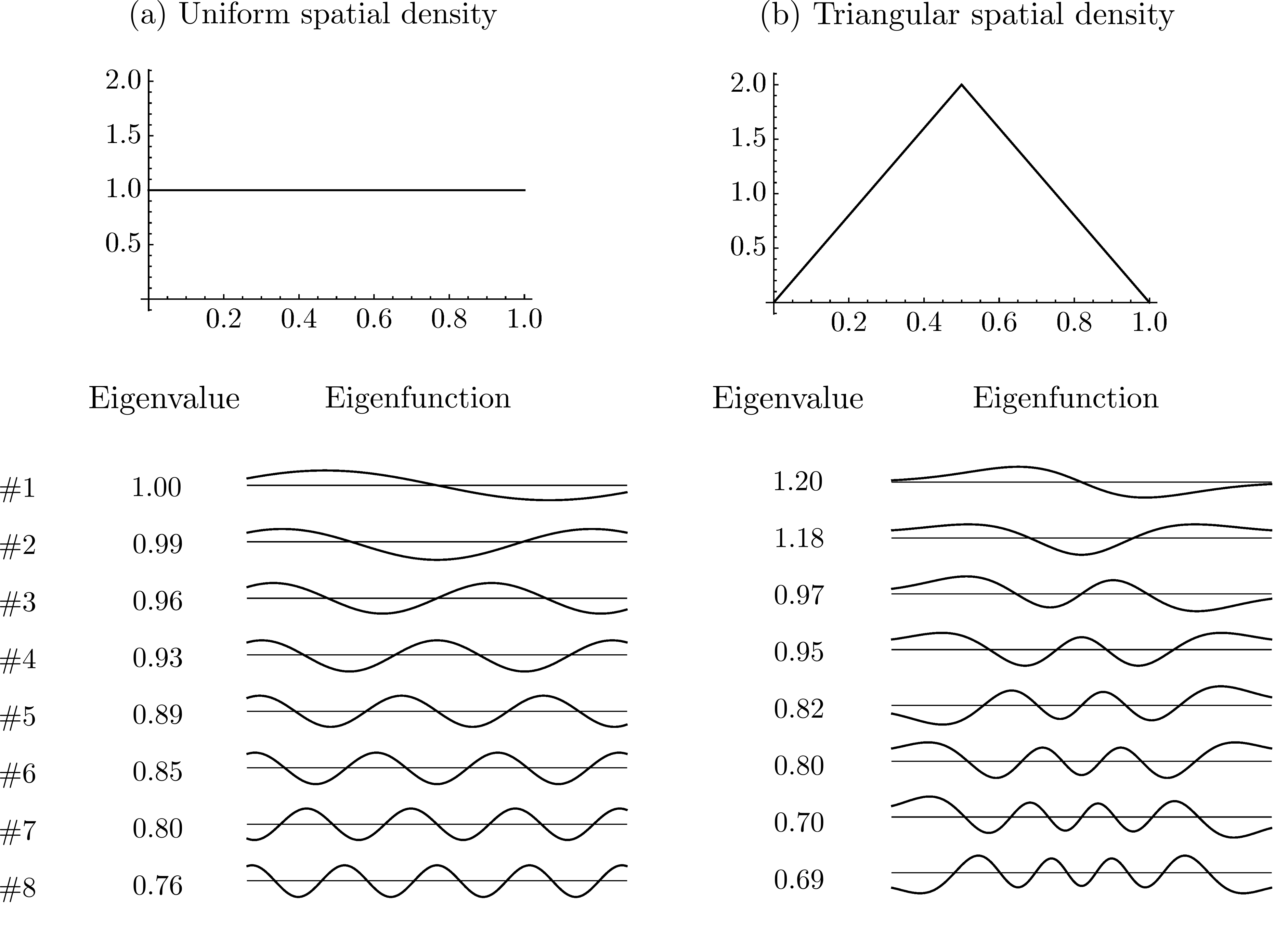}
\caption{Eigenfunctions for Two One-Dimensional Spatial Designs}
\label{fig:Eigenfunctions-for-time series}
\end{figure}

\begin{figure}[tbp]
\centering{}\includegraphics[scale=0.7]{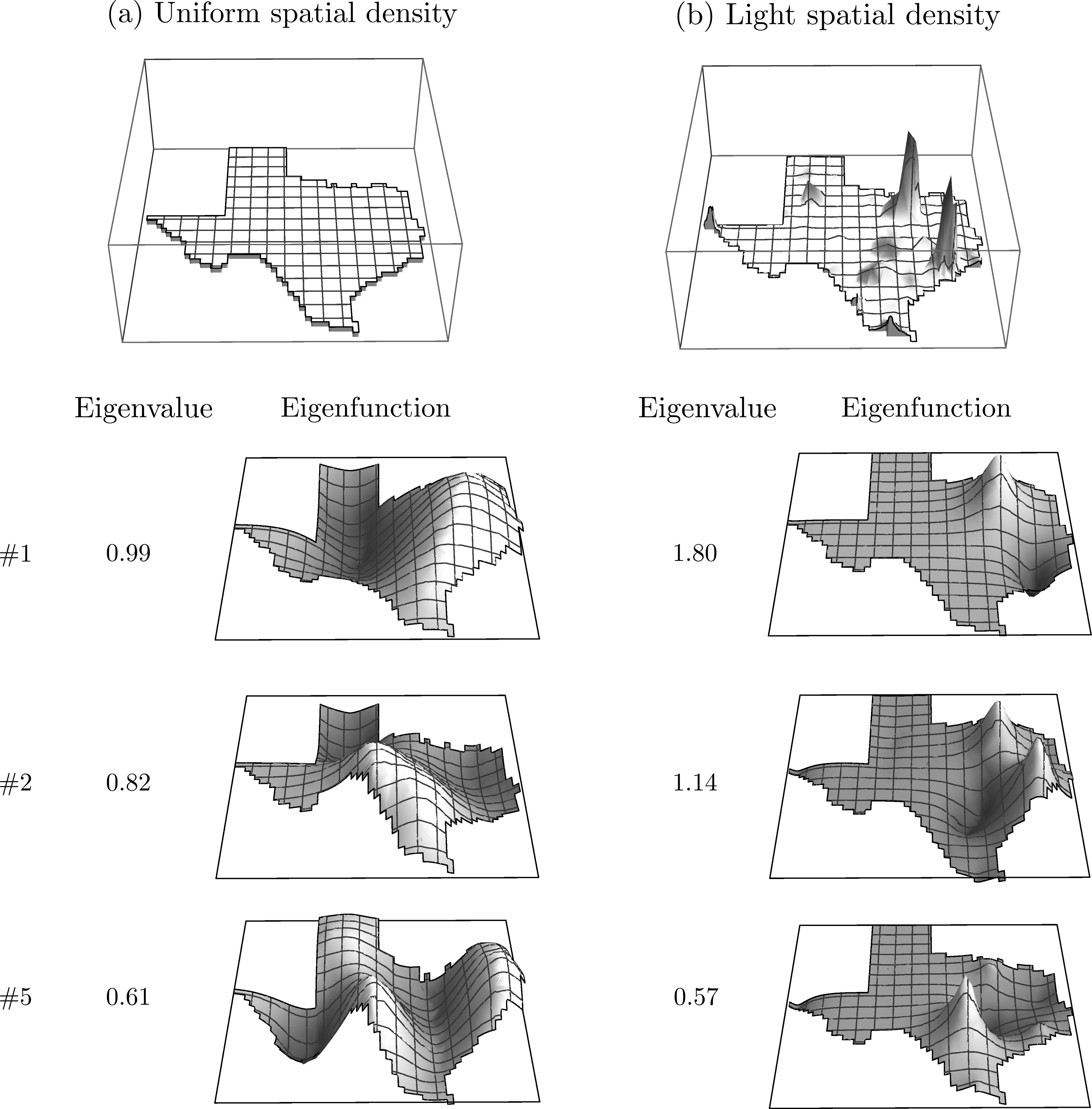}
\caption{Eigenfunctions for Two Geographic Spatial Designs}
\label{fig:Eigenfunctions-for-texas}
\end{figure}

\begin{remark}
The weights $\mathbf{r}_{j}$ used to construct the principal components and $%
\hat{\sigma}_{\text{SCPC}}^{2}(q)$ depend on $\mathbf{s}$, the sample values
of the spatial locations. Because the spatial locations are randomly drawn,
the $\mathbf{r}_{j}$ weights are random. But as shown in Section \ref%
{sec:An-Analytic-Framework}, the weights have well-defined limits in terms
of appropriately defined nonrandom eigenfunctions. Figure \ref%
{fig:Eigenfunctions-for-time series} plots selected eigenfunctions for two
one-dimensional spatial designs and Figure \ref{fig:Eigenfunctions-for-texas}
shows the associated plots for the Texas example, where in both cases $%
\overline{\rho }_{0}=0.02$. With uniform spatial densities (panel (a) in
both figures), the eigenfunctions are much like the weighting functions used
for low-frequency projection methods for regularly spaced time series (e.g., 
\cite{Muller04}, \cite{Phillips05}, \cite{Sun13}) or its spatial analogue
(e.g., \cite{Sun_Kim_2012}). In contrast, the non-uniform densities (panel
(b) in the figures) produce weights that are distorted versions of their
uniform counterparts, with most of the variation concentrated in
high-density areas.

The figures also show the associated normalized eigenvalues, that is the
variance of the principal components under the assumed exponential model,
relative to the variance of $\sqrt{n}(\overline{y}-\mu)$. When the density
is uniform, these relative variances are slightly below $1.0$ for small $j$,
and decline monotonically with $j$. This leads to the familiar downward bias
of $\hat{\sigma}^{2}$ in projection methods. When the spatial density is not
uniform, the relative variance of the principal components can be larger
than unity, mitigating this downward bias.
\end{remark}

\begin{remark}
In the regular spaced time series case, the eigenvectors of SCPC for $\bar{%
\rho}_{0}\in\{0.02,0.10\}$ are numerically close to the type-II\ cosine
transforms considered in M\"{u}ller (2004, 2007), \cite%
{Lazarus_etal_JBES_2018} and \cite{Dou_2019}. What is more, the SCPC choice
of $q$ is also numerically close to the corresponding optimal choice of $q$
in \cite{Dou_2019}. So when applied to time series, SCPC comes close to
replicating Dou's (2019) suggestion for optimal inference, with $c_{0}$
representing the upper bound for the degree of persistence. The same is true
in a spatial design with locations that happen to fall on a line with
approximately uniform empirical distribution.
\end{remark}

\noindent \textbf{\textit{U.S.~states spatial correlation designs}}. Before
making two additional remarks about the SCPC\ method, we introduce a set of
spatial correlation designs that will be used throughout the paper. The idea
is to consider a set of real world designs to learn about the usefulness of
the SCPC and other methods in practice. In particular, we randomly draw $%
n=500$ locations within the boundaries of the 48 contiguous states of the
U.S.~(we also considered $n=1000$ draws, and found nearly identical results
in all exercises). The density of locations $g$ within each state is either
uniform ($g_{\text{uniform}}$), or it is proportional to light measured from
space ($g_{\text{light}}$) as a proxy for economic activity. We draw five
sets of 500 independent locations under each density $g\in\{g_{\text{uniform}%
},g_{\text{light}}\}$ and $\bar{\rho}_{0}\in\{0.02,0.10\}$ for each state,
for a total of 240 (= 48 states $\times$ 5 location draws) sets of locations 
$\{s_{l}\}_{l=1}^{500}$ and associated covariances under each of the four $%
(g,\bar{\rho}_{0})$ pairs.

\begin{remark}
The critical value of the SCPC t-statistic reflects randomness in both $%
\overline{y}$ and $\hat{\sigma}_{\text{SCPC}}^{2}$. This is analogous to
inference in small-sample Gaussian models using critical values from the
Student-t distribution. Figure \ref{fig:Expected-length-relative to oracle}
shows the effect of the uncertainty in $\sigma ^{2}$ on the expected length
of $95\%$ confidence intervals in the U.S.~states spatial correlation
designs, by comparing the expected length of the SCPC confidence interval in
the i.i.d.~model to the the length with $\sigma ^{2}$ known: this relative
length is $\mathbb{E}^{1}[(\func{cv}/1.96)(\hat{\sigma}_{\text{SCPC}}/\sigma
)|\mathbf{s]}$, where 1.96 is the standard normal critical value. The figure
plots the CDF of these relative lengths over the 240 draws under each $(g,%
\bar{\rho}_{0})$ pair. For example, the left-most CDF (dashed blue, for $%
g=g_{\text{light}}$ and $\overline{\rho }_{0}=0.02$) shows that the relative
expected length ranges from roughly 1.08 to 1.18 across the 240 draws. The
figure indicates that the expected lengths are higher under $g_{\text{uniform%
}}$ than under the $g_{\text{light}}$ design and are higher under $\bar{\rho}%
_{0}=0.10$ than $\bar{\rho}_{0}=0.02.$ For comparison the figure also shows
the relative expected lengths of Student-t confidence intervals with $8$ and 
$3$ degrees of freedom, in multiples of the length of the known variance
z-interval. Evidently, when $\bar{\rho}_{0}=0.02$, the increase in expected
length of the SCPC confidence interval relative to an oracle endowed with
the value of $\sigma ^{2}$ is roughly like learning about the value of $%
\sigma ^{2}$ from $8$ i.i.d.$\mathcal{N}(0,\sigma ^{2})$ observations. When $%
\bar{\rho}_{0}=0.10$, relative lengths increase to approximately what would
obtain from Student-$t_{3}$ inference.
\end{remark}

\begin{figure}[tbp]
\centering{}\includegraphics[scale=0.7]{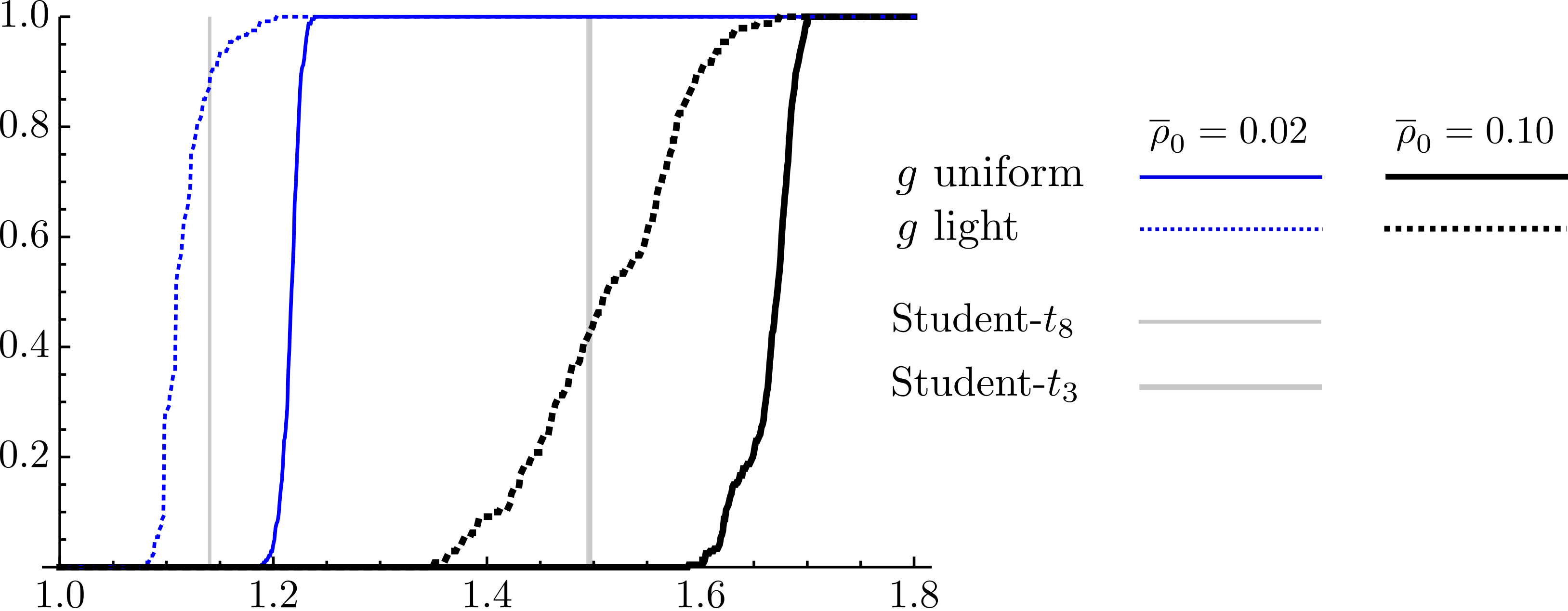}
\caption{CDFs of Expected Length of SCPC Confidence Interval Relative to
Known Variance Interval}
\label{fig:Expected-length-relative to oracle}
\end{figure}

\begin{remark}
Consider the related question about the efficiency of SCPC relative to other
methods that do not assume that the value of $\sigma^{2}$ is known. This
question can be answered in two ways. The first is to compare SCPC to
methods that have previously been proposed. This is done in Section \ref%
{sec:Comparison-with-Other}. A more ambitious approach compares SCPC to the
most efficient method constructed for any particular spatial density that,
like SCPC, produces confidence intervals with the desired coverage over a
wide range of covariance functions. This is done in Section \ref%
{sec:Efficiency-of-SCPC} which computes a lower bound on the expected length
of confidence intervals for any such method.
\end{remark}

\section{\label{sec:An-Analytic-Framework}Large-sample analysis of spatial
t-statistics}

This section outlines a large-sample framework used to study SCPC and other
spatial t-statistics. The first two subsections introduce notation and the
asymptotic sampling framework. With these in hand, the remainder of the
section summarizes the large-sample distribution of various statistics
including the SCPC and kernel-based t-statistics. Proofs are provided in the
appendix.

\subsection{Notation}

Some of this notation has been introduced earlier, but is repeated here for
easy reference.

The sample mean is denoted by $\overline{y}_{n}$, where here and elsewhere
we append the subscript $n$ for clarity in the asymptotic analysis. The
residual is $\hat{u}_{l}=y_{l}-\overline{y}_{n}$. Let $\mathbf{y}%
_{n}=(y_{1},...,y_{n})^{\prime}$, and similarly for $\mathbf{u}_{n}$, $\hat{%
\mathbf{u}}_{n}$ and $\mathbf{s}_{n}$. The vector $\mathbf{l}_{n}$ is a $%
n\times1$ vector of 1s, and $\mathbf{M}_{n}=\mathbf{I}_{n}-\mathbf{l}_{n}(%
\mathbf{l}_{n}^{\prime}\mathbf{l}_{n})^{-1}\mathbf{l}_{n}$, so that $\mathbf{%
\hat{u}}_{n}=\mathbf{M}_{n}\mathbf{u}_{n}$.

Generically, we consider estimators $\hat{\sigma}_{n}^{2}$ that are
quadratic forms in $\hat{\mathbf{u}}_{n}$. Let $\mathbf{Q}_{n}$ be a
positive semidefinite matrix with $\mathbf{Q}_{n}\mathbf{l}_{n}=\mathbf{0}$.
We consider estimators of the form 
\begin{equation}
\hat{\sigma}_{n}^{2}(\mathbf{Q}_{n})=n^{-1}\hat{\mathbf{u}}_{n}^{\prime}%
\mathbf{Q}_{n}\hat{\mathbf{u}}_{n}=n^{-1}\mathbf{u}_{n}^{\prime}\mathbf{Q}%
_{n}\mathbf{u}_{n}  \label{eq:sigma_hat(Q)}
\end{equation}
where the final equality follows from $\mathbf{Q}_{n}\mathbf{l}_{n}=\mathbf{0%
}$.

Two leading examples of estimators in this class are kernel-based estimators
and orthogonal-projections estimators. For kernel-based estimators, let $%
k(r,s)$ denote a positive semi-definite kernel, $k:\mathcal{S}\times 
\mathcal{S}\mapsto \mathbb{R}$. Let $\mathbf{K}_{n}$ denote an $n\times n$
matrix with $(l,\ell )$ element equal to $k(s_{l},s_{\ell })$ and let $%
\mathbf{Q}_{n}=\mathbf{M}_{n}\mathbf{K}_{n}\mathbf{M}_{n}$. Then $\hat{\sigma%
}_{n}^{2}=n^{-1}\sum_{l}\sum_{\ell }k(s_{l},s_{\ell })\hat{u}_{l}\hat{u}%
_{\ell }=$ $n^{-1}\hat{\mathbf{u}}_{n}^{\prime }\mathbf{Q}_{n}\hat{\mathbf{u}%
}_{n}$. For orthogonal-projection estimators, let $\mathbf{\hat{W}}_{n}$ be
an $n\times q$ matrix with $j$th column given by $\mathbf{\hat{w}}_{j}$
satisfying $n^{-1}\mathbf{\hat{W}}_{n}^{\prime }\mathbf{\hat{W}}_{n}=q^{-1}%
\mathbf{I}_{q}$ and $\mathbf{\hat{W}}_{n}^{\prime }\mathbf{l}_{n}=\mathbf{0}$
(the `hat' notation is a reminder that $\mathbf{\hat{W}}$ depends on the
locations $\mathbf{s}_{n}$, which are random). With $\mathbf{Q}_{n}=\mathbf{%
\hat{W}}_{n}\mathbf{\hat{W}}_{n}^{\prime }$, the orthogonal projection
estimator is $\hat{\sigma}_{n}^{2}=\sum_{j=1}^{q}(n^{-1/2}\mathbf{\hat{w}}%
_{j}^{\prime }\hat{\mathbf{u}}_{n})^{2}=n^{-1}\hat{\mathbf{u}}_{n}^{\prime }%
\mathbf{Q}_{n}\hat{\mathbf{u}}_{n}$. The SCPC estimator is an
orthogonal-projection estimator using the first $q$ eigenvectors of $\mathbf{%
M}_{n}\mathbf{\Sigma }(c_{0})\mathbf{M}_{n}$, scaled to have length $1/\sqrt{%
q},$ as the columns of $\mathbf{\hat{W}}_{n}$.

For quadratic form estimators $\hat{\sigma}_{n}^{2}(\mathbf{Q}_{n})$, under
the null hypothesis the squared t-statistic is a ratio of quadratic forms in 
$\mathbf{u}_{n}$ 
\begin{equation}
\tau_{n}^{2}(\mathbf{Q}_{n})=\frac{\left(\sqrt{n}(\overline{y}%
_{n}-\mu_{0})\right)^{2}}{\hat{\sigma}_{n}^{2}(\mathbf{Q}_{n})}=\frac{%
\mathbf{u}_{n}^{\prime}\mathbf{l}_{n}\mathbf{l}_{n}^{\prime}\mathbf{u}_{n}}{%
\mathbf{u}_{n}^{\prime}\mathbf{Q}_{n}\mathbf{u}_{n}}.  \label{eq:tau_2(Q)}
\end{equation}

\subsection{\label{subsec:Sampling-and-large-}Sampling and large-$n$
framework}

The spatial locations $s$ are chosen from $\mathcal{S}$, a compact subset of 
$\mathbb{R}^{d}$. Sample locations are selected as i.i.d.~draws from a
distribution $G$ with density $g$, where $g(s)$ is continuous and positive
for all $s\in\mathcal{S}$.

The average pairwise correlation of $y$, conditional on the sample locations
is $\bar{\rho}_{n}=\frac{1}{n(n-1)}\sum_{l=1}^{n}\sum_{\ell\neq l}\func{Cor}%
\left(y_{l},y_{\ell}\left\vert \mathbf{s}_{n}\right.\right)$. When $%
\overline{\rho}_{n}=0$, $\mathbf{y}_{n}\left\vert \mathbf{s}_{n}\right.$ is
white noise. When $\overline{\rho}_{n}=O_{p}(1)$ (and not $o_{p}(1)$), we
will say the process exhibits \emph{strong} \emph{correlation}. When $%
\overline{\rho}_{n}=O_{p}(1/c_{n}^{d})$ where $c_{n}$ is a sequence of
constants with $c_{n}\rightarrow\infty$, we follow \cite{Lahiri_2003} and
say the process exhibits \emph{weak correlation}.

The following asymptotic framework, adapted from \cite{Lahiri_2003}, is
useful for representing weak and strong correlation. Let $B$ be a zero-mean
stationary random field on $\mathbb{R}^{d}$ with continuous covariance
function $\mathbb{E}[B(s)B(r)]=\sigma_{B}\left(s-r\right)$, and $B$ and $%
\{s_{l}\}_{l=1}^{n}$ are independent. To avoid pathological cases, we
further assume $\int\sigma_{B}(s)ds>0$ and that $B$ is nonsingular in the
sense that $\inf_{||f||=1}\int\int f(r)f(s)\sigma_{B}(s-r)dG(r)dG(s)>0$ with 
$||f||^{2}=\int f^{2}(s)dG(s)$.

Let $c_{n}$ denote a sequence of constants with either $c_{n}\rightarrow
\infty $ or $c_{n}=c>0$. We consider a triangular-array framework with $%
u_{l}=B(c_{n}s_{l})$ for $s_{l}\in \mathcal{S}$, so that $\sigma
_{u}(s)=\sigma _{B}(c_{n}s)$. The sequence $c_{n}$ determines the `infill'
and `outfill' nature of the asymptotics. To see this, note that the volume
of the relevant domain for the random field $B$ is $c_{n}^{d}\func{vol}(%
\mathcal{S)}$, where $\func{vol}(\mathcal{S)}$ is the volume of $\mathcal{S}%
. $ The average number of sample points per unit of volume is then $%
n/(c_{n}^{d}\func{vol}(\mathcal{S)}).$ If $c_{n}^{d}\propto n,$ the volume
of the domain is increasing, while the number of points per unit of volume
is not; this is the usual outfill asymptotic sampling scheme. On the other
hand, when $c_{n}=c$, a constant, the volume of the domain is fixed, and the
number of points per unit of volume is proportional to $n$; this is the
usual infill sampling. Finally, when $c_{n}\rightarrow \infty $ with $%
c_{n}^{d}=o(n)$ the sampling scheme features both infill and outfill
asymptotics. A calculation shows that $\overline{\rho }%
_{n}=O_{p}(1/c_{n}^{d})$, so the sequence $c_{n}$ characterizes weak and
strong correlation as described above. With this background, let $%
a_{n}=c_{n}^{d}/n$; we will assume that $a_{n}\rightarrow a\in \lbrack
0,\infty )$.

Finally, we specify a set of weighting functions. To simplify the problem,
we initially consider weights that are nonrandom. For $j=1,\ldots ,q$, let $%
w_{j}:\mathcal{S}\mapsto 
\mathbb{R}
$ denote a set of continuous functions that satisfy $\int w_{j}(s)dG(s)=0$
and $\int w_{j}^{2}(s)dG(s)>0$. We introduce the following notation
involving these functions: $\mathbf{w}(s)$ is a $q\times 1$ vector-valued
continuous function with $\mathbf{w}(s)=(w_{1}(s),...,w_{q}(s))^{\prime }$; $%
\mathbf{w}^{0}(s)=(1,\mathbf{w}(s)^{\prime })^{\prime }$; $\mathbf{W}_{n}$
is a $n\times q$ matrix with $l$th row given by $\mathbf{w}(s_{l})^{\prime }$%
, and $\mathbf{W}_{n}^{0}$ is a $n\times (q+1)$ matrix with $l$th row given
by $\mathbf{w}^{0}(s_{l})^{\prime }$ so that $\mathbf{W}_{n}^{0}=[\mathbf{l}%
_{n},\mathbf{W}_{n}]$.

\begin{remark}
In our framework, locations $s_{l}$ are sampled within $\mathcal{S}$ for a
fixed and given $\mathcal{S}$. But nothing changes in our derivations if
instead we treated the observations $y_{l}$ as being indexed by $%
c_{n}s_{l}\in c_{n}\mathcal{S}$, as in \cite{Lahiri_2003}, or any other
one-to-one transformation of $s_{l}$. The essential characteristic is the
dependence pattern over the spatial domain of the observations, governed by $%
c_{n}$ and $B$.
\end{remark}

With this background, we now present the large-sample analysis.

\subsection{\label{sec:weighted_averages}Large-sample behavior of weighted
averages}

As is evident from equation (\ref{eq:tau_2(Q)}) the squared t-statistic is a
ratio of squares of weighted average of the elements of $\mathbf{u}_{n}$.
This subsection discusses the large-sample distribution of such weighted
averages. These results involve weak converge (i.e., convergence in
distribution) where our interest lies in these limits conditional on the
locations $\mathbf{s}_{n}$. With this in mind, for $\mathbf{X}_{n}$ and $%
\mathbf{X}$ $p$-dimensional random vectors, we use the notation $\mathbf{X}%
_{n}|\mathbf{s}_{n}\Rightarrow _{p}\mathbf{X}$ to denote $\mathbb{E}[h(%
\mathbf{X}_{n})|\mathbf{s}_{n}]\overset{p}{\rightarrow }\mathbb{E}[h(\mathbf{%
X})]$ for any bounded continuous function $h:\mathbb{R}^{p}\mapsto \mathbb{R}
$. This notion of weak convergence in probability is slightly weaker than
almost sure weak convergence of conditional distributions, but still ensures
that the limiting distribution is not induced by the randomness in the
locations $\mathbf{s}_{n}$.

Lemma \ref{lm:NormalLimit_q} characterizes the large-sample behavior of sums
of the form $\sum_{l=1}^{n}\mathbf{w}^{0}(s_{l})u(s_{l}).$ For the weak
correlation result, we invoke the mixing and moment assumptions of \cite%
{Lahiri_2003} on $B$ that underlie his Theorem 3.2.

\begin{lemma}
\label{lm:NormalLimit_q}(i) (strong correlation) Suppose $c_{n}=c>0$ and $B$
is a Gaussian process. Then 
\begin{equation*}
n^{-1}\mathbf{W}_{n}^{0\prime }\mathbf{u}_{n}|\mathbf{s}_{n}\Rightarrow _{p}%
\mathbf{X}\sim \mathcal{N}(0,\mathbf{\Omega }_{sc})
\end{equation*}%
with 
\begin{equation*}
\mathbf{\Omega }_{sc}=\int \int \mathbf{w}^{0}(r)\mathbf{w}^{0}(s)^{\prime
}\sigma _{B}(c(r-s))dG(r)dG(s).
\end{equation*}

(ii) (weak correlation) Suppose $c_{n}\rightarrow \infty $, and the
assumptions of Lahiri's (2003) Theorem 3.2 hold. Then 
\begin{equation*}
a_{n}^{1/2}n^{-1/2}\mathbf{W}_{n}^{0\prime }\mathbf{u}_{n}|\mathbf{s}%
_{n}\Rightarrow _{p}\mathbf{X}\sim \mathcal{N}(0,\mathbf{\Omega }_{wc})
\end{equation*}%
with 
\begin{equation*}
\mathbf{\Omega }_{wc}=a\sigma _{B}(0)\mathbf{V}_{1}+\left( \int \sigma
_{B}(s)ds\right) \mathbf{V}_{2}
\end{equation*}%
where 
\begin{equation*}
\mathbf{V}_{1}=\int \mathbf{w}^{0}(s)\mathbf{w}^{0}(s)^{\prime }g(s)ds\text{
and $\mathbf{V}_{2}=\int \mathbf{w}^{0}(s)\mathbf{w}^{0}(s)^{\prime
}g(s)^{2}ds.$}
\end{equation*}
\end{lemma}

\begin{remark}
Note that the variance of $\sum_{l=1}^{n}\mathbf{w}^{0}(s_{l})u(s_{l})$
conditional on $\mathbf{s}_{n}$ is 
\begin{align}
\limfunc{Var}\left[\sum_{l=1}^{n}\mathbf{w}^{0}(s_{l})u(s_{l})\left\vert 
\mathbf{s}_{n}\right.\right] & =\sum_{l}\sum_{\ell}\mathbf{w}^{0}(s_{l})%
\mathbf{w}^{0}(s_{\ell})^{\prime}\sigma_{u}(s_{l}-s_{\ell})  \notag \\
& =\sum_{l}\sum_{\ell}\mathbf{w}^{0}(s_{l})\mathbf{w}^{0}(s_{\ell})^{\prime}%
\sigma_{B}(c_{n}\left(s_{l}-s_{\ell}\right)).  \label{eq: var(sum)}
\end{align}
The strong-correlation covariance matrix, $\mathbf{\Omega}_{sc}$, is
recognized as the large-$n$ analogue of this expression after appropriate
normalization and averaging over the locations. The weak-correlation
covariance matrix, $\mathbf{\Omega}_{wc},$ differs from $\mathbf{\Omega}%
_{sc} $ in two ways. First, because $c_{n}\rightarrow\infty$ in the
weak-correlation case, and $\sigma_{B}(r)$ vanishes for large $|r|$, the
second term in $\mathbf{\Omega}_{wc}$ is recognized as the limit of $\mathbf{%
\Omega}_{sc}$ as the double integral concentrates entirely on `the diagonal'
where $r\approx s$. Second, as outfill becomes more important (that is, $%
a_{n}=c_{n}^{n}/n$ gets larger), variances become more important relative to
covariances; this explains the first term in $\mathbf{\Omega}_{wc}.$
\end{remark}

\begin{remark}
The form of $\mathbf{V}_{2}$ is further recognized as the limit covariance
matrix in a model where the observations are independent, with variance
proportional to $g(s_{l})$. Thus, $\mathbf{V}_{2}$ is what one would obtain
for the limit covariance matrix under a specific form of non-stationarity.
Intuitively, a high density area does not only yield many observations, but
under spatial correlation, the variance contribution is further amplified by
the resulting high average correlation.
\end{remark}

\begin{remark}
In the strong-correlation case, normality is assumed. That said, CLTs have
been established also for strongly correlated models when $d=1$ (i.e., the
time series case), such as \cite{Taqqu1975}, \cite{Phillips87b} or \cite%
{Chan87}, and to a lesser extent also for $d>1$, as in \cite{Wang2014} or 
\cite{Lahiri2016}. For the weak correlation case, large-sample normality
follows from Theorem 3.2 in \cite{Lahiri_2003}.
\end{remark}

\begin{remark}
When $g(s)$ is constant, so the spatial distribution is uniform, $\mathbf{V}%
_{1}\propto \mathbf{V}_{2}$ and $\mathbf{\Omega }_{wc}\propto \int \mathbf{w}%
^{0}(s)\mathbf{w}^{0}(s)^{\prime }ds$. Thus, in a leading case with
orthogonal $w_{j}$ of length $1/\sqrt{q}$, $\int w_{j}(s)w_{i}(s)dG(s)=q^{-1}%
\mathbf{1}[i=j]$, $\mathbf{\Omega }_{wc}\propto \func{diag}(1,q^{-1}\mathbf{I%
}_{q})$, a familiar result from the literature on HAR inference in time
series with regularly spaced observations. Importantly, while this result
holds under constant $g(s)$, it does \emph{not} hold for other spatial
distributions, so that the typical HAR results about inconsistent variance
estimators for regularly spaced time series under weak dependence do not
carry over to the spatial case.
\end{remark}

\subsection{\label{sec:asy_tau}Large-sample null rejection probability of
spatial t-tests}

This section presents a useful representation for the limiting distribution
of $\tau_{n}^{2}(\mathbf{W}_{n}\mathbf{W}_{n}^{\prime})$ under the
assumptions of Lemma \ref{lm:NormalLimit_q}.

\begin{theorem}
\label{thm:tau_finite_q}For $\func{cv}>0$, let $\mathbf{D}(\func{cv})=%
\limfunc{diag}(1,-\func{cv}^{2}\mathbf{I}_{q})$, $\mathbf{A}=\mathbf{D}(%
\func{cv})\mathbf{\Omega }$ with $\mathbf{\Omega }\in \{\mathbf{\Omega }%
_{sc},\mathbf{\mathbf{\Omega }}_{wc}\}$, and let $(\omega _{0},\omega
_{1},...,\omega _{q})$ denote the eigenvalues of $\mathbf{A}$ ordered from
largest to smallest. Then under the assumptions of Lemma \ref%
{lm:NormalLimit_q}, under the null hypothesis and with $%
(Z_{0},Z_{1},...,Z_{q})^{\prime }\sim \mathcal{N}(0,\mathbf{I}_{q+1})$,

(i) $\omega_{0}>0$, and $\omega_{i}\leq0$ for $i\geq1$;

(ii) $\mathbb{P}\left(\tau_{n}^{2}(\mathbf{W}_{n}\mathbf{W}_{n}^{\prime})>%
\func{cv}^{2}|\mathbf{s}_{n}\right)\overset{p}{\rightarrow}\mathbb{P}%
\left(Z_{0}^{2}>\sum_{i=1}^{q}(-\frac{\omega_{i}}{\omega_{0}}%
)Z_{i}^{2}\right)$.
\end{theorem}

\begin{remark}
In the weak-correlation case with constant spatial density $g(s)$ and
orthogonal $w_{j}$ of length $1/\sqrt{q}$, $\mathbf{\Omega =\mathbf{\Omega }}%
_{wc}\propto \func{diag}(1,q^{-1}\mathbf{I}_{q})$. Thus $-\omega _{i}/\omega
_{0}=\func{cv}^{2}/q$, and the asymptotic rejection probability becomes the
corresponding quantile of the $F_{1,q}$ distribution, a result familiar from
the limiting distribution of projection based squared t-statistics in the
regularly spaced time series case.
\end{remark}

\begin{remark}
In the general weak correlation case with arbitrary spatial density $g$, $%
\mathbf{\Omega}_{wc}=a\sigma_{B}(0)\mathbf{V}_{1}+\left(\int\sigma_{B}(s)ds%
\right)\mathbf{V}_{2}$. Because $\tau_{n}^{2}$ is a scale-invariant function
of $\mathbf{u}_{n}$, it is without loss of generality to normalize the scale
of $\sigma_{B}(\cdot)$ so that $a\sigma_{B}(0)+\int\sigma_{B}(s)ds=1$. Under
this normalization 
\begin{equation}
\mathbf{\Omega}_{wc}=\kappa\mathbf{V}_{1}+(1-\kappa)\mathbf{V}_{2}
\label{eq:Om_wd_kappa}
\end{equation}
where $\kappa$ is scalar with $0\leq\kappa<1$. Thus, the limiting CDF of $%
\tau_{n}^{2}$ is seen to depend on $\sigma_{B}$ only through the scalar $%
\kappa$; the matrices $\mathbf{V}_{1}$ and $\mathbf{V}_{2}$ are functions of
the weights $\mathbf{w}^{0}$ and the spatial density $g$. The scalar $\kappa$
thus completely summarizes the large sample effect of alternative underlying
random fields $B$ and weak correlation sequences $c_{n}\rightarrow\infty$.
\end{remark}

\subsubsection{Extensions for estimated weights}

For SCPC and other estimators, the weights in $\mathbf{w}(s)$ are estimated
using the sample locations $\mathbf{s}_{n}$. The conditions under which
Lemma \ref{lm:NormalLimit_q} continues to hold for such estimated weights is
given in the following theorem.

\begin{theorem}
\label{thm:what}Suppose the mapping $\mathbf{\hat{w}}^{0}:\mathcal{S}\mapsto 
\mathbb{R}
^{q+1}$ is a function of $\mathbf{s}_{n}$ (but not of $B$), and 
\begin{equation}
\sup_{s\in \mathcal{S}}||\mathbf{\hat{w}}^{0}(s)-\mathbf{w}^{0}(s)||\overset{%
p}{\rightarrow }0.  \label{eq:what_conv}
\end{equation}%
Then Lemma \ref{lm:NormalLimit_q} and Theorem \ref{thm:tau_finite_q}
continue to hold with $\mathbf{\hat{W}}_{n}^{0}$ in place of $\mathbf{W}%
_{n}^{0}$, where the $l$th row of $\mathbf{\hat{W}}_{n}^{0}$ is equal to $(1,%
\mathbf{\hat{w}}(s_{l})^{\prime })$.
\end{theorem}

\begin{remark}
The theorem also accommodates location dependent convergent critical values $%
\func{cv}_{n}\overset{p}{\rightarrow}\func{cv}$ by setting $\mathbf{\hat{w}}%
^{0}(s)=(\func{cv}_{n}/\func{cv})\mathbf{w}^{0}(s)$.
\end{remark}

\subsubsection{\label{sec:kernel_est}Extension for kernel variance estimators%
}

This subsection discusses how these results can be generalized so they apply
to kernel-based variance estimators, $\hat{\sigma}_{n}^{2}(\mathbf{M}_{n}%
\mathbf{K}_{n}\mathbf{M}_{n})$ and associated t-statistics $\tau _{n}^{2}(%
\mathbf{M}_{n}\mathbf{K}_{n}\mathbf{M}_{n})$, where the $n\times n$ matrix $%
\mathbf{K}_{n}$ has $(l,\ell )$ element equal to $k(s_{l},s_{\ell })$ for a
positive semidefinite continuous kernel $k:\mathcal{S\times S}\mapsto 
\mathbb{R}
$. Since in our framework, $s_{l}\in \mathcal{S}$ for a fixed sampling
region $\mathcal{S}$, and $k$ does not depend on $n$, these kernel
estimators are spatial analogues of fixed-$b$ time series long-run variance
estimators considered by \cite{Kiefer05}, as also investigated by \cite%
{Bester_Conley_Hansen_Vogelsang_2016}.

Let $\mathbf{\hat{K}}_{n}=\mathbf{M}_{n}\mathbf{K}_{n}\mathbf{M}_{n}$, and
note that the $(l,\ell )$ element of $\mathbf{\hat{K}}_{n}$ is $\hat{k}%
_{n}(s_{l},s_{\ell })$ with 
\begin{equation}
\hat{k}_{n}(r,s)=k(r,s)-n^{-1}\sum_{l=1}^{n}k(s_{l},s)-n^{-1}%
\sum_{l=1}^{n}k(r,s_{l})-n^{-2}\sum_{l=1}^{n}\sum_{\ell
=1}^{n}k(s_{l},s_{\ell }).  \label{eq: kbar_n}
\end{equation}%
To begin, consider a simpler problem using a kernel that replaces the sample
means in (\ref{eq: kbar_n}) with populations means 
\begin{equation}
\overline{k}(r,s)=k(r,s)-\int k(u,s)dG(u)-\int k(r,u)dG(u)+\int \int
k(u,t)dG(u)dG(t).  \label{eq:kbar}
\end{equation}%
By Mercer's Theorem, $\overline{k}(r,s)$ has the representation 
\begin{equation}
\overline{k}(s,r)=\sum_{i=1}^{\infty }\lambda _{i}\varphi _{i}(s)\varphi
_{i}(r)  \label{eq:kbar spectral representation}
\end{equation}%
where $\{\lambda _{i},\varphi _{i}\}$ are the eigenvalues and eigenfunctions
of $\overline{k}$, with eigenvalues ordered from largest to smallest, $%
\mathbb{\int }\varphi _{i}(s)dG(s)=0$ and $\mathbb{\int }\varphi
_{i}(s)\varphi _{j}(s)dG(s)=\mathbf{1[}i=j]$.

Consider the problem with a truncated version of $\overline{k}$, 
\begin{equation*}
\overline{k}_{q}(s,r)=\sum_{i=1}^{q}\lambda_{i}\varphi_{i}(s)\varphi_{i}(r).
\end{equation*}
We can directly apply Theorem \ref{thm:tau_finite_q} using $%
w_{j}(s)=\lambda_{j}^{1/2}\varphi_{j}(s)$. Specifically, let $\mathbf{\bar{K}%
}_{n,q}$ be an $n\times n$ matrix with $(l,\ell)$ element equal to $%
\overline{k}_{q}(s_{l},s_{\ell})$. Then $\mathbf{u}_{n}^{\prime}\mathbf{\bar{%
K}}_{n,q}\mathbf{u}_{n}=\mathbf{u}_{n}^{\prime}\mathbf{W}_{n}\mathbf{W}%
_{n}^{\prime}\mathbf{u}_{n}$ so that $\tau_{n}^{2}(\mathbf{\bar{K}}%
_{n,q})=\tau_{n}^{2}(\mathbf{W}_{n}\mathbf{W}_{n}^{\prime})$, and $\mathbb{P}%
\left(\tau_{n}^{2}(\mathbf{\bar{K}}_{n,q})>\func{cv}^{2}|\mathbf{s}%
_{n}\right)\overset{p}{\rightarrow}\mathbb{P}\left(Z_{0}^{2}>\sum_{i=1}^{q}(-%
\frac{\omega_{i}}{\omega_{0}})Z_{i}^{2}\right)$ by Theorem \ref%
{thm:tau_finite_q}.

To extend this result to the original problem, it is useful to reformulate
it in terms of eigenvalues of linear operators. Specifically, denote by $%
\mathcal{L}_{G}^{2}$ the Hilbert space of functions $\mathcal{S}\mapsto 
\mathbb{R}
$ with inner product $\langle f_{1},f_{2}\rangle =\int f_{1}(s)f_{2}(s)dG(s)$%
. Normalize $\mathbf{\Omega }_{wc}=\kappa \mathbf{V}_{1}+(1-\kappa )\mathbf{V%
}_{2}$, as in (\ref{eq:Om_wd_kappa}). A tedious but straightforward
calculation (see (\ref{eq:op_equivalence})\ in the appendix) shows that the
eigenvalues $\omega _{i}$ of $\mathbf{A}=\mathbf{D}(\func{cv})\mathbf{\Omega 
}$ with $\mathbf{\Omega =\{\Omega }_{sc}\mathbf{,\Omega }_{wc}\}$ are also
the eigenvalues of finite rank self-adjoint linear operators $\mathcal{L}%
_{G}^{2}\mapsto \mathcal{L}_{G}^{2}$, namely $R_{sc}T_{q}R_{sc}$ and $%
R_{wc}T_{q}R_{wc}$ in the strong and weak correlation case, respectively,
where 
\begin{eqnarray*}
R_{sc}^{2}(f)(s) &=&\int \sigma _{B}(c(s-r))f(r)dG(r) \\
R_{wc}^{2}(f)(s) &=&(\kappa +(1-\kappa )g(s))f(s) \\
T_{q}(f)(s) &=&\int \left( 1-\func{cv}^{2}\overline{k}_{q}(s,r)\right)
f(r)dG(r).
\end{eqnarray*}%
This suggests that the limiting rejection probability for the original
non-truncated $\bar{k}$ might be characterized by the (potentially infinite)
number of eigenvalues of the operators $RTR:\mathcal{L}_{G}^{2}\mapsto 
\mathcal{L}_{G}^{2}$ with $R\in \{R_{wc},R_{sc}\}$, where 
\begin{equation*}
T(f)(s)=\int \left( 1-\func{cv}^{2}\overline{k}(s,r)\right) f(r)dG(r).
\end{equation*}%
The following theorem shows this to be the case, and it also includes the
generalization to sample demeaned kernels (\ref{eq: kbar_n}) instead of (\ref%
{eq:kbar}).

\begin{theorem}
\label{thm:kernel_conv}Let $\omega_{0}$ denote the largest eigenvalue, and $%
\omega_{i},i\geq1$ the remaining eigenvalues of $RTR$ for $%
R\in\{R_{wc},R_{sc}\}$. Then under the assumptions of Lemma \ref%
{lm:NormalLimit_q}, $\omega_{0}>0$ and $\omega_{i}\leq0$ for $i\geq1$, and $%
\mathbb{P}(\tau_{n}^{2}(\mathbf{\hat{K}}_{n})>\func{cv}^{2}|\mathbf{s}_{n})%
\overset{p}{\rightarrow}\QTR{up}{\mathbb{P}(Z_{0}^{2}>\sum_{i=1}^{\infty}(-%
\omega_{i}/\omega_{0})Z_{i}^{2})}.$
\end{theorem}

\begin{remark}
Under weak correlation the limit distribution of kernel-based spatial
t-statistics depends on the spatial density $g$, since the eigenvalues of $%
R_{wc}TR_{wc}$ are a function of $g$. This is analogous to the results for
projection estimators discussed above. Thus, in both cases, using a critical
value that is appropriate for i.i.d.~data (that is, setting $\kappa=1$) does
not, in general, lead to valid inference under weak correlation.
\end{remark}

\begin{remark}
The theorem is also applicable to projection estimators using basis
functions $\phi_{i}$ that are orthogonalized using the sample locations
(such as those suggested in \cite{Sun_Kim_2012}) by setting $%
k(r,s)=q^{-1}\sum_{i=1}^{q}\phi_{i}(r)\phi_{i}(s)$.
\end{remark}

\begin{remark}
\label{rmk:asy_bias}The framework of Theorem \ref{thm:kernel_conv} also
sheds light on the asymptotic bias of kernel-based and orthogonal projection
estimators under weak correlation. The estimand $\sigma ^{2}$ is the
limiting variance of $a_{n}^{1/2}n^{-1/2}\sum_{l=1}^{n}u_{l}$, which under
the normalization (\ref{eq:Om_wd_kappa}) is equal to the (single) eigenvalue
of the operator $R_{wc}T_{\sigma ^{2}}R_{wc}$ with $T_{\sigma
^{2}}(f)(s)=\int f(r)dG(r)$, that is $\int (\kappa +(1-\kappa )g(s))dG(s).$
The expectation of $a_{n}\hat{\sigma}_{n}^{2}(\mathbf{\hat{K}}_{n})$
converges to the trace of the operator $R_{wc}T_{\bar{k}}R_{wc}$ with $T_{%
\bar{k}}(f)(s)=\int \overline{k}(s,r)f(r)dG(r)$, that is $\int (\kappa
+(1-\kappa )g(s))\bar{k}(s,s)dG(s)$. Thus, the estimator is asymptotically
unbiased for all $g$ if and only if $\bar{k}(s,s)=1$. For standard choices
of $k$, $k(s,s)=1$, so the only source of asymptotic bias is the demeaning
(and if the estimator $\hat{\sigma}_{n}^{2}$ uses the null value\textbf{%
\thinspace }$\mathbf{y}_{n}-\mu _{0}\mathbf{l}_{n}$ instead of the residuals 
$\mathbf{\hat{u}}_{n}$, the asymptotic bias is zero under the null
hypothesis). Moreover, if $k(r,s)$ concentrates around the `diagonal' where $%
r\approx s$, corresponding to a fixed-$b$ kernel estimator with small $b$,
the demeaning effect is small, as is the asymptotic variability of $a_{n}%
\hat{\sigma}_{n}^{2}(\mathbf{\hat{K}}_{n})$. Thus, fixed-$b$ kernel
estimators with standard kernel choices and small $b$ yield nearly valid and
efficient inference under weak correlation.

In contrast, orthogonal projection estimators where $\bar{k}%
(r,s)=q^{-1}\sum_{i=1}^{q}\phi_{i}(r)\phi_{i}(s)$ do not share this
approximate unbiasedness property, even for $q$ large, since $%
\int\phi_{i}(s)^{2}dG(s)=1$ does not, in general, imply that $\bar{k}%
(s,s)=q^{-1}\sum_{i=1}^{q}\phi_{i}(s)^{2}\approx1$.
\end{remark}

The proof of Theorem \ref{thm:kernel_conv} involves showing that in large
samples, the difference between the eigenfunctions of the sample demeaned
kernel (\ref{eq: kbar_n}) and the population demeaned kernel (\ref{eq:kbar})
becomes small. The following lemma extends and adapts previous results by 
\cite{Rosasco2010} to the case of sample demeaned kernels.

\begin{lemma}
\label{thm:eigenfunc}Let $(\mathbf{\hat{v}}_{i},\hat{\lambda}_{i})$ with $%
\mathbf{\hat{v}}_{i}=(\hat{v}_{i,1},\ldots ,\hat{v}_{i,n})^{\prime }$ be the
eigenvector-eigenvalue pairs of $n^{-1}\mathbf{\hat{K}}_{n}$ with $\hat{%
\lambda}_{1}\geq \hat{\lambda}_{2}\geq \ldots \geq \hat{\lambda}_{n}$ and $%
n^{-1}\mathbf{\hat{v}}_{i}^{\prime }\mathbf{\hat{v}}_{i}=1$. For all $i$
with $\hat{\lambda}_{i}>0,$ define the $\mathcal{S}\mapsto 
\mathbb{R}
$ functions 
\begin{equation}
\hat{\varphi}_{i}(\cdot )=n^{-1}\hat{\lambda}_{i}^{-1}\sum_{l=1}^{n}\hat{v}%
_{i,l}\hat{k}_{n}(\cdot ,s_{l}).  \label{eq:phi_hat}
\end{equation}%
Let $\lambda _{(j)}$, $j=1,\ldots $ be the unique positive values of $%
\lambda _{i}$, ordered descendingly, and suppose $\lambda _{(j)}$ has
multiplicity $m_{j}\geq 1$. Then for any $p$ such that $\lambda _{(p)}>0$,

(a) there exist rotation matrices $\mathbf{\hat{O}}_{(j)}$ of dimension $%
m_{j}\times m_{j}$, $j=1,\ldots,p$ such that with $q=\sum_{j=1}^{p}m_{j}$, $%
\mathbf{\varphi}=(\varphi_{1},\ldots,\varphi_{q})^{\prime}$ and $\mathbf{%
\hat{\varphi}}=(\hat{\varphi}_{1},\ldots,\hat{\varphi}_{q})^{\prime}$, 
\begin{equation*}
\sup_{s\in\mathcal{S}}||\mathbf{\varphi}(s)-\limfunc{diag}(\mathbf{\hat{O}}%
_{(1)},\ldots,\mathbf{\hat{O}}_{(p)})\mathbf{\hat{\varphi}}%
(s)||=O_{p}(n^{-1/2})\text{;}
\end{equation*}

(b) $\sum_{i=1}^{q}(\hat{\lambda}_{i}-\lambda_{i})^{2}=O_{p}(n^{-1})$.
\end{lemma}

Part (a) shows convergence of the eigenspace corresponding to unique
eigenvalues, and part (b) shows convergence of the eigenvalues.

\subsubsection{\label{sec:SCPCconv}SCPC t-statistic}

Beyond its use in the proof of Theorem \ref{thm:kernel_conv}, Lemma \ref%
{thm:eigenfunc} can be used to establish the large sample distribution of
the SCPC t-statistic for nonrandom $q$ and critical value $\func{cv}$. Note
that in this application of Lemma \ref{thm:eigenfunc}, we are interested in
the eigenfunctions of the covariance kernel $k^{0}(r,s)=%
\sigma_{u}^{0}(r-s|c_{0})$ of the benchmark model, rather than the
eigenfunctions of a kernel that defines a kernel-based variance estimator.

Recall from Section \ref{sec:SCPC} that $\mathbf{r}_{i}$ is the eigenvector
of $\mathbf{M}_{n}\mathbf{\Sigma }_{n}(c_{0})\mathbf{M}_{n}$ corresponding
to the $i$th largest eigenvalue, normalized to satisfy $n^{-1}\mathbf{r}%
_{i}^{\prime }\mathbf{r}_{i}=1$. Let $\varphi _{i}^{0}$ be the eigenfunction
of the kernel $\overline{k}^{0}(r,s)$ corresponding to the $i$th largest
eigenvalue $\lambda _{i}^{0}$, where $k^{0}(r,s)=\sigma _{u}^{0}(r-s|c_{0})$
and $\bar{k}^{0}$ is the demeaned version of $k^{0}$ in analogy to (\ref%
{eq:kbar}). Lemma \ref{thm:eigenfunc} and a slightly extended version of
Theorem \ref{thm:what} (see Lemma \ref{lem:block_what} in the appendix) then
yields the following corollary.

\begin{corollary}
\label{cor:SCPC_conv}Suppose $\lambda _{q}^{0}>\lambda _{q+1}^{0}$. Then
Theorem \ref{thm:tau_finite_q} holds for $\tau _{\text{SCPC}}^{2}(q)=\tau
_{n}^{2}(q^{-1}\sum_{i=1}^{q}\mathbf{r}_{i}\mathbf{r}_{i}^{\prime })$ with $%
\mathbf{w}(s)=(\varphi _{1}^{0}(s),\ldots ,\varphi _{q}^{0}(s))^{\prime }/%
\sqrt{q}$.
\end{corollary}

\section{Size control of spatial t-statistics}

This section presents two results on size control of spatial t-statistics,
the first asymptotic and the second a finite-sample result, and applies
these to SCPC.

\subsection{Asymptotic size control under weak correlation}

As discussed above (see equation (\ref{eq:Om_wd_kappa})), under weak
correlation, the asymptotic rejection probability of $\tau _{n}$ for finite $%
q$ can be studied via $\mathbf{\Omega }_{wc}(\kappa )=\kappa \mathbf{V}%
_{1}+(1-\kappa )\mathbf{V}_{2}$, where the covariance function of $u$ and
the sequence $c_{n}$ affects the large-sample distribution of $\tau _{n}$
only through the scalar $\kappa \in \lbrack 0,1)$. Thus, if $\overline{\func{%
cv}}$ is such that $\sup_{0\leq \kappa <1}\mathbb{P}\left(
\sum_{i=0}^{q}\omega _{i}(\kappa ,\overline{\func{cv}})Z_{i}^{2}>0\right)
=\alpha $, where $\{\omega _{i}(\kappa ,\overline{\func{cv}})\}_{i=0}^{q}$
are the eigenvalues of $\mathbf{A}(\kappa ,\overline{\func{cv}})=\mathbf{D}(%
\overline{\func{cv}})\mathbf{\Omega }_{wc}(\kappa )$, then setting $\func{cv}%
_{n}\geq \overline{\func{cv}}$ for all $n$ yields inference that is
asymptotically robust under all forms of weak correlation covered by Theorem %
\ref{lm:NormalLimit_q} (ii). In the case of a kernel-based variance
estimator, the same holds as long as $\overline{\func{cv}}$ satisfies $%
\sup_{0\leq \kappa <1}\mathbb{P}\left( \sum_{i=0}^{\infty }\omega
_{i}(\kappa ,\overline{\func{cv}})Z_{i}^{2}>0\right) =\alpha $ where $%
\{\omega _{i}(\kappa ,\overline{\func{cv}})\}_{i=0}^{\infty }$ are the
eigenvalues of the linear operator $L(f)(s)=\int \sqrt{\kappa +(1-\kappa
)g(s)}\left( 1-\overline{\func{cv}}^{2}\overline{k}(s,r)\right) \sqrt{\kappa
+(1-\kappa )g(r)}f(r)dG(r)$.

The value $\overline{\func{cv}}$ depends on the spatial density $g$, which
can be seen directly by inspecting the form of $\mathbf{\Omega}_{wc}$ and
the operator $L$. In principle, one could use these expressions to estimate $%
\overline{\func{cv}}$ directly. But this would involve estimates of the
spatial density $g$, which leads to difficult bandwidth an other choices. We
now discuss a simpler approach.

Consider a benchmark model $B^{0}$ that satisfies the assumptions of Theorem %
\ref{lm:NormalLimit_q} (ii), such as the Gaussian exponential model
introduced in Section \ref{sec:SCPC}. Let $\sigma_{B}^{0}$ denote the
covariance kernel of $B^{0}$, and suppose $c_{n,0}$, is chosen so that $%
a_{n,0}=c_{n,0}^{d}/n\rightarrow a_{0}=0.$ For instance, $c_{n,0}=c_{0}>0$
satisfies this condition, as does $c_{n,0}=n^{1/d}/\log(n)$. Note that for
this model $\kappa=0$. Suppose $\func{cv}_{n}=\func{cv}_{n}(\mathbf{s}_{n})$
satisfies 
\begin{equation}
\sup_{c\geq c_{n,0}}\mathbb{P}_{\mathbf{\Sigma}(c)}^{0}(\tau_{n}^{2}\geq%
\func{cv}_{n}^{2}|\mathbf{s}_{n})\leq\alpha  \label{eq:robust_cv_n}
\end{equation}
where $\mathbb{P}_{\mathbf{\Sigma}(c)}^{0}$ is computed under the benchmark
model, that is under $\mathbf{u}_{n}|\mathbf{s}_{n}\sim\mathcal{N}(0,\mathbf{%
\Sigma}(c))$ with $\mathbf{\Sigma}(c)$ the covariance matrix of $%
(B^{0}(cs_{1}),...,B^{0}(cs_{n}))^{\prime}$.

\begin{theorem}
\label{thm:cv_n}Let $\func{cv}_{n}^{2}$ satisfy (\ref{eq:robust_cv_n}).
Under weak correlation in the sense of Lemma \ref{lm:NormalLimit_q} (ii),
for t-statistics covered by Theorems \ref{thm:tau_finite_q}, \ref{thm:what}, %
\ref{thm:kernel_conv} and Corollary \ref{cor:SCPC_conv}, $\max (\overline{%
\func{cv}}^{2}-\func{cv}_{n}^{2},0)\overset{p}{\rightarrow }0$.
Consequently, for any $\epsilon >0$, $\lim \sup_{n}\mathbb{P(P}(\tau
_{n}^{2}>\func{cv}_{n}^{2}|\mathbf{s}_{n})>\alpha +\epsilon )\rightarrow 0$,
so that $\lim \sup_{n}\mathbb{P}(\tau _{n}^{2}\geq \func{cv}_{n}^{2})\leq
\alpha $.
\end{theorem}

The intuition for Theorem \ref{thm:cv_n} is as follows. The critical value $%
\func{cv}_{n}$ in (\ref{eq:robust_cv_n}) is valid in the benchmark model for
all $c\geq c_{n,0}$ and $n$. Thus, it is also valid along arbitrary
sequences $c_{n}\geq c_{n,0}$. Since the $c_{n,0}$ model has $\kappa=0$,
there exists sequences $c_{n}\geq c_{n,0}$ that induce any $%
\kappa\in\lbrack0,1)$ in the benchmark model; thus different sequences $%
c_{n} $ in the benchmark model recreate any possible limit distribution
under generic weak correlation, so that size control in the benchmark model
for all $c\geq c_{n,0}$ translates into size control under generic weak
correlation.

\subsubsection{Implications for SCPC}

For SCPC, the benchmark covariance kernel for $B^{0}$ is exponential $\sigma
_{B}^{0}(r,s)=\exp (-||r-s||)$ and (from equation (\ref{eq:SCPC cvq})) the
critical value is chosen to satisfy (\ref{eq:robust_cv_n}) with equality.
Thus, with a fixed value of $c_{0}$, the SCPC t-test $\tau _{\text{SCPC}}(q)$
controls size in large samples under generic weak correlation.\footnote{%
Technically, the SCPC choice of $q$ in (\ref{eq:SCPC Length}) is also a
function of the locations of $\mathbf{s}_{n}$, so $q_{\text{SCPC}}$ is
random. However, the argument that establishes Theorem \ref{thm:cv_n} can be
extended under this complication as long as $q_{\text{SCPC}}\leq q_{\max }$
almost surely for some finite and fixed $q_{\max }$. See Theorem \ref%
{thm:asysize_qhat} in the appendix for a formal statement.}

In addition and by construction, the SCPC critical value is chosen to
satisfy the size constraint for all values of $c\geq c_{0}$ in the benchmark
model. Thus, size is controlled by construction also in strong-correlation
models with exponential covariance kernels for all $c\geq c_{0}$.

\subsection{\label{subsec:Finite-sample-size-control}Finite sample size
control in the Gaussian model}

The asymptotic results of the last subsection are comforting, but in finite
samples, the robustness of a spatial t-statistic with critical value chosen
according to (\ref{eq:robust_cv_n}) still depends on the choice of $c_{n,0}$
and the benchmark model. This motivates investigating size control in finite
samples, which potentially includes `strong' correlation cases.

We restrict attention to Gaussian models where $\mathbf{y}\sim \mathcal{N}(%
\mathbf{l}\mu ,\mathbf{\Sigma })$ for some $\mathbf{\Sigma }$ and implicitly
condition on $\mathbf{s}$, and we also omit the dependence on $n$ to ease
notation. In this finite sample conditional framework, the distinction
between $\mathbf{W}$ and $\mathbf{\hat{W}}$ is immaterial, so for
simplicity, we write $\tau ^{2}(\mathbf{W}\mathbf{W}^{\prime })$ for the
t-statistic.\footnote{%
This also covers kernel variance estimators by setting $q=T-1$ and using the
Choleksy decomposition $\mathbf{MKM}=\mathbf{WW}^{\prime }$.}

Let $\mathcal{V}$ denote a set of covariance matrices. A test using the
t-statistic $\tau ^{2}(\mathbf{W}\mathbf{W}^{\prime })$ with critical value $%
\func{cv}$ is robust for\emph{\ }$\mathcal{V}$ if $\sup_{\mathbf{\Sigma }\in 
\mathcal{V}}\mathbb{P}_{\mathbf{\Sigma }}(\tau ^{2}(\mathbf{W}\mathbf{W}%
^{\prime })>\func{cv}^{2})\leq \alpha $. For a finite or parametric set of $%
\mathcal{V}$, $\sup_{\mathbf{\Sigma }\in \mathcal{V}}\mathbb{P}_{\mathbf{%
\Sigma }}(\tau ^{2}(\mathbf{W}\mathbf{W}^{\prime })>\func{cv}^{2})$ can be
established numerically. We therefore focus on an analytical robustness
result for a non-parametric class $\mathcal{V}$.

Specifically, we establish a set of readily verifiable sufficient conditions
to check robustness for sets $\mathcal{V}$ that are composed of mixtures of
parametric covariance matrices $\mathbf{\Sigma }^{p}(\theta )$ for $\theta
\in \Theta $. We then apply this result to a set of Mat\'{e}rn covariance
matrices with parameter $\theta $ and investigate the robustness of SCPC
over arbitrary mixtures of these Mat\'{e}rn models. In addition, we use the
result to study the robustness of a popular projection based t-test in a
regularly spaced time series setting.

Consider a benchmark model with $\mathbf{\Sigma }=\mathbf{\Sigma }_{0}$, and
suppose that $\func{cv}$ has been chosen so that $\mathbb{P}_{\mathbf{\Sigma 
}_{0}}(\tau ^{2}(\mathbf{W}\mathbf{W}^{\prime })>\func{cv}^{2})=\alpha .$ We
are interested in conditions under which 
\begin{equation}
\mathbb{P}_{\mathbf{\Sigma }_{1}}(\tau ^{2}(\mathbf{W}\mathbf{W}^{\prime })>%
\func{cv}^{2})\leq \alpha \text{ for $\mathbf{\Sigma }_{1}=\int_{\Theta }%
\mathbf{\Sigma }^{p}(\theta )dF(\theta )$}  \label{eq:Sig1_sizecontrol}
\end{equation}%
for a probability distribution $F$.

Let $\lambda_{j}(\cdot)$ denote the $j$th largest eigenvalue of some matrix.

\begin{theorem}
\label{thm:robust}Let $\mathbf{\Omega }_{0}=\mathbf{W}^{0\prime }\mathbf{%
\mathbf{\Sigma }}_{0}\mathbf{W}^{0}$, $\mathbf{\Omega }(\theta )=\mathbf{W}%
^{0\prime }\mathbf{\Sigma }^{p}(\theta )\mathbf{W}^{0}$, and assume $\mathbf{%
\Omega }_{0}$ and $\mathbf{\Omega }(\theta ),$ $\theta \in \Theta $ are full
rank. Suppose $\mathbf{A}_{0}=\mathbf{D}(\func{cv})\mathbf{\Omega }_{0}$ is
diagonalizable, and let $\mathbf{P}$ be its eigenvectors. Let $\mathbf{A}%
(\theta )=\mathbf{P}^{-1}\mathbf{D}(\func{cv})\mathbf{\Omega }(\theta )%
\mathbf{P}$ and $\mathbf{\bar{A}}(\theta )=\tfrac{1}{2}(\mathbf{A}(\theta )+%
\mathbf{A}(\theta )^{\prime })$. Suppose $\mathbf{A}_{0}$ and $\mathbf{A}%
(\theta ),$ $\theta \in \Theta $ are scale normalized such that $\lambda
_{1}(\mathbf{A}_{0})=\lambda _{1}(\mathbf{A}(\theta ))=1$. Let 
\begin{eqnarray*}
\nu _{1}(\theta ) &=&\lambda _{q}(-\mathbf{\bar{A}}(\theta ))-\lambda _{1}(%
\mathbf{\bar{A}}(\theta ))\lambda _{q}(-\mathbf{A}_{0})-(\lambda _{1}(%
\mathbf{\bar{A}}(\theta ))-1) \\
\nu _{i}(\theta ) &=&\lambda _{q+1-i}(-\mathbf{\bar{A}}(\theta ))-\lambda
_{1}(\mathbf{\bar{A}}(\theta ))\lambda _{q+1-i}(-\mathbf{A}_{0})\text{ for }%
i=2,\ldots ,q.
\end{eqnarray*}%
If for some probability distribution $F$ on $\Theta $, $\sum_{i=1}^{j}\int
\nu _{i}(\theta )dF(\theta )\geq 0$ for all $1\leq j\leq q$, then (\ref%
{eq:Sig1_sizecontrol}) holds.
\end{theorem}

\begin{remark}
If $\sum_{i=1}^{j}\nu_{i}(\theta)\geq0$ for all $\theta\in\Theta$ and $1\leq
j\leq q$, then the theorem implies that $\mathbb{P}_{\mathbf{\Sigma}%
_{1}}(\tau^{2}(\mathbf{W}\mathbf{W}^{\prime})>\func{cv}^{2})\leq\alpha$ for $%
\mathbf{\Sigma}_{1}$ an arbitrary mixture of $\mathbf{\Sigma}^{p}(\theta)$.
\end{remark}

\begin{remark}
Note that for $\mathbf{\Sigma }^{p}(\theta _{0})=\mathbf{\Sigma }_{0}$, $\nu
_{i}(\theta _{0})=0$ for $1\leq j\leq q$, so the inequalities of the theorem
have no `minimal slack' and potentially apply also to parametric models with
a covariance matrix $\mathbf{\Sigma }^{p}(\theta )$ that takes on values
arbitrarily close to $\mathbf{\Sigma }_{0}$.\textbf{\ }
\end{remark}

\begin{remark}
As shown in Theorem \ref{thm:tau_finite_q}, the eigenvalues of $\mathbf{A}%
_{0}$ and $\mathbf{A}(\theta)$ (or, equivalently, of $\mathbf{D}(\func{cv})%
\mathbf{\Omega}(\theta)$) govern the rejection probability of $\tau^{2}(%
\mathbf{W}\mathbf{W}^{\prime})$ under $\mathbf{\Sigma}_{0}$ and $\mathbf{%
\Sigma}^{p}(\theta)$. Given the scale normalization $\lambda_{1}(\mathbf{A}%
_{0})=\lambda_{1}(\mathbf{A}(\theta))=1$, if $-\lambda_{j}(\mathbf{A}%
(\theta))\geq-\lambda_{j}(\mathbf{A}_{0})$ for all $j\geq2$, then the result
there implies that $\mathbb{P}_{\mathbf{\Sigma}^{p}(\theta)}(\tau^{2}(%
\mathbf{W}\mathbf{W}^{\prime})>\func{cv}^{2})\leq\mathbb{P}_{\mathbf{\Sigma}%
_{0}}(\tau^{2}(\mathbf{W}\mathbf{W}^{\prime})>\func{cv}^{2})$. It follows
from an integral representation (cf.~equation (\ref{eq:BS}) below) that the
null rejection probability of the t-statistic is Schur convex in these
negative eigenvalues, so that the inequality holds whenever the negative
eigenvalues of $\mathbf{A}(\theta)$ weakly majorize those of $\mathbf{A}_{0}$%
. Majorization inequalities about eigenvalues of sums of matrices from \cite%
{Marshall2011} and additional calculations then extend this further to the
result in Theorem \ref{thm:robust}.
\end{remark}

\begin{remark}
The conditions of Theorem \ref{thm:robust} implicitly depend on the
locations $\mathbf{s}$, so the implications are specific to the application.
In the spatial case, the practical importance of the theorem is that the
conditions are straightforward to check numerically for a given parametric
family $\mathbf{\Sigma}^{p}(\theta)$. This can establish a range of
robustness of a spatial t-test in a given application and is illustrated in
the next subsection with the SCPC t-test and the Mat\'{e}rn class of spatial
correlations. The theorem also provides insights for inference in the
regularly-spaced time series case, where the spatial design is fixed across
applications. This is illustrated in the subsequent subsection for a
projection-based t-statistic for mixtures of AR(1) processes and processes
that are `less persistent' than a benchmark AR(1) model.
\end{remark}

\subsubsection{Implications for SCPC}

The critical value for the SCPC t-test is chosen to control size in
exponential models with $c\geq c_{0}$, where $c_{0}$ is calibrated to a
value $\overline{\rho}_{0}$. Because $\overline{\rho}$ is monotone in $c$,
the resulting SCPC t-test controls size for all $\overline{\rho}\leq%
\overline{\rho}_{0}$ in the exponential model by construction.

Let $\mathbf{\Sigma}^{p}(\theta)$ denote the covariance matrix associated
with a parameter $\theta$, with average pairwise correlation $\overline{\rho}%
(\theta)$. Let $\Theta_{\overline{\rho}_{L},\overline{\rho}_{U}}=\{$$\theta|%
\overline{\rho}_{L}\leq\overline{\rho}(\theta)\leq\overline{\rho}_{U}\}$
denote the set of values of $\theta$ that induce correlations between $%
\overline{\rho}_{L}$ and $\overline{\rho}_{U}$. If the inequalities in
Theorem \ref{thm:robust} are satisfied for all values of $\theta\in\Theta_{%
\overline{\rho}_{L},\overline{\rho}_{U}}$, then the SCPC t-test controls
size for all mixtures of $\mathbf{\Sigma}^{p}(\theta)$ in this set.

In this section we consider $\mathbf{\Sigma }^{p}(\theta )$ computed from Mat%
\'{e}rn processes with parameter $\theta =(\nu ,c)$, where $\nu $ and $c$
are positive constants. If $u$ follows a Mat\'{e}rn process, its covariance
function $\sigma _{u}(r-s)$ depends on the locations only through $d=||r-s||$%
. For $\nu \in \{1/2,3/2,5/2,\infty \}$, the Mat\'{e}rn covariance functions
are

\begin{itemize}
\item $\nu=1/2$: $\sigma_{u}(d)\propto\exp[-cd]$

\item $\nu=3/2$: $\sigma_{u}(d)\propto\left(1+\sqrt{3}dc\right)\exp[-\sqrt{3}%
cd]$

\item $\nu=5/2$: $\sigma_{u}(d)\propto\left(1+\sqrt{5}dc+(5/2)d^{2}c^{2}%
\right)\exp[-\sqrt{5}cd]$

\item $\nu=\infty$: $\sigma_{u}(d)\propto\exp[-c^{2}d^{2}/2].$
\end{itemize}

For any $\mathbf{\Sigma }(c_{0})$ it is straightforward to compute the
bounds $\overline{\rho }_{L}$ and $\overline{\rho }_{U}$ such that the
inequalities in Theorem \ref{thm:robust} are satisfied for all values of $%
\theta \in \Theta _{\overline{\rho }_{L},\overline{\rho }_{U}}$ with $\nu
\in \{1/2,3/2,5/2,\infty \}$ and $c>0$. We carried out this exercise for the
U.S.~states spatial correlation designs of Section \ref{sec:SCPC} (the
calculations for one set of locations take less than a second). We find $%
\overline{\rho }_{L}\leq 0.001$ and $\overline{\rho }_{U}=\overline{\rho }%
_{0}\in \{0.02,0.10\}$, with very few minor exceptions.

We conclude that SCPC controls size in finite Gaussian samples for a wide
range of Mat\'{e}rn process mixtures that imply $\overline{\rho}\leq%
\overline{\rho}_{0}$, at least for this set of spatial designs.

\subsubsection{Implications for regularly-spaced time series}

The spatial design is fixed for regularly-spaced time series, so the theorem
can provide general robustness results. Consider, for instance, the equal
weighted cosine (EWC) projection estimator of M\"{u}ller (2004, 2007), \cite%
{Lazarus_etal_JBES_2018} and Dou (2019) where $\mathbf{w}(s)=\sqrt{2/q}(\cos
\pi s,\cos (2\pi s),\ldots ,\cos (q\pi s))$. Suppose the critical value $%
\func{cv}_{n}$ is chosen so that size is controlled in a Gaussian AR(1) with
coefficient $\exp (-c_{0}/n)$, and $q$ is chosen to minimize expected length
in the i.i.d.~model. For $c_{0}=10$, $c_{0}=25$ and $c_{0}=50$, we obtain $%
q=5,7$ and $10$, respectively, for all $n\in \{50,100,500\}$. Call this test
the EWC$(c_{0})$ t-test.

Calculations based on Theorem \ref{thm:robust} for these values of $c_{0}$
and $n$ show that the EWC$(c_{0})$ t-test controls size for arbitrary
mixtures of AR(1) processes with coefficients $\exp (-c/n)$, $c\geq c_{0}$.
By taking the limit in $n$ and using standard local-to-unity weak
convergence results (as in \cite{Muller14}), one can further apply Theorem 1
to the limiting covariance matrices $\mathbf{\Omega }_{0}$ and $\mathbf{%
\Omega }(\theta )$ to study asymptotic robustness of the EWC$(c_{0})$ t-test
with an asymptotically justified critical value (which are equal to $\func{cv%
}=3.53$, $2.71$, $2.40$ for $c_{0}=10$, $25$, $50$, respectively). Another
numerical calculation based on Theorem \ref{thm:robust} then shows that
these EWC$(c_{0})$ t-tests control asymptotic size for underlying processes
that are arbitrary mixtures of local-to-unity models with parameters $c\geq
c_{0}$.

Moreover, let $f_{n,0}:[-\pi ,\pi ]\mapsto \lbrack 0,\infty )$ be the
spectral density of an AR(1) process with coefficient $\exp (-c_{0}/n)$, so $%
f_{n,0}(\omega )\propto (1-2e^{-c_{0}/n}\cos \omega +e^{-2c_{0}/n})^{-1}$. A
spectral density $f_{n,1}$ would naturally be considered less persistent
than $f_{n,0}$ if $f_{n,1}(\omega )/f_{n,0}(\omega )$ is (weakly)
monotonically increasing in $|\omega |$. Denote all such functions by $%
\mathcal{F}_{n}$. Define 
\begin{equation*}
M=\frac{f_{n,1}(\pi )/f_{n,0}(\pi )}{f_{n,1}(0)/f_{n,0}(0)},
\end{equation*}%
so $M$ measures by how much $f_{n,1}(\omega )/f_{n,0}(\omega )$ increases
over $[0,\pi ]$, and denote by $\mathcal{F}_{n}^{\bar{M}}$ all functions in $%
\mathcal{F}_{n}$ with $M\leq \bar{M}$ for some $\bar{M}>1$. Then for any $%
f_{n,1}\in \mathcal{F}_{n}^{\bar{M}}$, there exists a CDF $H$ on $[0,\pi ]$
such that 
\begin{eqnarray*}
f_{1,n}(\omega ) &\propto &f_{n,0}(\omega )+(M-1)H(|\omega |)f_{n,0}(\omega )
\\
&=&\frac{\bar{M}-M}{\bar{M}-1}f_{n,0}(\omega )+\frac{M-1}{\bar{M}-1}\int [{{{%
{f_{n,0}(\omega )+(\bar{M}-1)\mathbf{1}[|\omega |\geq \theta ]f_{n,0}(\omega
)}}}}]dH(\theta )
\end{eqnarray*}%
so $f_{n,1}$ has a representation as a scale mixture of $f_{n,0}(\omega )+(%
\bar{M}-1)\mathbf{1}[|\omega |\geq \theta ]f_{n,0}(\omega )$, $0\leq \theta
\leq \pi $. After translating this back into a corresponding mixture of
covariance matrices $\mathbf{\Sigma }^{p}(\theta )$, an application of
Theorem \ref{thm:robust} shows that the EWC$(c_{0})$ t-test also controls
size in this class, for $(c_{0},\bar{M})\in \{(10,10),(25,10),(50,5)\}$ and
all $n\in \{50,100,500\}$. These results refine corresponding results in 
\cite{Dou_2019} that are based on a Whittle-type diagonal approximation to $%
\mathbf{\Sigma }$.

Taking limits as $n\rightarrow \infty $ yields a corresponding asymptotic
robustness statement: The function $f_{0}:%
\mathbb{R}
\mapsto \lbrack 0,\infty )$ with $f_{0}(\omega )=(\omega
^{2}+c_{0}^{2})^{-1} $ is proportional to the `local-to-zero' spectral
density (cf. M\"{u}ller and Watson (2016, 2017))\nocite{Muller14b}\nocite%
{Muller15d} of a local-to-unity process with parameter $c_{0}.$\ Consider
any process whose local-to-zero spectral density $f_{1}$ is such that $%
f_{1}(\omega )/f_{0}(\omega )$ is monotonically increasing in $|\omega |$
with $\lim_{\omega \rightarrow \infty }f_{1}(\omega )/f_{0}(\omega )\leq 
\bar{M}f_{1}(0)/f_{0}(0)$ and that satisfies the CLT in M\"{u}ller and
Watson (2016, 2017). A numerical calculation based on Theorem \ref%
{thm:robust} then shows that the EWC$(c_{0})$ t-tests for $(c_{0},\bar{M}%
)\in \{(10,10),(25,10),(50,5)\}$ controls size in large samples under all
such processes.

\subsection{\label{sec:hetlocsize}Size properties of SCPC under
heteroskedasticity and mismeasured locations}

The SCPC t-test is not robust to heteroskedasticity or measurement error in
locations by construction. For example, suppose that $u(s)=h(s)\tilde{u}(s)$%
, where $\tilde{u}$ is homoskedastic and satisfies the assumptions outlined
above for $u$, and $h:\mathcal{S}\mapsto%
\mathbb{R}%
$ is a non-random function that induces heteroskedasticity in the $u$
process. The linear combinations of $u$ studied in Lemma \ref%
{thm:tau_finite_q} are now $\sum_{l=1}^{n}\mathbf{w}^{0}(s_{l})u(s_{l})=$$%
\sum_{l=1}^{n}\mathbf{w}_{h}^{0}(s_{l})\tilde{u}(s_{l})$ where $\mathbf{w}%
_{h}^{0}(s)=\mathbf{w}^{0}(s)h(s)$. The results of the lemma and subsequent
theorems then follow with $\mathbf{w}_{h}^{0}$ replacing $\mathbf{w}^{0}$.
But, the test statistic and critical value is computed using $\mathbf{w}^{0}$%
, not $\mathbf{w}_{h}^{0}$, so that size control is not guaranteed, even in
large samples. An analogous problem arises when the locations $s_{i}$ are
measured with error.

In both cases, the particulars of the size distortion depend on the
distribution of spatial locations, $g$, the weights $\mathbf{w}^{0}$ (which
in turn depend on the value of $\overline{\rho}_{0}$ used to calibrate $%
c_{0} $), the function $h$ in the heteroskedastic model and the distribution
of the measurement error for the locations.

We summarize two experiments that illustrate and quantify the size
distortions in the U.S.~states spatial correlation designs. The first
experiment is a heteroskedastic model with $\log h$ increasing or decreasing
linearly from $\log h(s)=0$ to $\log h(s)=\log 3$ moving from the most
westward to the most eastward location, the experiment is repeated with $h$
increasing or decreasing moving north to south, and we record the largest of
the four rejection frequencies. Panel (a) of Figure \ref{fig:hlocerror}
plots the CDF of rejection frequencies for nominal 5\% SCPC tests for each $(%
\overline{\rho }_{0},g)$ pair. For these designs, the resulting size
distortions are not large, except for a few states with $\overline{\rho }%
_{0}=0.02$ and the light spatial density $g$, where rejection frequencies
approach 10\%.

The second experiment investigates location measurement error of a form
studied in \cite{Conley_Molinari_2007}. Specifically for each location, $%
s_{i}^{\ast }=s_{i}+e_{i}$ where $s_{i}^{\ast }$ is the measured location, $%
s_{i}$ is the true location and $e_{i}$ is the measurement error. The error
term is $e_{i}=(e_{1,i},e_{2,i})$ with $e_{1,i}$ the north-south and $%
e_{2,i} $ the east-west coordinate and $e_{j,i}$ i.i.d.$\mathcal{U}(-\delta
,\delta ) $\emph{\ }over $j$ and $i$, and $\text{$\delta =0.0375H$}$ with $H$
the length of the smallest square that encompasses all locations,
corresponding to \textquotedblleft level 4\textquotedblright\ errors in
Conley and Molinari's (2007) classification. The CDFs for the rejection
frequencies are shown in panel (b) of Figure \ref{fig:hlocerror}. Evidently,
measurement error of this sort has little effect on the size of SCPC under
uniformly distributed locations, but can have a substantial effect for
highly concentrated spatial distributions, especially when $\overline{\rho }%
_{0}=0.02$.

\begin{figure}[tbp]
\centering{}\includegraphics[scale=0.8]{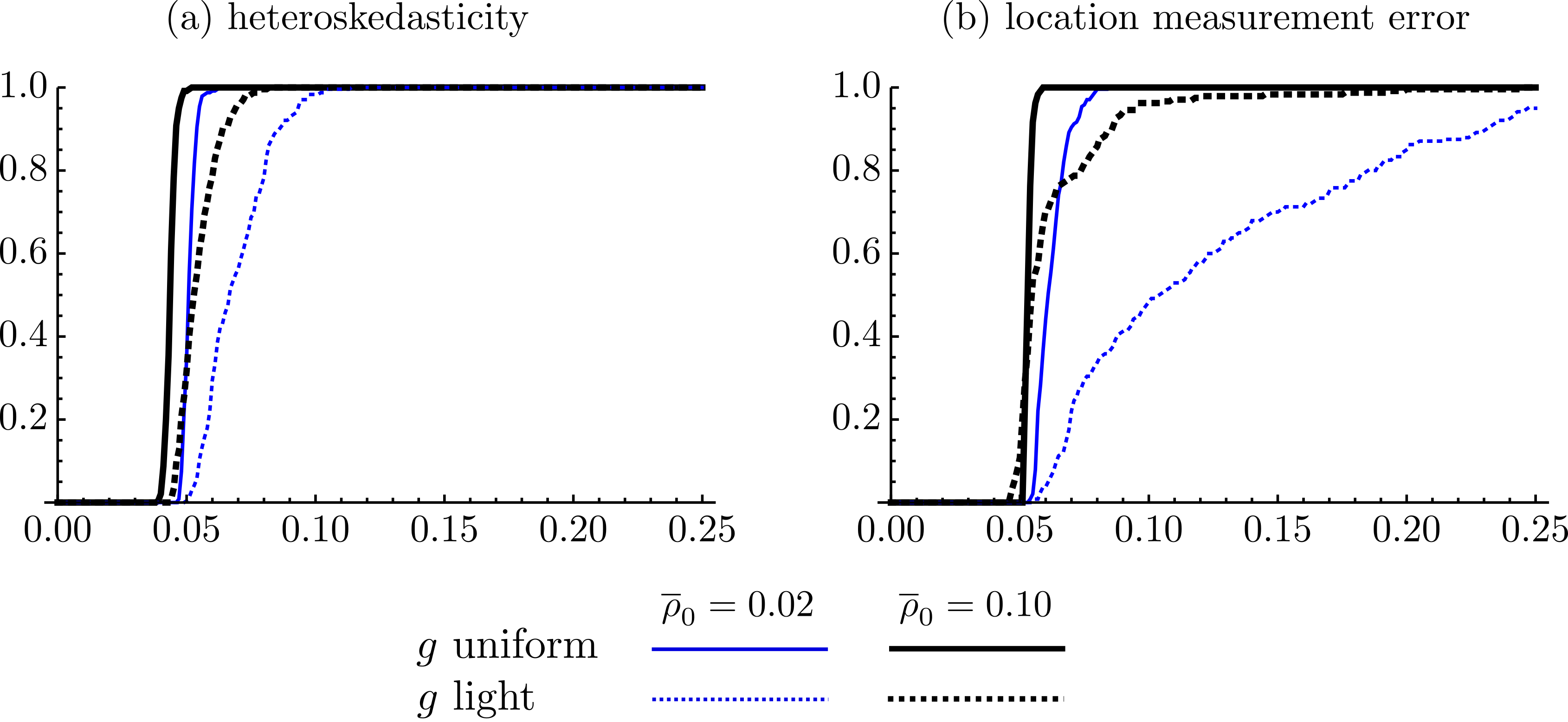}
\caption{CDFs of Size under Heteroskedasticity and Location Measurement
Error }
\label{fig:hlocerror}
\end{figure}

\section{\label{sec:Efficiency-of-SCPC}Efficiency of SCPC}

Figure \ref{fig:Expected-length-relative to oracle} showed the expected
length of the SCPC confidence interval relative to the length of an oracle
confidence interval that uses the true value of $\func{Var}(\sqrt{n}(%
\overline{y}-\mu))$ conditional on the observed locations $\mathbf{s}$. (As
before, in this subsection we keep the conditioning on $\mathbf{s}$ and the
dependence on $n$ implicit.) For studying efficiency, a more relevant
comparison involves the expected length of the SCPC confidence interval
relative to a confidence interval that, like SCPC, does not depend on the
true (unknown) value of $\func{Var}(\sqrt{n}(\overline{y}-\mu))$. Ideally,
such a comparison would involve SCPC and the most efficient method for
constructing a confidence interval. We undertake such a comparison here.

To be specific, let $\func{CS}(\mathbf{y})\subset%
\mathbb{R}%
$ denote a confidence set for $\mu$ constructed from $\mathbf{y}$. We
restrict attention to location and scale equivariant confidence sets, that
is $\func{CS}$ satisfies $\func{CS}(a_{\mu}+a_{\sigma}\mathbf{y}%
)=\{\mu_{0}:(\mu_{0}-a_{\mu})/a_{\sigma}\in\func{CS}(\mathbf{y})\}$ for all $%
\mathbf{y}$, $a_{\mu}\in%
\mathbb{R}%
$ and $a_{\sigma}>0$. As in Section \ref{subsec:Finite-sample-size-control},
we focus on the Gaussian model $\mathbf{y}\sim\mathcal{N}(\mathbf{l}\mu,%
\mathbf{\Sigma})$. We want to compare the SCPC interval with a confidence
interval that, like SCPC, has good coverage $\mathbb{P}_{\mathbf{\Sigma}%
}(\mu\in\func{CS}(\mathbf{y}))$ over a range of potential spatial
correlation patterns $\mathbf{\Sigma\in}\mathcal{V}$. The metric for
measuring efficiency is the expected length $\mathbb{E}^{1}[\int\mathbf{1}%
[x\in\func{CS}(\mathbf{y})]dx]$ in the i.i.d.~model $\mathbf{y}\sim\mathcal{N%
}(\mathbf{l}\mu,\mathbf{I})$.

Our choice of $\mathcal{V}$ is motivated by the structure of the SCPC
benchmark covariance matrix $\mathbf{\Sigma}(c_{0})$. The idea is to include
in $\mathcal{V}$ covariance matrices that are weakly less persistent than $%
\mathbf{\Sigma}(c_{0})$, and that cannot be easily distinguished from the
i.i.d.~model. To characterize these covariance matrices, note that $\mathbf{%
\Sigma}(c_{0})$ is generated from $u$, an isotropic random field with
covariance function $\sigma_{u}(s,r)=\exp(-c_{0}||s-r||)$. Isotropy implies
that the spectrum of this random field $F_{0}:%
\mathbb{R}%
^{d}\mapsto\lbrack0,\infty)$ at frequency $\mathbf{\omega}\in%
\mathbb{R}%
^{d}$ can be written as function of the scalar $\omega=||\mathbf{\omega}||$,
that is $F_{0}(\mathbf{\omega})=f_{0}(\omega)$ for some $f_{0}:%
\mathbb{R}%
\mapsto\lbrack0,\infty)$. As is well known, the exponential covariance model
for $d=2$ corresponds to a spectral density function $f_{0}$ proportional to 
$(c_{0}+\omega^{2})^{-3/2}$. By scale invariance of both $\func{CS}$ and the
SCPC interval, it is without loss of generality to set $f_{0}$ equal to 
\begin{equation*}
f_{0}(\omega)=\frac{1}{(c_{0}+\omega^{2})^{3/2}}.
\end{equation*}
For some $\bar{\omega}>0$, define $f_{\Delta}(\omega)=\mathbf{1}[|\omega|\leq%
\bar{\omega}](f_{0}(\omega)-f_{\Delta}(\bar{\omega}))$, and let $%
f_{R}(\omega)=f_{0}(\omega)-f_{\Delta}(\omega)$, so that 
\begin{equation*}
f_{0}(\omega)=f_{\Delta}(\omega)+f_{R}(\omega).
\end{equation*}
For $0\leq|\omega|\leq\bar{\omega}$, the density $f_{\Delta}$ is equal to $%
f_{0}(\omega)-f_{0}(\bar{\omega})$, so that the remainder $f_{R}(\omega)$ is
a continuous density that is flat for $|\omega|\leq\bar{\omega}$, and that
follows the same decline as $f_{0}$ for $|\omega|>\bar{\omega}$. Since both $%
f_{\Delta}(\omega)$ and $f_{R}(\omega)$ are non-negative, we have the
corresponding identity in covariance matrices 
\begin{equation}
\mathbf{\Sigma}(c_{0})=\mathbf{\Sigma}_{\Delta}(\bar{\omega})+\mathbf{\Sigma}%
_{R}(\bar{\omega})  \label{eq:Sig0_decomp}
\end{equation}
where $\mathbf{\Sigma}_{\Delta}(\bar{\omega})$ and $\mathbf{\Sigma}_{R}(\bar{%
\omega})$ are induced by the isotropic random fields with spectral densities 
$F_{\Delta}(\mathbf{\omega})=f_{\Delta}(||\mathbf{\omega||)}$ and $F_{R}(%
\mathbf{\omega})=f_{R}(||\mathbf{\omega||)}$, respectively.

Now consider the covariance matrix 
\begin{equation*}
\mathbf{\bar{\Sigma}}(\bar{\omega})=\mathbf{\Sigma }_{\Delta }(\bar{\omega}%
)+\lambda _{1}(\mathbf{\Sigma }_{R}(\bar{\omega}))\mathbf{I}_{n}
\end{equation*}%
where $\lambda _{1}(\mathbf{\Sigma }_{R}(\bar{\omega}))$ is the largest
eigenvalue of $\mathbf{\Sigma }_{R}(\bar{\omega})$. Since $f_{R}(\omega )$
is monotonically decreasing in $|\omega |$, also $\mathbf{\Sigma }_{R}(\bar{%
\omega})$ contributes to the persistence of $\mathbf{\Sigma }(c_{0})$ in (%
\ref{eq:Sig0_decomp}), so replacing it with white noise of weakly larger
variance should make inference about $\mu $ under $\mathbf{\bar{\Sigma}}(%
\bar{\omega})$ no harder than under $\mathbf{\Sigma }(c_{0})$.\footnote{%
In the regularly-spaced time series setting, white noise amounts to a flat
spectrum, so $\mathbf{\Sigma }_{0}(\bar{\omega})$ corresponds to an
underlying spectral density equal to $f_{\Delta }(\omega )+f_{0}(\bar{\omega}%
)$, which is the \textquotedblleft kinked\textquotedblright\ spectral
density considered by \cite{Dou_2019}. For arbitrary locations, however, the
domain of the spectrum doesn't fold onto the interval $[-\pi ,\pi ]$, so
that white noise cannot mathematically be represented by a flat spectrum.}
Said differently, a method that is robust under correlation patterns weakly
less persistent than $\mathbf{\Sigma }(c_{0})$ should continue to have good
coverage after replacing medium and high frequency variation in $\mathbf{y}$
by white noise, that is, under $\mathbf{\bar{\Sigma}}(\bar{\omega})$. This
motivates the set $\mathcal{V}=\{\mathbf{\bar{\Sigma}}(\bar{\omega})|\bar{%
\omega}>0\}$.

A calculation shows that in the U.S.~states spatial correlation designs, the
SCPC interval has good coverage properties under this $\mathcal{V}$. With $%
\alpha_{\text{SCPC}}(\bar{\omega})=\mathbb{P}_{\mathbf{\bar{\Sigma}}(\bar{%
\omega})}(\tau_{\text{SCPC}}^{2}>\func{cv}_{\text{SCPC}}^{2})$ for the
nominal 5\% level SCPC test, for most designs, $\sup_{\bar{\omega}%
\geq0}\alpha_{\text{SCPC}}(\bar{\omega})$ is equal or very close to 5\%, and
it never exceeds 8\%. To keep things on an equal footing, we allow $\func{CS}
$ the same degree of undercoverage, that is we consider the problem 
\begin{equation}
\inf_{\func{CS}}\mathbb{E}^{1}[\int\mathbf{1}[x\in\func{CS}(\mathbf{y})]dx]%
\text{ s.t. }\mathbb{P}_{\mathbf{\bar{\Sigma}}(\bar{\omega})}(\mu\notin\func{%
CS}(\mathbf{y}))\leq\max(\alpha_{\text{SCPC}}(\bar{\omega}),\alpha)\text{
for all }\bar{\omega}>0.  \label{eq:eff_prob}
\end{equation}
In words, we seek the invariant confidence set with the shortest expected
length in the i.i.d.~location model among all confidence sets that are as
robust as the SCPC interval under $\mathbf{\bar{\Sigma}}(\bar{\omega})$, $%
\bar{\omega}>0$.

Since $\bar{\omega}$ is one-dimensional, one can apply the numerical
techniques of \cite{Elliott15} and \cite{Muller15c} (also see \cite%
{Mueller20}) to obtain an informative lower bound on the objective $\inf_{%
\func{CS}}\mathbb{E}^{1}[\int\mathbf{1}[x\in\func{CS}(\mathbf{y})]dx]$ that
holds for \emph{any} equivariant $\func{CS}(\mathbf{y})$ that satisfies the
constraint in (\ref{eq:eff_prob}).

We compute such lower bounds in the U.S.~states spatial correlation designs.
Panel (a) of Figure \ref{fig:CDF-of-Expected lenght relative to lower bound}
shows the CDFs of the length of SCPC confidence intervals relative to the
lower bounds for the 240 designs in each $(\overline{\rho }_{0},g)$ pair.
The expected lengths of SCPC are within 7\% of the efficiency bound for all
designs when $\overline{\rho }_{0}=0.02$. When $\overline{\rho }_{0}=0.10$,
so that spatial correlation is high, and the spatial locations are highly
concentrated as under the light design, the expected length of the SCPC
confidence interval can be more that 15\% longer than the efficiency bound.
In part, this is because the implied efficient confidence sets are
complicated and rather uninterpretable functions of $\mathbf{y}$ in this
case. We thus repeat the exercise for confidence sets constrained to be
symmetric around $\overline{y}$ by imposing $\func{CS}(a_{\mu }+a_{\sigma }%
\mathbf{y})=\{\mu _{0}:(\mu _{0}-a_{\mu })/a_{\sigma }\in \func{CS}(\mathbf{y%
})\}$ for all $\mathbf{y}$, $a_{\mu }\in 
\mathbb{R}
$ and $a_{\sigma }\neq 0$. The results are summarized in panel (b), and we
can see that SCPC comes closer to the resulting higher bound on confidence
interval length.

\begin{figure}[tbp]
\centering{}\includegraphics[scale=0.8]{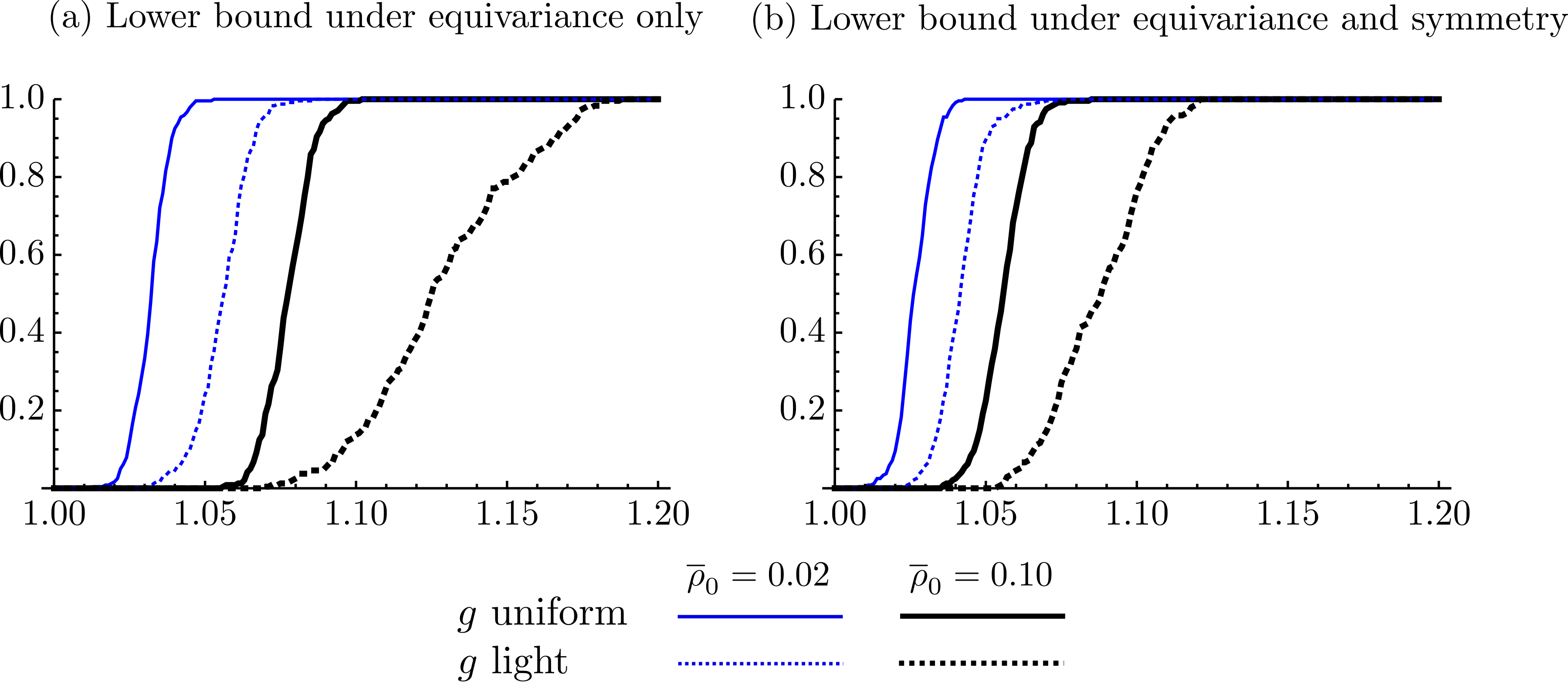}
\caption{CDFs of Expected Length of SCPC Relative to Lower Bound on Expected
Length }
\label{fig:CDF-of-Expected lenght relative to lower bound}
\end{figure}

\begin{remark}
These efficiency results also provide a limit on the possibility of using
data-dependent methods to learn about the value of the worst-case
correlation $c_{0}$: Since the i.i.d.~model corresponds to $%
c\rightarrow\infty$, if it was possible to learn the value of $c$ from the
data, one would be able to conduct much more efficient inference than what
is reported in Figure \ref{fig:CDF-of-Expected lenght relative to lower
bound}. The results here thus provide a rationalization for treating $c_{0}$
as given.\footnote{%
Also see Dou (2019) for a related discussion and associated impossibility
results.}
\end{remark}

\section{\label{sec:Comparison-with-Other}Comparison with other methods}

This section compares SCPC with other methods that have been proposed,
focusing on size and expected length of confidence intervals in the
benchmark Gaussian model with exponential covariance kernel and parameter $%
c_{0}$ (calibrated by $\bar{\rho}_{0}$). We consider two kernel-based
methods, two versions of a cluster method, and one projection method. All
these methods are t-statistic based tests of the form considered in Section
3.

The kernel based methods use a Bartlett kernel, $k(s,r)=k_{\text{Bartlett}%
}(||s-r||/b)$. The methods differ in their choice of bandwidth $b$ and
critical value. The first method uses a standard normal critical value with $%
b$ chosen so the resulting test has size as close as possible to $5\%$. This
is a version of the method proposed by \cite{Conley99}, but with an oracle
choice for the bandwidth. The second method sets $b=\max_{l,\ell}||s_{l}-s_{%
\ell}||$ and chooses the critical value to obtain exact coverage under $%
\mathbf{\Sigma}=\mathbf{I}$. This is the spatial analogue of the method
suggested by \cite{Kiefer00} (KVB) for regularly spaced time series. The
cluster methods follow the approach of \cite{Ibragimov10} (IM) with student-t%
$_{q}$ critical values and is implemented with $q=4$ and $q=9$ equal-sized
clusters.\footnote{%
The assignment of locations to clusters is performed sequentially, where at
each step, we minimize (across yet unassigned locations) the maximal
distance over clusters (among those that have not yet been assigned $n/q$
locations). Cluster distances are computed from the northwest, northeast,
southeast and southwest corners of the location circumscribing rectangle,
and in the $q=9$ case, also from the mid-points of the four sides of this
rectangle, and its center.} The projection method follows \cite{Sun_Kim_2012}%
. It uses a student-t$_{q}$ critical value and $q$ low-frequency Fourier
weights orthogonalized using the sample locations, where $q$ is chosen as a
function of the exponential model parameter $c_{0}$ using the formula in
their equation (8). The first and last method are thus tailored to the true
value $c_{0}$, just like SCPC.

\begin{figure}[tbp]
\centering{}%
\includegraphics[scale=0.8]{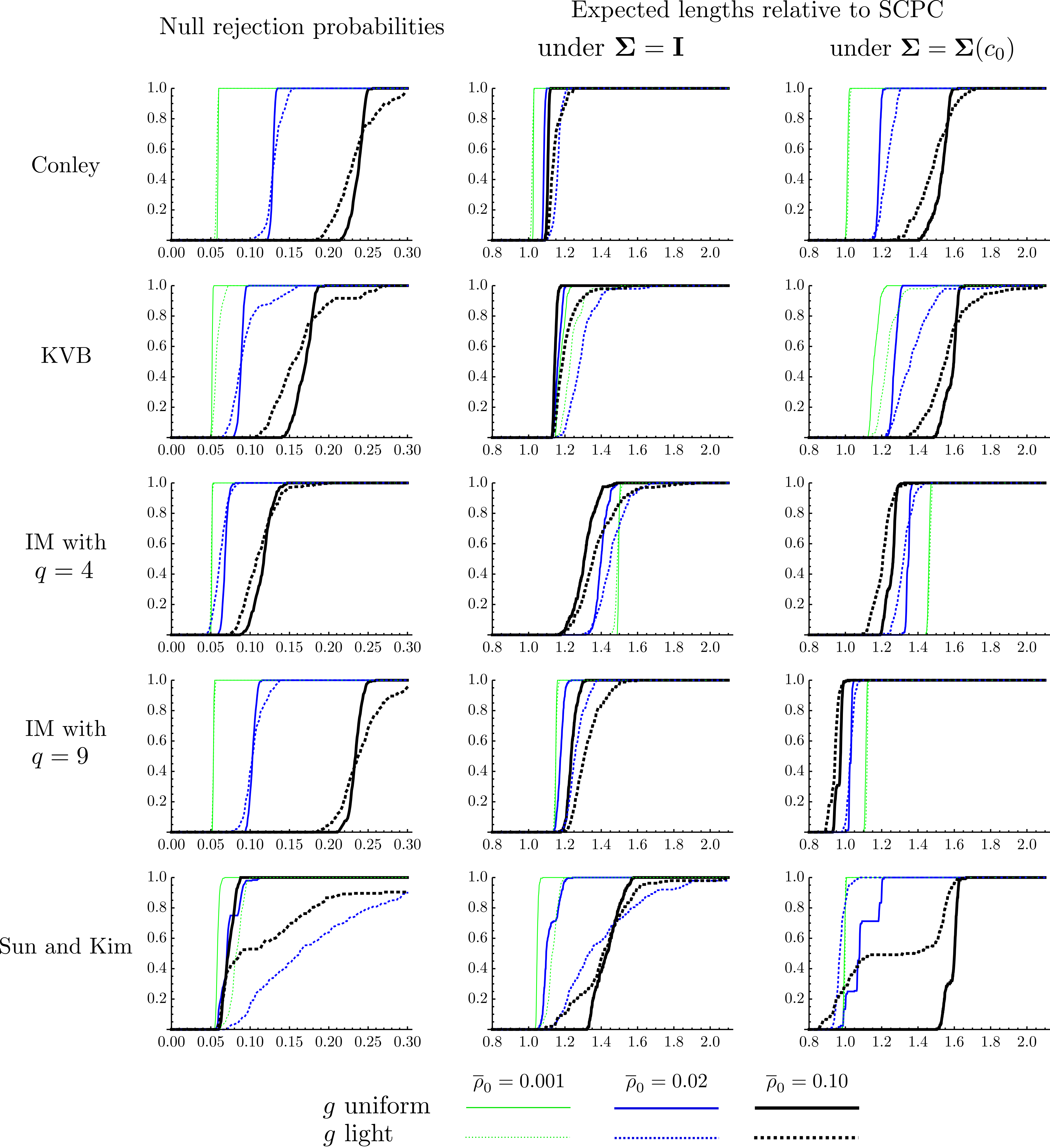}
\caption{CDFs of Null Rejection Probability and Relative Expected Length of
Alternative Methods }
\label{fig:Size-and-Relative Expected Length Alternative Methods}
\end{figure}

We analyze these methods in the U.S.~states spatial correlation designs,
augmented to also include the value $\overline{\rho }_{0}=0.001$ for the
average pairwise correlation to investigate performance under `weak' spatial
correlations. Figure \ref{fig:Size-and-Relative Expected Length Alternative
Methods} summarizes the results for size control and expected lengths by
plotting the CDFs for each $(\overline{\rho }_{0},g)$ pair. The first column
shows the null rejection frequency for each method; by construction, the
rejection frequency for SCPC is at most $5\%$ in all designs. The expected
lengths in the second and third column use size-corrected critical values to
ensure 95\% coverage under $\mathbf{\Sigma }(c_{0})$, and are given in
multiples of the expected length of the (non-adjusted) SCPC method. The
second column reports these relative expected lengths under $\mathbf{\Sigma }%
=\mathbf{I}$, and the third column under $\mathbf{\Sigma }(c_{0})$.

Looking at the first column, the kernel and cluster methods have null
rejection probabilities close to $5\%$ when $\overline{\rho }_{0}=0.001$,
but exhibit significant size distortions for $\overline{\rho }_{0}=0.02$ or $%
0.10$. Evidently, the kernel and cluster methods substantially underestimate
the variance of $\overline{y}$ for the latter two values of $\overline{\rho }%
_{0}$. In contrast, the Fourier projection method has relatively small size
distortions under $g=g_{\text{uniform}}$ but can have substantial size
distortions under $g=g_{\text{light}}$, even when $\bar{\rho}_{0}=0.001$.
This is consistent with the implications of Theorem \ref{thm:tau_finite_q}:
the student-t critical value for the projection method is appropriate when $%
\mathbf{\Omega }\propto \mathbf{I}$, which it is under weak-correlation with 
$g$ uniform, but not otherwise, even for large $q$ (cf. Remark \ref%
{rmk:asy_bias}).

The relative lengths shown in the second column are above unity, sometimes
by a wide margin, indicating that SCPC is closer to the efficiency bound
computed in Section \ref{sec:Efficiency-of-SCPC} than these alternative
methods, at least for the designs considered here. The third column shows
that this continues hold for lengths computed under $\mathbf{\Sigma }(c_{0})$
with a few exceptions. Notably, the expected length of the size-adjusted
9-cluster method is smaller than SCPC when $\overline{\rho }_{0}=0.10$. This
apparent good performance comes at the cost of substantially longer
confidence intervals in the i.i.d. model.

\section{Extensions and computational issues}

This section discusses extensions of the method to regression and GMM
models, some computational issues, and the multivariate extension of SCPC.

\subsection{\label{sec:gmm}Extensions to regression and GMM}

The extension of these results to regression and GMM problems follows from
standard arguments. For example, consider the linear regression problem 
\begin{equation}
w_{l}=x_{l}\beta+\mathbf{z}_{l}^{\prime}\delta+\varepsilon_{l}\text{ for $%
l=1,...,n$}
\end{equation}
where $\beta$ is the (scalar) parameter of interest, $\mathbf{z}_{l}$ are
additional controls in the regression, and $(w_{l},x_{l},\mathbf{z}_{l})$
are associated with location $s_{l}$. Let $\tilde{x}_{l}=x_{l}-\mathbf{S}%
_{xz}\mathbf{S}_{zz}^{-1}\mathbf{z}_{l}$ denote the residual from regressing 
$x_{l}$ on $\mathbf{z}_{l}$, where we use the notation $\mathbf{S}%
_{ab}=n^{-1}\sum_{l=1}^{n}\mathbf{a}_{l}\mathbf{b}_{l}^{\prime}$ for any
vectors $\mathbf{a}_{l}$ and $\mathbf{b}_{l}$. Suppose $\mathbf{S}_{\tilde{x}%
\tilde{x}}\overset{p}{\rightarrow}\sigma_{\tilde{x}\tilde{x}}^{2}>0$ and 
\begin{equation*}
n^{-1/2}\sum_{l=1}^{n}\tilde{x}_{l}\varepsilon_{l}|\mathbf{s}\Rightarrow_{p}%
\mathcal{N}(0,\sigma_{\tilde{x}\varepsilon}^{2}).
\end{equation*}
Then 
\begin{equation*}
\sqrt{n}(\hat{\beta}-\beta)|\mathbf{s}\Rightarrow_{p}\mathcal{N}%
(0,\sigma^{2})
\end{equation*}
where $\sigma^{2}=\sigma_{\tilde{x}\varepsilon}^{2}/\sigma_{\tilde{x}\tilde{x%
}}^{4}$. Spatial correlation affects inference in this model through $%
\sigma_{\tilde{x}\varepsilon}^{2}$ which incorporates potential correlation
between $\tilde{x}_{l}\varepsilon_{l}$ and $\tilde{x}_{\ell}\varepsilon_{%
\ell}$ at spatial locations $s_{l}$ and $s_{\ell}$.

Thus, suppose that $\tilde{x}_{l}\varepsilon _{l}$ satisfies the assumptions
previously made for $u_{l}$. Then a straightforward calculation shows that
setting 
\begin{equation*}
y_{l}=\hat{\beta}+\frac{\tilde{x}_{l}\hat{\varepsilon}_{l}}{%
n^{-1}\sum_{l=1}^{n}\tilde{x}_{l}^{2}}
\end{equation*}%
in the analysis of the previous sections leads to analogous results with $%
\beta $ replacing $\mu $ as the parameter of interest. The extension to GMM
inference is analogous; see, for instance, Section 4.4 of \cite{Mueller2020}.

\subsection{Computational issues}

We highlight two computational issues. The first involves the calculation of
the SCPC critical value, and the second involves the problem of computing
the eigenvectors $\mathbf{r}_{j}\ $of $\mathbf{M\Sigma }(c_{0})\mathbf{M}$
when $n$ is very large.

The critical value $\func{cv}=\func{cv}_{\text{SCPC}}(q)$ solves $%
\sup_{c\geq c_{0}}\mathbb{P}_{\mathbf{\Sigma }(c)}(\tau
^{2}(q^{-1}\sum_{j=1}^{q}\mathbf{r}_{j}\mathbf{r}_{j}^{\prime })>\func{cv}%
^{2})=\alpha $ or equivalently (from Theorem \ref{thm:tau_finite_q}) $%
\sup_{c\geq c_{0}}\mathbb{P}\left( Z_{0}^{2}>\sum_{i=1}^{q}\eta
_{i}Z_{i}^{2}\right) =\alpha $ where $\eta _{i}=-\omega _{i}/\omega _{0}$, $%
\omega _{i}$ are the eigenvalues of $\mathbf{\hat{W}}^{0\prime }\mathbf{%
\Sigma }(c)\mathbf{\hat{W}}^{0}\mathbf{D}(\func{cv})$ with $\mathbf{\hat{W}}%
^{0}=[\mathbf{l},\mathbf{r}_{1}/\sqrt{q},\ldots ,\mathbf{r}_{q}/\sqrt{q}]$
and $Z_{j}\sim $i.i.d.~$\mathcal{N}(0,1)$. \cite{Bakirov05} show that 
\begin{equation}
\mathbb{P}\left( Z_{0}^{2}\geq \sum_{i=1}^{q}\eta _{i}Z_{i}^{2}\right) =%
\frac{1}{\pi }\int_{0}^{1}\frac{x^{\frac{q-1}{2}}}{\sqrt{(1-x)%
\prod_{i=1}^{q}(x+\eta _{i})}}dx,  \label{eq:BS}
\end{equation}%
which is readily evaluated by numerical quadrature. Thus $\func{cv}_{\text{%
SCPC}}(q)$ can be obtained by combining a root-finder with a grid search
over $c\geq c_{0}$.

The second problem involves computing the eigenvectors $\mathbf{r}%
_{j}=(r_{j,1},\ldots ,r_{j,n})^{\prime }$ of the $n\times n$ matrix $\mathbf{%
M\Sigma }(c_{0})\mathbf{M}$ when $n$ is very large (say, larger than $n=2000$%
). Here we can leverage the eigenfunction convergence result in Lemma \ref%
{thm:eigenfunc} as discussed in Section \ref{sec:SCPCconv}: In the notation
defined there, we seek to approximate $\mathbf{r}_{j}=(\hat{\varphi}%
_{j}^{0}(s_{1}),\ldots ,\hat{\varphi}_{j}^{0}(s_{n}))^{\prime }$. Consider a
random subset of size $\tilde{n}<n$ of the observed locations $\{\tilde{s}%
_{l}\}_{l=1}^{\tilde{n}}\subset \{s_{l}\}_{l=1}^{n}$, and let $\mathbf{%
\tilde{\Sigma}}(c_{0})$ be the implied $\tilde{n}\times \tilde{n}$
covariance matrix of $(u(\tilde{s}_{1}),\ldots ,u(\tilde{s}_{n}))^{\prime }$
using the benchmark covariance function $\sigma _{u}^{0}(r-s|c_{0})=\exp
[-c_{0}||r-s||]$. Let the eigenvector corresponding to the $j$th largest
eigenvalue $\tilde{\lambda}_{j}$ of $\mathbf{\tilde{\Sigma}}(c_{0})$ be $%
\mathbf{\tilde{r}}_{j}=(\tilde{r}_{1,j},\ldots ,\tilde{r}_{\tilde{n}%
,j})^{\prime }$ with $\tilde{n}^{-1}\mathbf{\tilde{r}}_{j}^{\prime }\mathbf{%
\tilde{r}}_{j}=1$. As long as $\tilde{n}\rightarrow \infty $ and $\lambda
_{q+1}>\lambda _{q}$, Lemma \ref{thm:eigenfunc} implies that the span of the 
$\mathcal{S}\mapsto 
\mathbb{R}
$ functions 
\begin{equation*}
\tilde{\varphi}_{j}^{0}(s)=\tilde{n}^{-1}\tilde{\lambda}_{j}^{-1}\sum_{l=1}^{%
\tilde{n}}\tilde{r}_{j,l}\left( \exp [-c_{0}||s-\tilde{s}_{l}||]-\tilde{n}%
^{-1}\sum_{\ell =1}^{\tilde{n}}\exp [-c_{0}||\tilde{s}_{l}-\tilde{s}_{\ell
}||]\right) \text{, }j=1,\ldots ,q
\end{equation*}%
converges to the eigenspace spanned by $\varphi _{j}^{0}$, $j=1,\ldots ,q$,
just like the full sample estimators $\hat{\varphi}_{j}^{0}$. Thus, it is
formally justified to approximate the value of $\hat{\varphi}_{j}^{0}$ at
locations $\{s_{l}\}_{l=1}^{n}\ni s_{\ell }\notin \{\tilde{s}_{l}\}_{l=1}^{%
\tilde{n}}$ via $r_{j,\ell }=\hat{\varphi}_{j}^{0}(s_{\ell })\approx \tilde{%
\varphi}_{j}^{0}(s_{\ell })$---this is a version of the so-called Nystr\"{o}%
m method (see, for instance, \cite{Rasmussen05} for discussion and
references).

In practice, such approximations can be carried out for several random
subsets of $\tilde{n}$ locations, followed by a (sample) principle component
analysis to extract the best approximation to the space spanned by the first 
$q$ eigenvectors. The resulting algorithm has $O(n)$ running time (in
contrast to the $O(n^{2})$ running time of a basic implementation of \cite%
{Conley99}-type kernel estimators). We provide corresponding STATA and
Matlab code in the replication files.

\subsection{Extension to F-tests}

Consider the case where $\mathbf{y}_{l}=\mathbf{\mu }+\mathbf{u}_{l}$ with $%
\mathbf{y}_{l}$, $\mathbf{\mu }$ and $\mathbf{u}_{l}$ $m\times 1$ vectors,
and we seek to test the hypothesis $H_{0}:\mathbf{\mu }=\mathbf{\mu }_{0}$.
Suppose the observations conditional on $\mathbf{s}$ are generated by the
model 
\begin{equation*}
\mathbf{u}(s_{l})=\mathbf{B}(c_{n}s_{l})\text{, }l=1,\ldots ,n
\end{equation*}%
where $\mathbf{B}(s)$ is an $%
\mathbb{R}
^{m}$-valued mean-zero stationary random field on $%
\mathbb{R}
^{d}$ with covariance function $\mathbb{E}[\mathbf{B}(s)\mathbf{B}%
(r)^{\prime }]=\mathbf{\sigma }_{B}(r-s)$. Let $\mathbf{Y}$ and $\mathbf{U}$
be the $n\times m$ matrices of observations and innovations, respectively,
and\textbf{\ }$\mathbf{\bar{y}}=n^{-1}\sum_{l=1}^{n}\mathbf{y}_{l}$ the
sample mean. The natural analogue to the t-statistic $\tau ^{2}(\mathbf{\hat{%
W}\hat{W}}^{\prime })$ is Hotelling's-$T^{2}$ statistic 
\begin{equation}
T^{2}(\mathbf{\hat{W}\hat{W}}^{\prime })=n(\mathbf{\bar{y}-\mu }_{0}\mathbf{)%
}^{\prime }\left( \mathbf{Y}^{\prime }\mathbf{\hat{W}\hat{W}^{\prime }}%
\mathbf{Y}\right) ^{-1}(\mathbf{\bar{y}-\mu }_{0}\mathbf{).}
\label{eq:HotellingT2}
\end{equation}

One would expect that under mixing and moment conditions similar to those of
Lemma \ref{lm:NormalLimit_q} (ii) 
\begin{equation}
\limfunc{vec}(\mathbf{W}^{0\prime }\mathbf{U})\mathbf{|s}\Rightarrow _{p}%
\mathcal{N}\left( 0,a\mathbf{\sigma }_{B}(r-s)\otimes \mathbf{V}_{1}+\left[
\int \mathbf{\sigma }_{B}(s)ds\right] \otimes \mathbf{V}_{2}\right) .
\label{eq:multi_conv_q}
\end{equation}%
Note that $T^{2}(\mathbf{\hat{W}\hat{W}}^{\prime })$ is invariant to the
transformation $\mathbf{Y\rightarrow YH}$ for nonsingular $\mathbf{H}$. For
the purposes of studying the limit distribution of $T^{2}(q)$ under weak
correlation, it is thus without loss of generality to normalize $\mathbf{%
\sigma }_{B}(\cdot )$ such that the limit covariance matrix in (\ref%
{eq:multi_conv_q})\ becomes 
\begin{equation}
\limfunc{diag}(\mathbf{\kappa })\otimes \mathbf{V}_{1}+(\mathbf{I}_{m}-%
\limfunc{diag}(\mathbf{\kappa }))\otimes \mathbf{V}_{2}
\label{eq:multi_asycov_norm}
\end{equation}%
where $\mathbf{\kappa }$ is a $m\times 1$ vector with elements in $[0,1)$.

For the extension of the SCPC method, consider a benchmark model indexed by $%
\mathbf{c}=(c_{1},\ldots,c_{m})$ where $\limfunc{vec}(\mathbf{Y)|s}\sim%
\mathcal{N}(\mathbf{\mu}\otimes\mathbf{l}_{n},\mathbf{\Sigma(c))}$ with $%
\mathbf{\Sigma(c})=\limfunc{diag}(\mathbf{\Sigma}(c_{1}),\ldots,\mathbf{%
\Sigma}(c_{m}))$, and $\mathbf{\Sigma}(c)$ is as in Section \ref{sec:SCPC}.
Let $\mathbf{c}_{0}=c_{0}\mathbf{l}_{m}$, a $m\times1$ vector of identical
elements $c_{0}$. The SCPC test statistic $T_{\text{SCPC}}^{2}(q)$ is a
special case of (\ref{eq:HotellingT2})\ with the columns of $\hat{\mathbf{W}}
$ equal to the first $q$ eigenvectors of $\mathbf{\Sigma}(c_{0})$, scaled to
have length $1/\sqrt{q}$, and with critical value $\func{cv}_{\text{SCPC}%
}^{T}$ chosen to satisfy 
\begin{equation*}
\sup_{\mathbf{c}\geq\mathbf{c}_{0}}\mathbb{P}_{\mathbf{\Sigma}(\mathbf{c}%
)}^{0}(T_{\text{SCPC}}^{2}(q)>\func{cv}_{\text{SCPC}}^{T}(q)|\mathbf{s}%
)=\alpha,
\end{equation*}
under the null hypothesis, where $\mathbf{c}\geq\mathbf{c}_{0}$ is
understood as an elementwise inequality. The value of $q$ that minimizes the
expected volume of the confidence ellipsoid under $\limfunc{vec}(\mathbf{Y)|s%
}\sim\mathcal{N}(\mathbf{\mu}\otimes\mathbf{l},\mathbf{I}_{m}\mathbf{\otimes
I}_{n}\mathbf{)}$ is 
\begin{equation*}
\min_{q\geq m}\mathbb{E}[\func{vol}\{\mathbf{m}:\mathbf{m}^{\prime}(q^{-1}%
\mathbf{S}_{q})^{-1}\mathbf{m}\leq n^{-1}\func{cv}_{\text{SCPC}%
}^{T}(q)\}]=\min_{q\geq m}\frac{(2\pi\func{cv}_{\text{SCPC}%
}^{T}(q)/n)^{m/2}\Gamma((q+1)/2)}{\sqrt{q}\Gamma((q-m+1)/2)\Gamma(m/2+1)}
\end{equation*}
where $\mathbf{S}_{q}$ is distributed Wishart with $q$ degrees of freedom,
and the equality follows from Bartlett's decomposition of a Wishart random
matrix, and the formulas for the expectation of a $\chi$ random variable and
the volume of an $m$ dimensional ellipsoid.

Since appropriate choices of $c_{j,n}\rightarrow \infty $, $j=1,\ldots ,m$
in the benchmark model can replicate the normalized limit distributions (\ref%
{eq:multi_asycov_norm}) for all $\mathbf{\kappa }$, by the same arguments
that lead to Theorem \ref{thm:cv_n}, $T_{\text{SCPC}}^{2}(q)$ controls size
under all weak correlation patterns that induce (\ref{eq:multi_conv_q}). And
as in Section \ref{sec:gmm}, it is straightforward to adapt $T_{\text{SCPC}%
}^{2}(q)$ to test $m$ restrictions in linear regression and GMM problems. We
omit details for brevity. Generalizing the results about the small sample
robustness of $\tau _{\text{SCPC}}$ under potentially strong correlations in
Theorem \ref{thm:robust} to $T_{\text{SCPC}}^{2}$ is interesting but
challenging, and beyond the scope of this paper.\newpage

\appendix

\section{\textbf{Appendix}}

\begin{lemma}
\label{lem:weak_in_p}If $\mathbf{X}_{n}|\mathbf{s}_{n}\Rightarrow _{p}%
\mathbf{X}$ and $\mathbf{Y}_{n}\overset{p}{\rightarrow }0$, then $(\mathbf{X}%
_{n}+\mathbf{Y}_{n})|\mathbf{s}_{n}\Rightarrow _{p}\mathbf{X}$.
\end{lemma}

\begin{proof}
Let $\func{BL}$ be the space of Lipschitz continuous functions $%
\mathbb{R}
^{p}\mapsto 
\mathbb{R}
$ bounded by one with unit Lipschitz constant. By \cite{Berti2006}, page 93, 
$\mathbf{X}_{n}|\mathbf{s}_{n}\Rightarrow _{p}\mathbf{X}$ is equivalent to $%
\sup_{h\in \func{BL}}|\mathbb{E}[h(\mathbf{X}_{n})-h(\mathbf{X})|\mathbf{s}%
_{n}]|\overset{p}{\rightarrow }0$, so it suffices to show that $\sup_{h\in 
\func{BL}}|\mathbb{E}[h(\mathbf{X}_{n}+\mathbf{Y}_{n})-h(\mathbf{X})|\mathbf{%
s}_{n}]|\overset{p}{\rightarrow }0$. Let $\mathbf{Y}_{n}^{\ast }=\mathbf{Y}%
_{n}\mathbf{1}[||\mathbf{Y}_{n}||\leq 1]$, so that 
\begin{equation*}
\sup_{h\in \func{BL}}|\mathbb{E}[h(\mathbf{X}_{n}+\mathbf{Y}_{n})-h(\mathbf{X%
})|\mathbf{s}_{n}]|\leq \sup_{h\in \func{BL}}|\mathbb{E}[h(\mathbf{X}_{n}+%
\mathbf{Y}_{n}^{\ast })-h(\mathbf{X})|\mathbf{s}_{n}]|+2\mathbb{P}(||\mathbf{%
Y}_{n}^{\ast }||>1|\mathbf{s}_{n}).
\end{equation*}%
Note that with $\Delta _{n}(h)=h(\mathbf{X}_{n}+\mathbf{Y}_{n}^{\ast })-h(%
\mathbf{X}_{n})$, $|\Delta _{n}(h)|\leq ||\mathbf{Y}_{n}^{\ast }||$ a.s. for
all $h\in \func{BL}$, so that 
\begin{eqnarray*}
\sup_{h\in \func{BL}}|\mathbb{E}[h(\mathbf{X}_{n}+\mathbf{Y}_{n}^{\ast })-h(%
\mathbf{X})|\mathbf{s}_{n}]| &=&\sup_{h\in \func{BL}}|\mathbb{E}[\Delta
_{n}(h)+h(\mathbf{X}_{n})-h(\mathbf{X})|\mathbf{s}_{n}]| \\
&\leq &\sup_{h\in \func{BL}}\left( |\mathbb{E}[\Delta _{n}(h)|\mathbf{s}%
_{n}]|+|\mathbb{E}[h(\mathbf{X}_{n})-h(\mathbf{X})|\mathbf{s}_{n}]|\right) 
\\
&\leq &\mathbb{E}[||\mathbf{Y}_{n}^{\ast }|||\mathbf{s}_{n}]+\sup_{h\in 
\func{BL}}|\mathbb{E}[h(\mathbf{X}_{n})-h(\mathbf{X})|\mathbf{s}_{n}]|.
\end{eqnarray*}%
We are left to show that $\mathbf{Y}_{n}\overset{p}{\rightarrow }0$ implies $%
\mathbb{P}(||\mathbf{Y}_{n}^{\ast }||>1|\mathbf{s}_{n})\overset{p}{%
\rightarrow }0$ and $\mathbb{E}[||\mathbf{Y}_{n}^{\ast }|||\mathbf{s}_{n}]%
\overset{p}{\rightarrow }0.$

Consider the latter claim. Suppose otherwise. Then for some $\varepsilon >0$%
, and some subsequence $n^{\prime }$ of $n$, $\lim_{n^{\prime }\rightarrow
\infty }\mathbb{P}(\mathbb{E}[||\mathbf{Y}_{n^{\prime }}^{\ast }|||\mathbf{s}%
_{n^{\prime }}]>\varepsilon )>\varepsilon $, so that $\lim \inf_{n^{\prime
}\rightarrow \infty }\mathbb{E}[||\mathbf{Y}_{n^{\prime }}^{\ast
}||]>\varepsilon ^{2}$. But since $\mathbf{Y}_{n}^{\ast }$ is bounded, $%
\mathbf{Y}_{n}\overset{p}{\rightarrow }0$ implies $\lim_{n\rightarrow \infty
}\mathbb{E}[||\mathbf{Y}_{n}^{\ast }||]=0$, a contradiction. A similar
argument yields $\mathbb{E}[||\mathbf{Y}_{n}^{\ast }|||\mathbf{s}_{n}]%
\overset{p}{\rightarrow }0$, concluding the proof. 
\end{proof}

\bigskip

\textbf{Proof of Lemma \ref{lm:NormalLimit_q}: }(i) Since $B$ is Gaussian, $%
n^{-1}\mathbf{W}_{n}^{0\prime }\mathbf{u}_{n}|\mathbf{s}_{n}\sim \mathcal{N}%
(0,\mathbf{\Omega }_{n})$ with $\mathbf{\Omega }_{n}=n^{-2}\sum_{l,\ell }%
\mathbf{w}^{0}(s_{l})\mathbf{w}^{0}(s_{\ell })^{\prime }\sigma _{B}(c\left(
s_{l}-s_{\ell }\right) )$. It thus suffices to show that $\mathbf{\Omega }%
_{n}\overset{p}{\rightarrow }\mathbf{\Omega }_{sc}$.

We have $\mathbf{\Omega }_{n}=\sigma _{B}(0)n^{-2}\sum_{l}\mathbf{w}%
^{0}(s_{l})\mathbf{w}^{0}(s_{l})^{\prime }+n^{-2}\sum_{l\neq \ell }\mathbf{w}%
^{0}(s_{l})\mathbf{w}^{0}(s_{\ell })^{\prime }\sigma _{B}(c\left(
s_{l}-s_{\ell }\right) )$, and $||n^{-2}\sum_{l}\mathbf{w}^{0}(s_{l})\mathbf{%
w}^{0}(s_{l})^{\prime }||\leq n^{-1}\sup_{s\in \mathcal{S}}||\mathbf{w}%
^{0}(s)||^{2}\rightarrow 0.$ Furthermore, 
\begin{equation*}
\mathbb{E}\left[ \frac{1}{n(n-1)}\sum_{l\neq \ell }\mathbf{w}^{0}(s_{l})%
\mathbf{w}^{0}(s_{\ell })^{\prime }\sigma _{B}(c\left( s_{l}-s_{\ell
}\right) )\right] =\mathbb{E}[\mathbf{w}^{0}(s_{1})\mathbf{w}%
^{0}(s_{2})^{\prime }\sigma _{B}(c\left( s_{1}-s_{2}\right) )]=\mathbf{%
\Omega }_{sc}
\end{equation*}%
and with $w_{i}^{0}(s)$ the $i$th element of $\mathbf{w}^{0}(s)$, 
\begin{multline*}
\mathbb{E}\left[ \left( \frac{1}{n(n-1)}\sum_{l\neq \ell
}w_{i}^{0}(s_{l})w_{j}^{0}(s_{\ell })^{\prime }\sigma _{B}(c\left(
s_{l}-s_{\ell }\right) )\right) ^{2}\right] \\
=\frac{(n-2)(n-3)}{n(n-1)}\mathbb{E}[w_{i}^{0}(s_{1})w_{j}^{0}(s_{2})^{%
\prime }\sigma _{B}(c\left( s_{1}-s_{2}\right) )]\mathbb{E}%
[w_{i}^{0}(s_{3})w_{j}^{0}(s_{4})^{\prime }\sigma _{B}(c\left(
s_{3}-s_{4}\right) )] \\
+\frac{4(n-2)}{n(n-1)}\mathbb{E}[w_{i}^{0}(s_{1})w_{j}^{0}(s_{2})^{\prime
}\sigma _{B}(c\left( s_{1}-s_{2}\right)
)w_{i}^{0}(s_{1})w_{j}^{0}(s_{3})^{\prime }\sigma _{B}(c\left(
s_{1}-s_{3}\right) )] \\
+\frac{2}{n(n-1)}\mathbb{E}[w_{i}^{0}(s_{1})w_{j}^{0}(s_{2})^{\prime }\sigma
_{B}(c\left( s_{1}-s_{2}\right) )w_{i}^{0}(s_{1})w_{j}^{0}(s_{2})^{\prime
}\sigma _{B}(c\left( s_{1}-s_{2}\right) )]
\end{multline*}%
so that $\limfunc{Var}[\frac{1}{n(n-1)}\sum_{l\neq \ell
}w_{i}^{0}(s_{l})w_{j}^{0}(s_{\ell })^{\prime }\sigma _{B}(c\left(
s_{l}-s_{\ell }\right) )]=O(n^{-1})$, and therefore $\mathbf{\Omega }_{n}%
\overset{p}{\rightarrow }\mathbf{\Omega }_{sc}$.

(ii) Follows from Theorem 3.2 in \cite{Lahiri_2003} and the Cram\'{e}r-Wold
device. $\blacksquare $

\bigskip{}

\textbf{Proof of Theorem \ref{thm:tau_finite_q}: }In the notation of Lemma %
\ref{lm:NormalLimit_q}, with $\mathbf{X}=(X_{0},\mathbf{X}_{1:q}^{\prime
})^{\prime }$ and $\mathbf{Z}=(Z_{0},\ldots ,Z_{q})^{\prime }$ we have 
\begin{eqnarray*}
\mathbb{P}\left( \tau _{n}^{2}(\mathbf{W}_{n}\mathbf{W}_{n}^{\prime })>\func{%
cv}^{2}|\mathbf{s}_{n}\right) &\overset{p}{\rightarrow }&\mathbb{P}\left( 
\frac{X_{0}^{2}}{\mathbf{X}_{1:q}^{\prime }\mathbf{X}_{1:q}}>\func{cv}%
^{2}\right) \\
&=&\mathbb{P}\left( X_{0}^{2}-\func{cv}^{2}\mathbf{X}_{1:q}^{\prime }\mathbf{%
X}_{1:q}>0\right) \\
&=&\mathbb{P}\left( \mathbf{X}^{\prime }\mathbf{D}(\func{cv})\mathbf{X}%
>0\right) \\
&=&\mathbb{P}(\mathbf{Z}^{\prime }\mathbf{\Omega }^{1/2}\mathbf{D}(\func{cv})%
\mathbf{\Omega }^{1/2}\mathbf{Z}>0) \\
&=&\mathbb{P}\left( \sum_{i=0}^{q}\omega _{i}Z_{i}^{2}>0\right)
\end{eqnarray*}%
where the convergence follows from Lemma \ref{lm:NormalLimit_q} and the
continuous mapping theorem, and the last equality follows by similarity of
the matrices $\mathbf{\Omega }^{1/2}\mathbf{D}(\func{cv})\mathbf{\Omega }%
^{1/2}$ and $\mathbf{D}(\func{cv})\mathbf{\Omega }$. The claim about the
sign of the eigenvalues follows from Lemma \ref{lm:lam1Abar} below. $%
\blacksquare $

\bigskip{}

\textbf{Proof of Theorem \ref{thm:what}: }We show that Lemma \ref%
{lm:NormalLimit_q} (i) and (ii) continue to hold with $\mathbf{w}^{0}$
replaced by $\mathbf{\hat{w}}^{0}$. We have 
\begin{equation*}
\mathbb{E}\left[ \left( \sum_{l=1}^{n}(\hat{w}%
_{i}^{0}(s_{l})-w_{i}^{0}(s_{l}))u(s_{l})\right) ^{2}\left\vert \mathbf{s}%
_{n}\right. \right] \leq \sup_{s\in \mathcal{S}}|\hat{w}%
_{i}^{0}(s)-w_{i}^{0}(s)|^{2}\sum_{l,\ell }|\sigma _{B}(c_{n}\left(
s_{l}-s_{\ell }\right) )|
\end{equation*}%
almost surely. Proceeding as in the proof of Lemma \ref{lm:NormalLimit_q}
(i) now shows that $\mathbb{E}[n^{-2}\sum_{l,\ell }|\sigma _{B}(c\left(
s_{l}-s_{\ell }\right) )|]=\int \int |\sigma _{B}(c(r-s))|g(r)g(s)drds$, so $%
n^{-2}\sum_{l,\ell }|\sigma _{B}(c\left( s_{l}-s_{\ell }\right) )|=O_{p}(1)$%
. Similarly, under the assumptions of part (ii) of Lemma \ref%
{lm:NormalLimit_q}, proceeding as in the proof of Lemma 5.2 of Lahiri (2003)
yields $\mathbb{E}[a_{n}n^{-1}\sum_{l,\ell }|\sigma _{B}(c_{n}\left(
s_{l}-s_{\ell }\right) )|]\rightarrow a\sigma _{u}^{2}+\int_{\mathbb{R}%
^{d}}|\sigma _{B}(s)|ds\int g(s)^{2}ds$. The result thus follows from (\ref%
{eq:what_conv}) and Lemma \ref{lem:weak_in_p}.

\bigskip

The proof of Theorem \ref{thm:kernel_conv} requires a slightly more general
version of Theorem \ref{thm:what}.

\begin{lemma}
\label{lem:block_what}In the notation of Lemma \ref{thm:eigenfunc}, suppose $%
\mathbf{\hat{W}}=\mathbf{\hat{L}\hat{\Phi}}$, where the $i$th column of the $%
n\times q$ matrix $\mathbf{\hat{\Phi}}$ is $\mathbf{\hat{v}}_{i}=(\hat{%
\varphi}_{i}(s_{1}),\ldots,\hat{\varphi}_{i}(s_{n}))^{\prime}$ and $\mathbf{%
\hat{L}}=\limfunc{diag}(\hat{\lambda}_{1},\ldots,\hat{\lambda}_{q})$. Under
the assumptions of Lemma \ref{lm:NormalLimit_q}, $c_{n}^{d}n^{-2}(\mathbf{u}%
^{\prime}\mathbf{\hat{W}\hat{W}}^{\prime}\mathbf{u-u}^{\prime}\mathbf{WW}%
^{\prime}\mathbf{u})|\mathbf{s}_{n}\overset{p}{\rightarrow}0$, where $%
\mathbf{W}=\mathbf{L\Phi}$, $\mathbf{L}=\limfunc{diag}(\lambda_{1}\mathbf{l}%
_{m_{1}},\ldots,\lambda_{p}\mathbf{l}_{m_{p}})$ and the $i$th column of $%
\mathbf{\Phi}$ is equal to $(\varphi_{i}(s_{1}),\ldots,\varphi_{i}(s_{n}))^{%
\prime}$.
\end{lemma}

\begin{proof}
With $\mathbf{\hat{O}=}\limfunc{diag}(\mathbf{\hat{O}}_{(1)},\ldots ,\mathbf{%
\hat{O}}_{(p)})$, 
\begin{eqnarray*}
c_{n}^{d}n^{-2}\mathbf{u}^{\prime }\mathbf{\hat{\Phi}\hat{L}}^{2}\mathbf{%
\hat{\Phi}}^{\prime }\mathbf{u} &=&c_{n}^{d}n^{-2}\mathbf{u}^{\prime }%
\mathbf{\hat{\Phi}\mathbf{\hat{O}\hat{O}}}^{\prime }\mathbf{\hat{L}}^{2}%
\mathbf{\hat{O}}^{\prime }\mathbf{\hat{O}\hat{\Phi}}^{\prime }\mathbf{u} \\
&=&c_{n}^{d}n^{-2}\mathbf{u}^{\prime }\mathbf{\Phi \mathbf{\hat{O}}}^{\prime
}\mathbf{\hat{L}}^{2}\mathbf{\hat{O}}^{\prime }\mathbf{\Phi }^{\prime }%
\mathbf{u}+o_{p}(1) \\
&=&c_{n}^{d}n^{-2}\mathbf{u}^{\prime }\mathbf{\Phi \mathbf{\hat{O}}}^{\prime
}\mathbf{L}^{2}\mathbf{\hat{O}}^{\prime }\mathbf{\Phi }^{\prime }\mathbf{u}%
+o_{p}(1) \\
&=&c_{n}^{d}n^{-2}\mathbf{u}^{\prime }\mathbf{\Phi L}^{2}\mathbf{\Phi }%
^{\prime }\mathbf{u}+o_{p}(1)
\end{eqnarray*}%
where the first line follows from $\mathbf{\hat{O}}^{\prime }\mathbf{\hat{O}%
=I}_{q}$, the second from Lemma \ref{thm:eigenfunc} (a) and (b) and the
reasoning in the proof of Theorem \ref{thm:what}, the third from Lemma \ref%
{thm:eigenfunc} (b) and $||c_{n}^{d/2}n^{-1}\mathbf{\hat{O}}^{\prime }%
\mathbf{\Phi }^{\prime }\mathbf{u}||\leq ||\mathbf{\hat{O}}||\cdot
||c_{n}^{d/2}n^{-1}\mathbf{\Phi }^{\prime }\mathbf{u}||=O_{p}(1)$ using
Lemma \ref{lm:NormalLimit_q}, and the fourth from $\mathbf{\mathbf{\hat{O}}}%
^{\prime }\mathbf{L}^{2}\mathbf{\hat{O}}^{\prime }=\mathbf{L}^{2}$ a.s. The
result now follows from Lemma \ref{lem:weak_in_p}.
\end{proof}

{}

\textbf{Proof of Theorem \ref{thm:kernel_conv}: }For the first claim, by
Theorem 4.4.6 of \cite{Vasudeva2017}, $\omega _{0}=\sup_{||f||=1}\langle
f,RTRf\rangle $, so it suffices to show that for some $f\in \mathcal{L}%
_{G}^{2}$, $\langle f,RTRf\rangle >0$. In the weak correlation case, this
holds for $f(s)=(\kappa +(1-\kappa )g(s))^{-1/2}$, since $\langle
f,R_{wc}TR_{wc}f\rangle =\langle 1,T1\rangle =\int \int (1-\bar{k}%
(r,s))dG(r)dG(s)=1$. In the strong correlation case, the same conclusion
holds by setting $f$ such that $R_{sc}f=1$. Such an $f$ exists, because the
kernel of $R_{sc}^{2}$ is equal to $\{0\}$ by assumption about $\sigma _{B}$%
, so the range of $R_{sc}$ is $\mathcal{L}_{G}^{2}\backslash \{0\}$ by
Theorem 3.5.8 of \cite{Vasudeva2017}.

Under the null hypothesis, $\mathbb{P}(\tau _{n}^{2}(\overline{\mathbf{K}}%
_{n})>\func{cv}^{2}|\mathbf{s}_{n})=\mathbb{P}(\hat{\xi}_{n}>0|\mathbf{s}%
_{n})$, where $\hat{\xi}_{n}=c_{n}^{d}n^{-2}\sum_{l,\ell }u_{l}u_{\ell }(1-%
\func{cv}^{2}\hat{k}_{n}(s_{l},s_{\ell }))$. By construction of $\hat{\lambda%
}_{i}$ and $\hat{\varphi}_{i}(\cdot )$ in Lemma \ref{thm:eigenfunc}, for all 
$1\leq l,\ell \leq n$, 
\begin{equation*}
\hat{k}_{n}(s_{l},s_{\ell })=\sum_{i=1}^{n}\hat{\lambda}_{i}\hat{\varphi}%
_{i}(s_{l})\hat{\varphi}_{i}(s_{\ell }).
\end{equation*}%
For a given $q$ satisfying the assumption of Lemma \ref{thm:eigenfunc}, and
all $n>q$, let 
\begin{equation*}
\hat{k}_{n,q}(r,s)=\sum_{i=1}^{q}\hat{\lambda}_{i}\hat{\varphi}_{i}(r)\hat{%
\varphi}_{i}(s)
\end{equation*}%
and $\hat{\xi}_{n}^{q}=c_{n}^{d}n^{-2}\sum_{l,\ell }u_{l}u_{\ell }(1-\func{cv%
}^{2}\hat{k}_{n,q}(s_{l},s_{\ell })).$ We now show the last claim, that is $%
\mathbb{P}(\hat{\xi}_{n}>0|\mathbf{s}_{n})\overset{p}{\rightarrow }\QTR{up}{%
\mathbb{P}(\sum_{i=0}^{\infty }\omega _{i}Z_{i}^{2}>0)}$, which is implied
by the following three claims 
\begin{eqnarray}
&&\text{(i) for any }\left. \varepsilon >0\text{\ \ }\lim_{q\rightarrow
\infty }\limsup_{n\rightarrow \infty }\mathbb{P}(|\hat{\xi}_{n}-\hat{\xi}%
_{n}^{q}|>\varepsilon )=0\right.  \label{eq:c1} \\
&&\text{(ii) for any fixed }q\text{, }\left. \mathbb{P}(\hat{\xi}_{n}^{q}>0|%
\mathbf{s}_{n})\overset{p}{\rightarrow }\mathbb{\mathbb{P}}\left(
\sum_{i=0}^{q}\omega _{q,i}Z_{i}^{2}>0\right) \right.  \label{eq:c2} \\
&&\text{(iii) }\left. \lim_{q\rightarrow \infty }\mathbb{\mathbb{P}}\left(
\sum_{i=0}^{q}\omega _{q,i}Z_{i}^{2}>0\right) =\QTR{up}{\mathbb{P}\left( 
\QTR{up}{\sum_{i=0}^{\infty }\omega _{i}Z_{i}^{2}>0}\right) }\right.
\label{eq:c3}
\end{eqnarray}%
for some double array of real numbers $\omega _{q,i}$ by invoking Lemma \ref%
{lem:weak_in_p}.

For claim (i), note that for all $n>q$, $\hat{\xi}_{n}\leq \hat{\xi}_{n}^{q}$
a.s., and 
\begin{eqnarray*}
\mathbb{E}[\hat{\xi}_{n}^{q}-\hat{\xi}_{n}|\mathbf{s}_{n}]
&=&c_{n}^{d}n^{-2}\sum_{l,\ell }\sigma _{B}(c_{n}(s_{l}-s_{\ell }))\left(
\sum_{i=q+1}^{n}\hat{\lambda}_{i}\hat{\varphi}_{i}(s_{l})\hat{\varphi}%
_{i}(s_{\ell })\right) \\
&\leq &\hat{\lambda}_{q+1}c_{n}^{d}n^{-2}\sum_{l,\ell }\sigma
_{B}(c_{n}(s_{l}-s_{\ell }))
\end{eqnarray*}%
where the inequality follows from $\limfunc{tr}(\mathbf{AB})\leq \lambda
_{1}(\mathbf{A})\limfunc{tr}\mathbf{B}$ for positive semidefinite matrices $%
\mathbf{A},\mathbf{B}$ and $\lambda _{1}(\mathbf{A})$ the largest eigenvalue
of $\mathbf{A}$. By the same reasoning as employed in Theorem \ref{thm:what}%
, $c_{n}^{d}n^{-2}\sum_{l,\ell }\sigma _{B}(c_{n}(s_{l}-s_{\ell }))=O_{p}(1)$%
. Furthermore, by Lemma \ref{thm:eigenfunc} (b), $|\hat{\lambda}%
_{q+1}-\lambda _{q+1}|=O_{q}(n^{-1/2})$, and $\lim_{q\rightarrow \infty
}\lambda _{q}=0$. Thus (\ref{eq:c1}) follows.

For claim (ii), let $\varphi _{0}(s)=1$ and $\lambda _{0}=1$. By Lemma \ref%
{thm:eigenfunc} (a), Lemma \ref{lem:block_what} and Theorem \ref%
{thm:tau_finite_q}, claim (\ref{eq:c2})\ holds, where $\omega _{q,i}$ are
the eigenvalues of $\mathbf{D}(\func{cv})\mathbf{\Omega }$ for $\mathbf{%
\Omega \in \{\Omega }_{sc}\mathbf{,\Omega }_{wc}\mathbf{\},}$ and the $%
(i+1),(j+1)$ element of $\mathbf{\Omega }$ is equal to $\sqrt{\lambda
_{i}\lambda _{j}}\int \int \varphi _{i}(s)\sigma _{B}(c(r-s))\varphi
_{j}(r)dG(s)dG(r)$ and $\sqrt{\lambda _{i}\lambda _{j}}\int \varphi
_{i}(s)\varphi _{j}(s)(\kappa +(1-\kappa )g(s))ds$ under strong and weak
correlation, respectively.

For claim (iii), we first show that these $\omega _{q,i}$ are also the
eigenvalues of the finite rank self-adjoint linear operators $RT_{q}R$, $%
R\in \{R_{sc},R_{wc}\}$. To this end, let $\varphi _{i}^{\ast }(s)=\sqrt{%
\lambda _{i}}R\varphi _{i}(s)$. With $d_{0}=1$ and $d_{i}=-\func{cv}^{2}$,
we have 
\begin{equation*}
RT_{q}R(f)(s)=\int \left( \sum_{i=0}^{q}d_{i}\varphi _{i}^{\ast }(s)\varphi
_{i}^{\ast }(r)\right) f(r)dG(r)
\end{equation*}%
and the $(i+1),(j+1)$ element of $\mathbf{\Omega }$ stated above is equal to 
$\sqrt{\lambda _{i}\lambda _{j}}\langle \varphi _{i},R^{2}\varphi
_{j}\rangle =\sqrt{\lambda _{i}\lambda _{j}}\langle R\varphi _{i},R\varphi
_{j}\rangle =\int \varphi _{i}^{\ast }(s)\varphi _{j}^{\ast }(s)dG(s)$. Let $%
\mathbf{v}=(v_{0},\ldots ,v_{q})^{\prime }$ be an eigenvector of $\mathbf{D}(%
\func{cv})\mathbf{\Omega }$ corresponding to eigenvalue $\omega $, $\mathbf{D%
}(\func{cv})\mathbf{\Omega v}=\omega \mathbf{v}$. Then $\mathbf{D}(\func{cv})%
\mathbf{\Omega v}=\omega \mathbf{v}$ implies 
\begin{equation*}
\int \left( 
\begin{array}{ccc}
\varphi _{0}^{\ast }(r)\varphi _{0}^{\ast }(r) & \cdots & \varphi _{q}^{\ast
}(r)\varphi _{0}^{\ast }(r) \\ 
-\func{cv}^{2}\varphi _{0}^{\ast }(r)\varphi _{1}^{\ast }(r) & \cdots & -%
\func{cv}^{2}\varphi _{q}^{\ast }(r)\varphi _{1}^{\ast }(r) \\ 
\vdots & \ddots & \vdots \\ 
-\func{cv}^{2}\varphi _{0}^{\ast }(r)\varphi _{q}^{\ast }(r) & \cdots & -%
\func{cv}^{2}\varphi _{q}^{\ast }(r)\varphi _{q}^{\ast }(r)%
\end{array}%
\right) dG(r)\mathbf{v}=\omega \mathbf{v.}
\end{equation*}%
Premultiplying both sides of this equation by $(\varphi _{0}^{\ast
}(s),\ldots ,\varphi _{q}^{\ast }(s))$ yields 
\begin{eqnarray}
\sum_{j=0}^{q}\sum_{i=0}^{q}v_{j}\varphi _{i}^{\ast }(s)\int d_{i}\varphi
_{j}^{\ast }(r)\varphi _{i}^{\ast }(r)dG(r) &=&\omega
\sum_{j=0}^{q}v_{j}\varphi _{j}^{\ast }(s)  \notag \\
\int \left( \sum_{i=0}^{q}d_{i}\varphi _{i}^{\ast }(s)\varphi _{i}^{\ast
}(r)\right) \left( \sum_{j=0}^{q}v_{j}\varphi _{j}^{\ast }(r)\right) dG(r)
&=&\omega \sum_{j=0}^{q}v_{j}\varphi _{j}^{\ast }(s)
\label{eq:op_equivalence}
\end{eqnarray}%
so $\sum_{j=0}^{q}v_{j}\varphi _{j}^{\ast }(r)$ is an eigenvector of $%
RT_{q}R $ with eigenvalue $\omega $, and since the kernel of $RT_{q}R$
contains all functions that are orthogonal to $\{\varphi _{i}^{\ast
}\}_{i=0}^{q}$, these are the only nonzero eigenvalues.

Now let $\omega _{q,i}^{\Delta }$ be the eigenvalues of the self-adjoint
linear operator $R(T-T_{q})R$. By \cite{Kato1987} (also see the development
on page 911 of \cite{Rosasco2010}), there is an enumeration of the
eigenvalues $\omega _{q,i}$ such that 
\begin{equation}
\sum_{i=0}^{\infty }(\omega _{q,i}-\omega _{i})^{2}\leq \sum_{i=0}^{\infty
}(\omega _{q,i}^{\Delta })^{2}=||R(T-T_{q})R||_{HS}  \label{eq:eval_diff}
\end{equation}%
where $||R(T-T_{q})R||_{HS}$ is the Hilbert-Schmidt norm on the operator $%
R(T-T_{q})R:\mathcal{L}_{G}^{2}\mapsto \mathcal{L}_{G}^{2}$ induced by the
norm $\sqrt{\langle f,f\rangle }$. Now $||R(T-T_{q})R||_{HS}\leq
||R||^{2}\cdot ||T-T_{q}||_{HS}$ (cf. (\ref{HS_op_normineq}) below), and
since $T-T_{q}$ is an integral operator, $||T-T_{q}||_{HS}=\int \int \left(
\sum_{i=q+1}^{\infty }\lambda _{i}\varphi _{i}(s)\varphi _{j}(s)\right)
^{2}dG(s)dG(r)$. By Mercer's Theorem, this converges to zero as $%
q\rightarrow \infty $, so that 
\begin{equation}
\lim_{q\rightarrow \infty }\sum_{i=0}^{\infty }(\omega _{q,i}-\omega
_{i})^{2}=0\text{.}  \label{eq:diff_om_conv}
\end{equation}%
Thus using the same order of eigenvalues as in (\ref{eq:eval_diff}), we also
have $\limfunc{Var}[\sum_{i=0}^{q}\omega _{q,i}Z_{i}^{2}-\sum_{i=0}^{\infty
}\omega _{i}Z_{i}^{2}]\leq 2\sum_{i=0}^{\infty }(\omega _{q,i}-\omega
_{i})^{2}$, with the right-hand side converging to zero as $q\rightarrow
\infty $ by (\ref{eq:diff_om_conv}). But mean-square convergence implies
convergence in distribution, and (\ref{eq:c3}) follows.

For the second claim of the theorem, by Lemma \ref{lm:NormalLimit_q}, $%
\omega_{q,i}\leq0$ for $i\geq1$, which in conjunction with (\ref%
{eq:diff_om_conv}) implies $\omega_{i}\leq0$ for $i\geq1$. $\blacksquare$

\bigskip{}

\textbf{Proof of Lemma \ref{thm:eigenfunc}: }We initially show a weaker
claim than part (a), namely that there exists a sequence of $q\times q$
rotation matrices $\mathbf{\hat{O}}_{n}=\mathbf{\hat{O}}_{n}(\mathbf{s}_{n})$
with elements $\hat{O}_{n,ij}$ such that 
\begin{equation}
\max_{i\leq q}\sup_{s\in \mathcal{S}}\left\vert \varphi
_{i}(s)-\sum_{j=1}^{q}\hat{O}_{n,ij}\hat{\varphi}_{i}(s)\right\vert
=O_{p}(n^{-1/2}).  \label{eq:rot_weak}
\end{equation}

The proof follows closely the development in \cite{Rosasco2010}, denoted RBV
in the following. Let $k_{0}(r,s)=\bar{k}(r,s)+1.$ Conditional on $\mathbf{s}%
_{n}$, define the linear operators $\mathcal{L}_{G}^{2}\mapsto\mathcal{L}%
_{G}^{2}$ 
\begin{eqnarray*}
M(f)(s) & = & f(s)-\int f(r)dG(r) \\
M_{n}(f)(s) & = & f(s)-\int f(r)dG_{n}(r) \\
L(f)(s) & = & \int k_{0}(r,s)f(r)dG(r) \\
L_{n}(f)(s) & = & \int k_{0}(r,s)f(r)dG_{n}(r)
\end{eqnarray*}
and the derived operators $\bar{L}=MLM$, $\bar{L}_{n}=ML_{n}M$ and $\hat{L}%
_{n}=M_{n}L_{n}M_{n}$, so that $\bar{L}(f)(s)=\int f(r)\bar{k}(r,s)dG(r)$, $%
\bar{L}_{n}(f)(s)=\int\bar{k}(r,s)f(r)dG_{n}(r)$ and $\hat{L}_{n}(f)(s)=\int%
\hat{k}_{n}(r,s)f(r)dG_{n}(r)$, where $G_{n}$ is the empirical distribution
of $\{s_{l}\}_{l=1}^{n}$.

Let $\mathcal{H}\subset\mathcal{L}_{G}^{2}$ be the Reproducing Kernel
Hilbert Space (RKHS) of functions $f:\mathcal{S}\mapsto%
\mathbb{R}%
$ with kernel $k_{0}$ and inner product $\langle\cdot,\cdot\rangle_{\mathcal{%
H}}$ satisfying 
\begin{equation*}
\langle f,k_{0}(\cdot,r)\rangle_{\mathcal{H}}=f(r)
\end{equation*}
and associated norm $||f||_{\mathcal{H}}$. Let $K=\sup_{s\in\mathcal{S}%
}k_{0}(s,s).$ Define $\mathcal{\bar{H}}$ as the RKHS of functions $f:%
\mathcal{S}\mapsto%
\mathbb{R}%
$ with kernel $\bar{k}$, and $\mathcal{H}_{1}$ as the RKHS of functions $f:%
\mathcal{S}\mapsto%
\mathbb{R}%
$ with kernel equal to $1$, which only consists of the constant function.
Since $k_{0}=\bar{k}+1$, $\mathcal{H}$ contains all functions that can be
written as linear combinations of $\mathcal{\bar{H}}$ and $\mathcal{H}_{1}$
(see, for instance, Theorem 2.16 in \cite{Saitoh2016}). Thus $\mathcal{H}$
contains the constant function, and $||1||_{\mathcal{H}}<\infty$.
Furthermore, since for any $f\in\mathcal{H}$, $|f(r)|=\langle
f(\cdot),k_{0}(\cdot,r)\rangle_{\mathcal{H}}\leq||f||_{\mathcal{H}%
}\cdot||k_{0}(\cdot,r)||_{\mathcal{H}}\leq\sqrt{K}||f||_{\mathcal{H}}$, we
have 
\begin{equation}
\sup_{r\in\mathcal{S}}|f(r)|\leq\sqrt{K}\cdot||f||_{\mathcal{H}}.
\label{sup_ineq}
\end{equation}

As in RBV, view the operators above as operators on $\mathcal{H\mapsto H}$.
The operator norm $||A||$ of the operator $A:\mathcal{H\mapsto H}$ is
defined as $\sup_{||f||_{\mathcal{H}}=1}||Af||_{\mathcal{H}}$, and $A$ is
called bounded if $||A||<\infty$. A bounded operator $A$ is Hilbert-Schmidt
if $\sum_{j=1}^{\infty}||Ae_{j}||<\infty$ for some (any) orthonormal basis $%
e_{j}$. The space of Hilbert-Schmidt operators is a Hilbert space endowed
with the norm $||A||_{HS}=\sqrt{\sum_{j=1}^{\infty}\langle
Ae_{j},Ae_{j}\rangle_{\mathcal{H}}}$, and for any Hilbert-Schmidt operator $%
A $ and bounded operator $B$, 
\begin{equation}
||AB||_{HS}\leq||A||_{HS}||B||\text{, }||BA||_{HS}\leq||B||\cdot||A||_{HS}%
\text{.}  \label{HS_op_normineq}
\end{equation}
By Theorem 7 of RBV, $L$ and $L_{n}$ are Hilbert-Schmidt.

Furthermore, for any $f\in\mathcal{H}$, 
\begin{eqnarray*}
||Mf||_{\mathcal{H}} & = & ||f-\int f(r)dG(r)||_{\mathcal{H}} \\
& \leq & ||f||_{\mathcal{H}}+||1||_{\mathcal{H}}\int f(r)dG(r) \\
& \leq & ||f||_{\mathcal{H}}+||1||_{\mathcal{H}}\sup_{r\in\mathcal{S}}|f(r)|
\end{eqnarray*}
so that (\ref{sup_ineq}) implies that $||M||$ is a bounded operator. By the
same argument, so is $M_{n}$ (almost surely). Thus, from (\ref%
{HS_op_normineq}), also $\bar{L}$, $\bar{L}_{n}$ and $\hat{L}_{n}$ are
Hilbert-Schmidt for almost all $\mathbf{s}_{n}$.

Conditioning on $\mathbf{s}_{n}$ throughout, we have the almost sure
inequalities%
\begin{equation*}
||\hat{L}_{n}-\bar{L}||_{HS}\leq ||\hat{L}_{n}-\bar{L}_{n}||_{HS}+||\bar{L}%
_{n}-\bar{L}||_{HS}
\end{equation*}%
and using (\ref{HS_op_normineq}) 
\begin{eqnarray*}
||\hat{L}_{n}-\bar{L}_{n}||_{HS} &\leq
&||(M_{n}-M)L_{n}M_{n}||_{HS}+||ML_{n}(M_{n}-M)||_{HS} \\
&\leq &||M_{n}-M||\cdot ||M_{n}||\cdot ||L_{n}||_{HS}+||M_{n}-M||\cdot
||M||\cdot ||L_{n}||_{HS}
\end{eqnarray*}%
and 
\begin{eqnarray*}
||(M_{n}-M)f||_{\mathcal{H}} &=&\left\Vert \int f(r)dG_{n}(r)-\int
f(r)dG(r)\right\Vert _{\mathcal{H}} \\
&=&||1||_{\mathcal{H}}\left\vert \int f(r)dG_{n}(r)-\int
f(r)dG(r)\right\vert .
\end{eqnarray*}

Now consider the sequence of real independent random variables $f(s_{l}),$
which have mean $\mathbb{E}[f(s_{l})]=\int f(r)dG(r)$, and, by (\ref%
{sup_ineq}), are almost surely bounded. Since $\int
f(r)(dG_{n}(r)-dG(r))=n^{-1}\sum_{l=1}^{n}f(s_{l})-\mathbb{E}[f(s_{1})]$, so
that by Hoeffding's inequality, with probability of at least $1-2e^{-\delta
} $ 
\begin{equation*}
\left\vert \int f(r)(dG_{n}(r)-dG(r))\right\vert \leq \sqrt{2\delta }%
n^{-1/2}\sup_{r\in \mathcal{S}}|f(r)|
\end{equation*}%
for all $\delta \geq 0$. This holds for all $f\in \mathcal{H}$, so we
conclude that $||M_{n}-M||=O_{p}(n^{-1/2})$.

Furthermore, applying the same reasoning as in the proof of Theorem 7 of
RBV, $||\bar{L}_{n}-\bar{L}||_{HS}=O_{p}(n^{-1/2})$. Thus, $||\hat{L}_{n}-%
\bar{L}||_{HS}=O_{p}(n^{-1/2})$.

The conclusion now follows from similar arguments as employed in Proposition
10 and 12 of RBV. In particular, note that $\varphi _{i}\in \mathcal{H}$ for
all $i$. Furthermore, $\int \varphi _{i}(s)dG(s)=\lambda _{i}^{-1}\int
\varphi _{i}(r)\bar{k}(r,s)dG(r)dG(s)=0$. Thus, with $e_{i}=\sqrt{\lambda
_{i}}\varphi _{i}\in \mathcal{H}$, $Me_{i}=e_{i}$, and $\langle
e_{i},e_{i}\rangle _{\mathcal{H}}=\langle e_{i}(\cdot ),\lambda
_{i}^{-1}\int \bar{k}(r,\cdot )e_{i}(r)dG(r)\rangle _{\mathcal{H}}=\lambda
_{i}^{-1}\langle e_{i},\bar{L}e_{i}\rangle _{\mathcal{H}}=\lambda
_{i}^{-1}\langle e_{i},Le_{i}\rangle _{\mathcal{H}}=\lambda _{i}^{-1}\int
\langle e_{i}(\cdot ),k_{0}(r,\cdot )\rangle _{\mathcal{H}%
}e_{i}(r)dG(r)=\lambda _{i}^{-1}\int e_{i}^{2}(r)dG(r)=1$, so that $e_{i}$
are normalized eigenvectors of $\bar{L}:\mathcal{H\mapsto H}$. Since $%
\mathcal{H}\subset \mathcal{L}_{G}^{2}$, these are the only eigenfunctions
of $\bar{L}:\mathcal{H\mapsto H}$ with positive eigenvalue, so that the
spectrum of $\bar{L}$ is equal to $\{\lambda _{i}\}_{i=1}^{\infty }$ (cf.
Proposition 8 of RBV).

Also, $\hat{\varphi}_{i}\in \mathcal{H}$, and since $\mathbf{\hat{v}}_{i}$
is the eigenvector of $n^{-1}\mathbf{\hat{K}}_{n}$ with eigenvalue $\hat{%
\lambda}_{i}$, $n^{-1}\mathbf{\hat{K}}_{n}\mathbf{\hat{v}}_{i}=\hat{\lambda}%
_{i}\mathbf{\hat{v}}_{i}$, we obtain for $\hat{\lambda}_{i}>0$ that 
\begin{eqnarray*}
\hat{L}_{n}(\hat{\varphi}_{i})(\cdot ) &=&\int \hat{k}_{n}(r,\cdot )\hat{%
\varphi}_{i}(r)dG_{n}(r) \\
&=&n^{-1}\sum_{j=1}^{n}\hat{k}_{n}(\cdot ,s_{j})\hat{\varphi}_{i}(s_{j}) \\
&=&n^{-2}\hat{\lambda}_{i}^{-1}\sum_{j=1}^{n}\hat{k}_{n}(\cdot
,s_{j})\sum_{l=1}^{n}\hat{v}_{i,l}\hat{k}_{n}(s_{j},s_{l}) \\
&=&n^{-1}\sum_{j=1}^{n}\hat{k}_{n}(\cdot ,s_{j})\hat{v}_{i,j} \\
&=&\hat{\lambda}_{i}\hat{\varphi}_{i}(\cdot )
\end{eqnarray*}%
and 
\begin{equation*}
\int \hat{\varphi}_{i}(r)^{2}dG_{n}(r)=n^{-3}\hat{\lambda}%
_{i}^{-2}\sum_{j=1}^{n}\sum_{\ell =1}^{n}\sum_{t=1}^{n}\hat{v}_{i,j}\hat{k}%
_{n}(s_{j},s_{\ell })\hat{k}_{n}(s_{\ell },s_{t})\hat{v}_{i,t}=1.
\end{equation*}%
Furthermore, from $\sum_{l=1}^{n}\hat{v}_{i,l}=0$, also $\int \hat{\varphi}%
_{i}(s)dG_{n}(s)=0$, so that $M_{n}\hat{e}_{i}=\hat{e}_{i}$. Thus, with $%
\hat{e}_{i}=\sqrt{\hat{\lambda}_{i}}\hat{\varphi}_{i}\in \mathcal{H},$ $%
\langle \hat{e}_{i},\hat{e}_{i}\rangle _{\mathcal{H}}=\langle \hat{e}%
_{i}(\cdot ),\hat{\lambda}_{i}^{-1}\int \hat{k}_{n}(r,\cdot )\hat{e}%
_{i}(r)dG_{n}(r)\rangle _{\mathcal{H}}=\hat{\lambda}_{i}^{-1}\langle \hat{e}%
_{i},\hat{L}_{n}\hat{e}_{i}\rangle _{\mathcal{H}}=\hat{\lambda}%
_{i}^{-1}\langle \hat{e}_{i},L_{n}\hat{e}_{i}\rangle _{\mathcal{H}}=\hat{%
\lambda}_{i}^{-1}\int \langle \hat{e}_{i}(\cdot ),k_{0}(r,\cdot )\rangle _{%
\mathcal{H}}\hat{e}_{i}(r)dG_{n}(r)=\hat{\lambda}_{i}^{-1}\int \hat{e}%
_{i}(r)^{2}dG_{n}(r)=1.$ Therefore $\hat{e}_{i}$ are normalized
eigenfunctions of $\hat{L}_{n}:\mathcal{H\mapsto H}$, and since all $f\in 
\mathcal{H}$ that are orthogonal to $\hat{e}_{i}$, $i=1,\ldots ,n$ are in
the kernel of $\hat{L}_{n}$, these are the only eigenfunctions of $\bar{L}:%
\mathcal{H\mapsto H}$ with positive eigenvalue, so the spectrum of $\hat{L}%
_{n}:\mathcal{H\mapsto H}$ is equal to $\{\hat{\lambda}_{i}\}_{i=1}^{n}$
(cf. Proposition 9 of RBV).

Part (b) of the lemma now follows from $||\hat{L}_{n}-\bar{L}%
||_{HS}^{2}=O_{p}(n^{-1})$ and the development on page 911 of RBV.

To establish (\ref{eq:rot_weak}), note that with the projection operators $%
P^{q}:\mathcal{H\mapsto H}$ and $\hat{P}^{q}:\mathcal{H\mapsto H}$ defined
via $P^{q}(f)(\cdot )=\sum_{i=1}^{q}\langle f,e_{i}\rangle _{\mathcal{H}%
}e_{i}(\cdot )$ and $\hat{P}^{q}(f)(\cdot )=\sum_{i=1}^{q}\langle f,\hat{e}%
_{i}\rangle _{\mathcal{H}}\hat{e}_{i}(\cdot )$, by Proposition 6 of RBV, $||%
\hat{P}^{q}-P^{q}||_{HS}\leq 2(\lambda _{q}-\lambda _{q+1})^{-1}||\hat{L}%
_{n}-\bar{L}||_{HS}+o_{p}(n^{-1/2})=O_{p}(n^{-1/2})$. Define the $q\times q$
matrix $\mathbf{\tilde{O}}_{n}$ with $i,j$th element $\tilde{O}%
_{n,ij}=\langle \hat{e}_{i},e_{j}\rangle _{\mathcal{H}}$. Then the $j,t$th
element of $\mathbf{\tilde{O}}_{n}^{\prime }\mathbf{\tilde{O}}_{n}$ is given
by $\sum_{i=1}^{q}\tilde{O}_{n,ij}\tilde{O}_{n,it}=\sum_{i=1}^{q}\langle 
\hat{e}_{i},e_{j}\rangle _{\mathcal{H}}\langle \hat{e}_{i},e_{t}\rangle _{%
\mathcal{H}}=\langle e_{j},\hat{P}^{q}(e_{t})\rangle _{\mathcal{H}}$, and $%
\mathbf{1}[j=t]=\langle e_{j},P^{q}(e_{t})\rangle _{\mathcal{H}}$, so that
by the Cauchy-Schwarz inequality%
\begin{eqnarray*}
\left\vert \sum_{i=1}^{q}\tilde{O}_{n,ij}\tilde{O}_{n,it}-\mathbf{1}%
[j=t]\right\vert &=&\left\vert \langle e_{j},(\hat{P}^{q}-P^{q})e_{t}\rangle
_{\mathcal{H}}\right\vert \\
&\leq &||\hat{P}^{q}-P^{q}||_{HS}=O_{p}(n^{-1/2}).
\end{eqnarray*}%
Thus $|\mathbf{|\tilde{O}}_{n}^{\prime }\mathbf{\tilde{O}}_{n}-\mathbf{I}%
_{q}||=O_{p}(n^{-1/2})$, and with $\mathbf{\hat{O}}_{n}=(\mathbf{\tilde{O}}%
_{n}^{\prime }\mathbf{\tilde{O}}_{n})^{-1/2}\mathbf{\tilde{O}}_{n}$, also $||%
\mathbf{\hat{O}}_{n}-\mathbf{\tilde{O}}_{n}||=O_{p}(n^{-1/2})$. Furthermore,
with $\hat{r}_{i}^{2}=\lambda _{i}/\hat{\lambda}_{i}\overset{p}{\rightarrow }%
1$ using part (b) of the lemma, 
\begin{eqnarray*}
\sqrt{\lambda _{i}}||\sum_{j=1}^{q}\hat{O}_{n,ij}\hat{\varphi}_{j}-\varphi
_{i}||_{\mathcal{H}} &=&||\hat{r}_{i}\sum_{j=1}^{q}\hat{O}_{n,ij}\hat{e}%
_{j}-e_{i}||_{\mathcal{H}} \\
&\leq &||\sum_{j=1}^{q}\tilde{O}_{n,ij}\hat{e}_{j}-e_{i}||_{\mathcal{H}%
}+||\sum_{j=1}^{q}(\hat{r}_{i}\hat{O}_{n,ij}-\tilde{O}_{n,ij})\hat{e}_{j}||_{%
\mathcal{H}} \\
&\leq &||(\hat{P}^{q}-P^{q})e_{i}||_{\mathcal{H}}+\sum_{j=1}^{q}|\hat{r}_{i}%
\hat{O}_{n,ij}-\tilde{O}_{n,ij}| \\
&\leq &||\hat{P}^{q}-P^{q}||_{HS}+\sum_{j=1}^{q}|\hat{r}_{i}\hat{O}_{n,ij}-%
\tilde{O}_{n,ij}|=O_{p}(n^{-1/2})
\end{eqnarray*}%
so (\ref{eq:rot_weak}) follows from (\ref{sup_ineq}).

The claim in part (a) of the lemma now follows by induction from (\ref%
{eq:rot_weak}): For $p=1$, this follows directly. Suppose the result holds
for $p-1$, and let $\mathbf{\hat{O}}_{B}=\limfunc{diag}(\mathbf{\hat{O}}%
_{(1)},\ldots ,\mathbf{\hat{O}}_{(p-1)})$, so that 
\begin{equation}
\sup_{s\in \mathcal{S}}||\mathbf{\hat{O}}_{B}\mathbf{\hat{\varphi}}_{B}(s)-%
\mathbf{\varphi }_{B}(s)||=O_{p}(n^{-1/2}),  \label{eq:ind_premise}
\end{equation}%
with $\mathbf{\varphi }_{B}$ and $\mathbf{\hat{\varphi}}_{B}$ the vector of
the first $\sum_{j=1}^{p-1}m_{j}$ eigenfunctions. Now let 
\begin{equation*}
\mathbf{\hat{O}}_{I}=\left( 
\begin{array}{cc}
\mathbf{\hat{O}}_{11} & \mathbf{\hat{O}}_{12} \\ 
\mathbf{\hat{O}}_{21} & \mathbf{\hat{O}}_{22}%
\end{array}%
\right)
\end{equation*}%
be the $(\sum_{j=1}^{p}m_{j})\times (\sum_{j=1}^{p}m_{j})$ matrix $\mathbf{%
\hat{O}}_{n}$ of (\ref{eq:rot_weak}) applied with $q=\sum_{j=1}^{p}m_{j}$,
with $\mathbf{\hat{O}}_{11}$ of the same dimensions as $\mathbf{\hat{O}}_{B}$%
. Let $\mathbf{\varphi }_{I-B}$ and $\mathbf{\hat{\varphi}}_{I-B}$ be the $%
m_{p}\times 1$ vectors of eigenfunctions with indices $%
\sum_{j=1}^{p-1}m_{j}+1,\ldots ,\sum_{j=1}^{p}m_{j}$, so that by the
conclusion of (\ref{eq:rot_weak}), $\sup_{s\in \mathcal{S}}||\mathbf{\hat{O}}%
_{11}\mathbf{\hat{\varphi}}_{B}(s)+\mathbf{\hat{O}}_{12}\mathbf{\hat{\varphi}%
}_{I-B}(s)-\mathbf{\varphi }_{B}(s)||=O_{p}(n^{-1/2})$ and $\sup_{s\in 
\mathcal{S}}||\mathbf{\hat{O}}_{21}\mathbf{\hat{\varphi}}_{B}(s)+\mathbf{%
\hat{O}}_{22}\mathbf{\hat{\varphi}}_{I-B}(s)-\mathbf{\varphi }%
_{I-B}(s)||=O_{p}(n^{-1/2})$. In conjunction with (\ref{eq:ind_premise}),
the former yields $\sup_{s\in \mathcal{S}}||(\mathbf{\hat{O}}_{11}-\mathbf{%
\hat{O}}_{B})\mathbf{\hat{\varphi}}_{B}(s)+\mathbf{\hat{O}}_{12}\mathbf{\hat{%
\varphi}}_{I-B}(s)||=O_{p}(n^{-1/2})$, which implies in light of (\ref%
{eq:rot_weak}) and the linear independence of eigenvectors that both $||%
\mathbf{\hat{O}}_{11}-\mathbf{\hat{O}}_{B}||=O_{p}(n^{-1/2})$ and $||\mathbf{%
\hat{O}}_{12}||=O_{p}(n^{-1/2})$. Since $\mathbf{\hat{O}}_{I}$ and $\mathbf{%
\hat{O}}_{B}$ are rotation matrices, $\mathbf{\hat{O}}_{B}^{\prime }\mathbf{%
\hat{O}}_{B}=\mathbf{\hat{O}}_{11}^{\prime }\mathbf{\hat{O}}_{11}+\mathbf{%
\hat{O}}_{21}^{\prime }\mathbf{\hat{O}}_{21}=\mathbf{I}$, so that $||\mathbf{%
\hat{O}}_{11}-\mathbf{\hat{O}}_{B}||=O_{p}(n^{-1/2})$ further implies $||%
\mathbf{\hat{O}}_{21}||=O_{p}(n^{-1/2})$. We conclude that also $\sup_{s\in 
\mathcal{S}}||\mathbf{\hat{O}}_{22}\mathbf{\hat{\varphi}}_{I-B}(s)-\mathbf{%
\varphi }_{I-B}(s)||=O_{p}(n^{-1/2})$, so that the result for $p$ holds with 
$\mathbf{\hat{O}}_{(p)}=\mathbf{\hat{O}}_{22}$, which concludes the proof. $%
\blacksquare $

\bigskip

\textbf{Proof of Theorem \ref{thm:cv_n}: }Suppose $\max (\overline{\func{cv}}%
^{2}-\func{cv}_{n}^{2},0)\overset{p}{\rightarrow }0$ does not hold. Then
there exists $\delta >0$ such that $\limsup_{n\rightarrow \infty }\mathbb{P}(%
\overline{\func{cv}}^{2}-\func{cv}_{n}^{2}>\delta )>\delta $. Define $%
\varkappa (\kappa ,\overline{\func{cv}}^{2})=\mathbb{P}\left(
\sum_{i=0}^{\infty }\omega _{i}(\kappa ,\overline{\func{cv}}%
)Z_{i}^{2}>0\right) $, so that $\sup_{0\leq \kappa <1}\varkappa (\kappa ,%
\overline{\func{cv}}^{2})=\alpha $ by definition of $\overline{\func{cv}}$.
By continuity of $\varkappa $, there exists $0\leq \kappa _{0}<1$ and $%
\overline{\func{cv}}^{2}-\delta /2\leq \overline{\func{cv}}_{0}^{2}\leq 
\overline{\func{cv}}^{2}$ such that $\varkappa (\kappa _{0},\overline{\func{%
cv}}_{0}^{2})=\alpha $. If $\kappa _{0}=0$, set $c_{n,1}=c_{n,0}$.
Otherwise, let $c_{n,1}\rightarrow \infty $ be such that the corresponding $%
a_{n,1}=c_{n,1}^{d}/n\rightarrow a_{1}$ satisfies $a_{1}\sigma
_{B}^{0}(0)/(a_{1}\sigma _{B}^{0}(0)+\int \sigma _{B}^{0}(s)ds)=\kappa _{0}$%
. Now let $\func{cv}_{n,1}^{2}$ solve 
\begin{equation*}
\mathbb{P}_{\mathbf{\Sigma (}c_{n,1})}^{0}(\tau _{n}^{2}\geq \func{cv}%
_{n,1}^{2}|\mathbf{s}_{n})=\alpha \ \ \text{ a.s.},
\end{equation*}%
so that clearly, $\func{cv}_{n,1}^{2}\leq \func{cv}_{n}^{2}$ a.s. for all
large enough $n$. Thus, with $\mathcal{A}_{n}$ the event that $\mathbf{s}%
_{n} $ takes on a value such that $\overline{\func{cv}}^{2}-\func{cv}%
_{n^{\prime },1}^{2}>\delta $, we also have $\limsup_{n\rightarrow \infty }%
\mathbb{P}(\mathcal{A}_{n})>\delta $, and there exists a subsequence $%
n^{\prime }\rightarrow \infty $ of $n$ such that $\mathbb{P}(\mathcal{A}%
_{n^{\prime }})>\delta $ for all $n^{\prime }$.

For all such $n^{\prime }$, 
\begin{equation}
\alpha =\mathbb{P}_{\mathbf{\Sigma (}c_{n^{\prime },1})}^{0}(\tau
_{n^{\prime }}^{2}\geq \func{cv}_{n^{\prime },1}^{2}|\mathcal{A}_{n^{\prime
}})\geq \mathbb{P}_{\mathbf{\Sigma (}c_{n^{\prime },1})}^{0}(\tau
_{n^{\prime }}^{2}\geq \overline{\func{cv}}^{2}-\delta |\mathcal{A}%
_{n^{\prime }})\text{ a.s.}  \label{eq:cvnp1}
\end{equation}%
and by Theorem \ref{thm:kernel_conv}, $\mathbb{P}_{\mathbf{\Sigma (}%
c_{n^{\prime },1})}^{0}(\tau _{n^{\prime }}^{2}\geq \overline{\func{cv}}%
^{2}-\delta |\mathcal{A}_{n^{\prime }})\rightarrow \varkappa (\kappa _{0},%
\overline{\func{cv}}^{2}-\delta )>\alpha $. This contradicts (\ref{eq:cvnp1}%
), and the result follows. $\blacksquare $

\bigskip {}{}

\begin{theorem}
\label{thm:asysize_qhat}Let $\hat{q}_{n}$ be an arbitrary function of $%
\mathbf{s}_{n}$ taking values in $\mathcal{Q}=\{1,2,\ldots,q_{\max}\}$ for
some sample size independent finite and nonrandom $q_{\max}$. Then for a
t-statistic $\tau_{n}(q)$ that satisfies the conditions of Theorem \ref%
{thm:cv_n} for all $q\in\mathcal{Q}$ with critical value $\func{cv}_{n}(q)$
as in (\ref{eq:robust_cv_n}), for any $\epsilon>0$, $\lim\sup_{n\rightarrow%
\infty}\mathbb{P(P}(\tau_{n}^{2}(\hat{q}_{n})>\func{cv}_{n}(\hat{q}_{n})^{2}|%
\mathbf{s}_{n})>\alpha+\epsilon)=0$.
\end{theorem}

\begin{proof}
Suppose otherwise. Then there exists $\epsilon >0$ and a subsequence $%
n^{\prime }\rightarrow \infty $ such that with $\mathcal{B}_{n}=\{\mathbf{s}%
_{n}:\mathbb{P}(\tau _{n}^{2}(\hat{q})>\func{cv}_{n}(\hat{q})^{2}|\mathbf{s}%
_{n})>\alpha +\epsilon \}\subset \mathcal{S}$, 
\begin{equation*}
\lim_{n^{\prime }\rightarrow \infty }\mathbb{P(}\mathbf{s}_{n^{\prime }}\in 
\mathcal{B}_{n^{\prime }})>\epsilon \text{.}
\end{equation*}%
Let $\mathcal{A}_{n,i}=\{\mathbf{s}_{n}:\hat{q}_{n}=i\}$, so that $%
\lim_{n^{\prime }\rightarrow \infty }\sum_{i=1}^{q_{\max }}\mathbb{P(}%
\mathbf{s}_{n^{\prime }}\in \mathcal{B}_{n^{\prime }}\cap \mathcal{A}%
_{n^{\prime },i})>\epsilon $. There hence exists some $1\leq q\leq q_{\max }$
and a further subsequence $n^{\prime \prime }$ of $n^{\prime }$ such that $%
\lim_{n^{\prime \prime }}\mathbb{P(}\mathbf{s}_{n^{\prime \prime }}\in 
\mathcal{B}_{n^{\prime \prime }}\cap \mathcal{A}_{n^{\prime \prime},q})>\epsilon
/q_{\max }$. But along this subsequence, $q$ is fixed, so Theorem \ref%
{thm:cv_n} applies and yields $\lim_{n^{\prime \prime }\rightarrow \infty }%
\mathbb{P(}\mathbf{s}_{n^{\prime \prime }}\in \mathcal{B}_{n^{\prime \prime
}}\cap \mathcal{A}_{n^{\prime \prime},q})\rightarrow 0$, yielding the desired
contradiction.
\end{proof}

\bigskip

The proof of Theorem \ref{thm:robust} relies on some preliminary results.

\begin{lemma}
\label{lm:schurcncve}The $%
\mathbb{R}
^{q}\mapsto 
\mathbb{R}
$ function 
\begin{equation*}
J(\mathbf{\eta })=\frac{1}{\pi }\int_{0}^{1}\frac{x^{\frac{q-1}{2}}}{\sqrt{%
(1-x)\prod_{i=1}^{q}(x+\eta _{i})}}dx
\end{equation*}%
with $\mathbf{\eta }=(\eta _{1},\ldots ,\eta _{q})$ is Schur convex.
\end{lemma}

\begin{proof}
By the Schur-Ostrowski criterion (Theorem 3.A.4 in \cite{Marshall2011}), $J$
is Schur convex if (and only if) 
\begin{equation*}
(\eta _{i}-\eta _{j})\left( \frac{\partial J}{\partial \eta _{i}}-\frac{%
\partial J}{\partial \eta _{j}}\right) \geq 0\text{ for all }1\leq i,j\leq q%
\text{.}
\end{equation*}%
With $\tilde{J}=(x+\eta _{i})^{-1/2}(x+\eta _{j})^{-1/2}$, by a direct
calculation, 
\begin{equation*}
(\eta _{i}-\eta _{j})\left( \frac{\partial \tilde{J}}{\partial \eta _{i}}-%
\frac{\partial \tilde{J}}{\partial \eta _{j}}\right) =\frac{(\eta _{i}-\eta
_{j})^{2}}{2(x+\eta _{i})^{3/2}(x+\eta _{j})^{3/2}}\geq 0
\end{equation*}%
so the result follows.
\end{proof}

\begin{lemma}
\label{lm:quadconvex}For any two $q\times q$ positive semi-definite matrices 
$\mathbf{B}_{1}$ and $\mathbf{B}_{2}$ and vectors $\mathbf{v}_{1},\mathbf{v}%
_{2}\in 
\mathbb{R}
^{q}$, and all $p\in \lbrack 0,1]$, 
\begin{multline*}
\varsigma (p)=(p\mathbf{v}_{1}+(1-p)\mathbf{v}_{2})^{\prime }(\mathbf{I}%
_{q}+p\mathbf{B}_{1}+(1-p)\mathbf{B}_{2})^{-1}(p\mathbf{v}_{1}+(1-p)\mathbf{v%
}_{2}) \\
-p\mathbf{v}_{1}^{\prime }(\mathbf{I}_{q}+\mathbf{B}_{1})^{-1}\mathbf{v}%
_{1}-(1-p)\mathbf{v}_{2}^{\prime }(\mathbf{I}_{q}+\mathbf{B}_{2})^{-1}%
\mathbf{v}_{2}\leq 0\text{.}
\end{multline*}
\end{lemma}

\begin{proof}
We first show that $\varsigma (p)$ is convex. Write $\mathbf{G}(p)=\mathbf{I}%
_{q}+p\mathbf{B}_{1}+(1-p)\mathbf{B}_{2}$. The first derivative of the
nonlinear part of $\tfrac{1}{2}\varsigma (p)$ is given by 
\begin{equation*}
(\mathbf{v}_{1}-\mathbf{v}_{2})^{\prime }\mathbf{G}(p)^{-1}(p\mathbf{v}%
_{1}+(1-p)\mathbf{v}_{2})-\tfrac{1}{2}(p\mathbf{v}_{1}+(1-p)\mathbf{v}%
_{2})^{\prime }\mathbf{G}(p)^{-1}(\mathbf{B}_{1}-\mathbf{B}_{2})\mathbf{G}%
(p)^{-1}(p\mathbf{v}_{1}+(1-p)\mathbf{v}_{2})
\end{equation*}%
so that the second derivative of $\tfrac{1}{2}\varsigma (p)$ equals 
\begin{multline*}
(\mathbf{v}_{1}-\mathbf{v}_{2})^{\prime }\mathbf{G}(p)^{-1}(\mathbf{v}_{1}-%
\mathbf{v}_{2})-2(\mathbf{v}_{1}-\mathbf{v}_{2})^{\prime }\mathbf{G}(p)^{-1}(%
\mathbf{B}_{1}-\mathbf{B}_{2})\mathbf{G}(p)^{-1}(p\mathbf{v}_{1}+(1-p)%
\mathbf{v}_{2}) \\
+(p\mathbf{v}_{1}+(1-p)\mathbf{v}_{2})^{\prime }\mathbf{G}(p)^{-1}(\mathbf{B}%
_{1}-\mathbf{B}_{2})\mathbf{G}(p)^{-1}(\mathbf{B}_{1}-\mathbf{B}_{2})\mathbf{%
G}(p)^{-1}(p\mathbf{v}_{1}+(1-p)\mathbf{v}_{2}).
\end{multline*}%
With $\mathbf{\Delta }(p)=\mathbf{G}(p)^{-1/2}(\mathbf{v}_{1}-\mathbf{v}_{2})
$ and $\mathbf{r}(p)=-\mathbf{G}(p)^{-1/2}(\mathbf{B}_{1}-\mathbf{B}_{2})%
\mathbf{G}(p)^{-1}(p\mathbf{v}_{1}+(1-p)\mathbf{v}_{2})$, the second
derivative may be rewritten as 
\begin{equation*}
\left( 
\begin{array}{c}
\mathbf{\Delta }(p) \\ 
\mathbf{r}(p)%
\end{array}%
\right) ^{\prime }\left( 
\begin{array}{cc}
\mathbf{I}_{q} & \mathbf{I}_{q} \\ 
\mathbf{I}_{q} & \mathbf{I}_{q}%
\end{array}%
\right) \left( 
\begin{array}{c}
\mathbf{\Delta }(p) \\ 
\mathbf{r}(p)%
\end{array}%
\right) \geq 0
\end{equation*}%
and convexity follows. Thus $\max_{p\in \lbrack 0,1]}\varsigma (p)\leq \max
(\varsigma (1),\varsigma (0))=0$.
\end{proof}

\begin{lemma}
\label{lm:lam1Abar}Let $\mathbf{A}_{1}=\int \mathbf{P}^{-1}\mathbf{D}(\func{%
cv})\mathbf{\Omega }(\theta )\mathbf{P}dF(\theta )$. The $q+1$ eigenvalues
of $\mathbf{A}_{1}$ are real, and only one is positive, and the same holds
for $\mathbf{A}(\theta )$, $\theta \in \Theta $. Furthermore, $\lambda _{1}(%
\mathbf{A}_{1})\geq 1$.{}
\end{lemma}

\begin{proof}
By similarity, the eigenvalues of $\mathbf{A}_{1}$ are equal to those of $%
\mathbf{PA}_{1}\mathbf{P}^{-1}$, which in turn is similar to the symmetric
matrix 
\begin{equation*}
\left( 
\begin{array}{cc}
\mathbf{l}^{\prime }\mathbf{\Sigma }_{1}\mathbf{l} & \mathbf{l}^{\prime }%
\mathbf{\Sigma }_{1}\mathbf{\tilde{W}} \\ 
\mathbf{\tilde{W}}^{\prime }\mathbf{\Sigma }_{1}\mathbf{l} & \mathbf{\tilde{W%
}}^{\prime }\mathbf{\Sigma }_{1}\mathbf{\tilde{W}}%
\end{array}%
\right) ^{1/2}\left( 
\begin{array}{cc}
1 & 0 \\ 
0 & -\mathbf{I}_{q}%
\end{array}%
\right) \left( 
\begin{array}{cc}
\mathbf{l}^{\prime }\mathbf{\Sigma }_{1}\mathbf{l} & \mathbf{l}^{\prime }%
\mathbf{\Sigma }_{1}\mathbf{\tilde{W}} \\ 
\mathbf{\tilde{W}}^{\prime }\mathbf{\Sigma }_{1}\mathbf{l} & \mathbf{\tilde{W%
}}^{\prime }\mathbf{\Sigma }_{1}\mathbf{\tilde{W}}%
\end{array}%
\right) ^{1/2}
\end{equation*}%
with $\mathbf{\tilde{W}=(l,W/}\func{cv})$, and the first claim follows for $%
\mathbf{A}_{1}$. The claim for $\mathbf{A}(\theta )$ follows from the same
argument.

For the last claim, let $\bar{h}:%
\mathbb{R}
\mapsto 
\mathbb{R}
$ 
\begin{equation*}
\bar{h}(t)=1-t\mathbf{l}^{\prime }\mathbf{\Sigma }_{1}\mathbf{l}+t^{2}%
\mathbf{l}^{\prime }\mathbf{\Sigma }_{1}\mathbf{\tilde{W}}(\mathbf{I}_{q}+t%
\mathbf{\tilde{W}}^{\prime }\mathbf{\Sigma }_{1}\mathbf{\tilde{W}})^{-1}%
\mathbf{\tilde{W}}^{\prime }\mathbf{\Sigma }_{1}\mathbf{l}.
\end{equation*}%
Note that $\bar{h}(t)$ is weakly decreasing in $t>0$, since with $\mathbf{%
\tilde{H}}=-t\mathbf{\tilde{W}}(\mathbf{I}_{q}+t\mathbf{\tilde{W}}^{\prime }%
\mathbf{\Sigma }_{1}\mathbf{\tilde{W}})^{-1}\mathbf{\tilde{W}}^{\prime }%
\mathbf{\Sigma }_{1}\mathbf{l}$%
\begin{equation*}
\bar{h}^{\prime }(t)=-\left( 
\begin{array}{c}
\mathbf{l} \\ 
\mathbf{\tilde{H}}%
\end{array}%
\right) ^{\prime }\left( 
\begin{array}{cc}
\mathbf{\Sigma }_{1} & \mathbf{\Sigma }_{1} \\ 
\mathbf{\Sigma }_{1} & \mathbf{\Sigma }_{1}%
\end{array}%
\right) \left( 
\begin{array}{c}
\mathbf{l} \\ 
\mathbf{\tilde{H}}%
\end{array}%
\right) <0.
\end{equation*}

The characteristic polynomial of $\mathbf{A}_{1}$ is given by 
\begin{eqnarray*}
&&\det \left( 
\begin{array}{cc}
s-\mathbf{l}^{\prime }\mathbf{\Sigma }_{1}\mathbf{l} & \mathbf{l}^{\prime }%
\mathbf{\Sigma }_{1}\mathbf{\tilde{W}} \\ 
-\mathbf{\tilde{W}}^{\prime }\mathbf{\Sigma }_{1}\mathbf{l} & s\mathbf{I}%
_{q}+\mathbf{\tilde{W}}^{\prime }\mathbf{\Sigma }_{1}\mathbf{\tilde{W}}%
\end{array}%
\right)  \\
&=&(s-\mathbf{l}^{\prime }\mathbf{\Sigma }_{1}\mathbf{l}+\mathbf{l}^{\prime }%
\mathbf{\Sigma }_{1}\mathbf{\tilde{W}}(s\mathbf{I}_{q}+\mathbf{\tilde{W}}%
^{\prime }\mathbf{\Sigma }_{1}\mathbf{\tilde{W}})^{-1}\mathbf{\tilde{W}}%
^{\prime }\mathbf{\Sigma }_{1}\mathbf{l})\det (s\mathbf{I}_{q}+\mathbf{%
\tilde{W}}^{\prime }\mathbf{\Sigma }_{1}\mathbf{\tilde{W}}) \\
&=&s\bar{h}(s^{-1})\det (s\mathbf{I}_{q}+\mathbf{\tilde{W}}^{\prime }\mathbf{%
\Sigma }_{1}\mathbf{\tilde{W}})
\end{eqnarray*}%
so that $\lambda _{1}(\mathbf{A}_{1})$ satisfies $\bar{h}(1/\lambda _{1}(%
\mathbf{A}_{1}))=0$. Similarly, $1/\lambda _{1}(\mathbf{A}(\theta ))=1$ is a
root of 
\begin{equation*}
h_{\theta }(t)=1-t\mathbf{l}^{\prime }\mathbf{\Sigma }(\theta )\mathbf{l}%
+t^{2}\mathbf{l}^{\prime }\mathbf{\Sigma }(\theta )\mathbf{\tilde{W}}(%
\mathbf{I}_{q}+t\mathbf{\tilde{W}}^{\prime }\mathbf{\Sigma }(\theta )\mathbf{%
\tilde{W}})^{-1}\mathbf{\tilde{W}}^{\prime }\mathbf{\Sigma }(\theta )\mathbf{%
l}.
\end{equation*}%
By Lemma \ref{lm:quadconvex}, for any $t>0$, 
\begin{eqnarray*}
&&\mathbf{l}^{\prime }\mathbf{\Sigma }_{1}\mathbf{\tilde{W}}(\mathbf{I}_{q}+t%
\mathbf{\tilde{W}}^{\prime }\mathbf{\Sigma }_{1}\mathbf{\tilde{W}})^{-1}%
\mathbf{\tilde{W}}^{\prime }\mathbf{\Sigma }_{1}\mathbf{l} \\
&=&\left( \int \mathbf{\tilde{W}}^{\prime }\mathbf{\Sigma }(\theta )\mathbf{l%
}dF(\theta )\right) ^{\prime }\left( \mathbf{I}_{q}+t\int \mathbf{\tilde{W}}%
^{\prime }\mathbf{\Sigma }(\theta )\mathbf{\tilde{W}}dF(\theta )\right)
^{-1}\left( \int \mathbf{\tilde{W}}^{\prime }\mathbf{\Sigma }(\theta )%
\mathbf{l}dF(\theta )\right)  \\
&\leq &\int \mathbf{l}^{\prime }\mathbf{\Sigma }(\theta )\mathbf{\tilde{W}}(%
\mathbf{I}_{q}+t\mathbf{\tilde{W}}^{\prime }\mathbf{\Sigma }(\theta )\mathbf{%
\tilde{W}})^{-1}\mathbf{\tilde{W}}^{\prime }\mathbf{\Sigma }(\theta )\mathbf{%
l}dF(\theta ).
\end{eqnarray*}%
Thus, $\bar{h}(t)\leq \int h_{\theta }(t)dF(\theta )$, and from $h_{\theta
}(1)=0$ for all $\theta $, $\bar{h}(1)\leq 0$. Since $h$ is decreasing, its
root $1/\lambda _{1}(\mathbf{A}_{1})$ must thus be smaller than unity, and
the conclusion follows. 
\end{proof}

\bigskip

\textbf{Proof of Theorem \ref{thm:robust}: }Proceeding as in the proof of
Theorem \ref{thm:tau_finite_q}, $\mathbb{P}_{\mathbf{\Sigma }_{1}}(\tau ^{2}(%
\mathbf{W}\mathbf{W}^{\prime })>\func{cv}^{2})=\mathbb{P}\left(
Z_{0}^{2}\geq \sum_{i=1}^{q}\bar{\eta}_{i}Z_{i}^{2}\right) $ with $\bar{\eta}%
_{i}=\lambda _{i}\left( -\mathbf{A}_{1}\right) /\lambda _{1}(\mathbf{A}_{1})$%
. By Lemma \ref{lm:lam1Abar}, $\bar{\eta}_{i}\geq 0$ for $i=1,\ldots ,q$.
For future reference, note that $\mathbb{P}_{\mathbf{\Sigma }_{0}}(\tau ^{2}(%
\mathbf{W}\mathbf{W}^{\prime })>\func{cv}^{2})=\alpha $ yields 
\begin{equation}
\mathbb{P}\left( Z_{0}^{2}\geq \sum_{i=1}^{q}\eta _{i}Z_{i}^{2}\right) \leq
\alpha .  \label{zeroprob}
\end{equation}%
for $\eta _{i}=\lambda _{i}\left( -\mathbf{A}_{0}\right) $.

In the following, we write $\mathbf{a}\prec \mathbf{b}$ for two vectors $%
\mathbf{a},\mathbf{b}\in 
\mathbb{R}
^{q}$ to indicate that $\mathbf{b}$ majorizes $\mathbf{a}$, that is, with
the elements of $a_{i}$ and $b_{i}$ sorted in descending order, 
\begin{equation*}
\sum_{i=1}^{j}a_{i}\leq \sum_{i=1}^{j}b_{i}\text{ for all }j=1,\ldots ,q
\end{equation*}%
and $\sum_{i=1}^{q}a_{i}=\sum_{i=1}^{q}b_{i}$. Let $\mathbf{\bar{A}}_{1}=%
\tfrac{1}{2}(\mathbf{A}_{1}+\mathbf{A}_{1}^{\prime })$. From Theorems 9.F.1
and 9.G.1 in \cite{Marshall2011} 
\begin{eqnarray}
(\lambda _{1}(-\mathbf{A}_{1}),\ldots ,\lambda _{q+1}(-\mathbf{A}_{1}))
&\prec &(\lambda _{1}(-\mathbf{\bar{A}}_{1}),\ldots ,\lambda _{q+1}(-\mathbf{%
\bar{A}}_{1}))  \label{majAbarB} \\
&\prec &\left( \int \lambda _{1}(-\mathbf{\bar{A}}(\theta ))dF(\theta
),\ldots ,\right.  \notag \\
&&\ \ \ \ \ \ \ \ \ \ \ \ \ \left. \int \lambda _{q}(-\mathbf{\bar{A}}%
(\theta ))dF(\theta ),\int \lambda _{q+1}(-\mathbf{\bar{A}}(\theta
))dF(\theta )\right) .  \notag
\end{eqnarray}%
Since $\int \lambda _{q+1}(-\mathbf{\bar{A}}(\theta ))dF(\theta )=-\int
\lambda _{1}(\mathbf{\bar{A}}(\theta ))dF(\theta )$ and $\lambda _{q+1}(-%
\mathbf{A}_{1})=-\lambda _{1}(\mathbf{A}_{1})$, we have 
\begin{equation*}
-\lambda _{1}(\mathbf{A}_{1})+\sum_{j=1}^{q}\lambda _{j}(-\mathbf{A}%
_{1})=-\int \lambda _{1}(\mathbf{\bar{A}}(\theta ))dF(\theta
)+\sum_{j=1}^{q}\int \lambda _{j}(-\mathbf{\bar{A}}(\theta ))dF(\theta ).
\end{equation*}%
The majorization result (\ref{majAbarB}) further implies 
\begin{equation}
\lambda _{1}(\mathbf{A}_{1})\leq \lambda _{1}(\mathbf{\bar{A}}_{1})\leq \int
\lambda _{1}(\mathbf{\bar{A}}(\theta ))dF(\theta )  \label{lam1ineq}
\end{equation}%
so that also 
\begin{multline*}
(\lambda _{1}(-\mathbf{A}_{1}),\ldots ,\lambda _{q}(-\mathbf{A}_{1}))\prec
\left( \int \lambda _{1}(-\mathbf{\bar{A}}(\theta ))dF(\theta ),\ldots
,\right. \\
\left. \int \lambda _{q-1}(-\mathbf{\bar{A}}(\theta ))dF(\theta ),\int
\lambda _{q}(-\mathbf{\bar{A}}(\theta ))dF(\theta )-\left( \int \lambda _{1}(%
\mathbf{\bar{A}}(\theta ))dF(\theta ))-\lambda _{1}(\mathbf{A}_{1})\right)
\right) .
\end{multline*}%
with the elements still sorted in descending order. Thus, with $\tilde{\eta}%
_{i}=\int \lambda _{i}(-\mathbf{\bar{A}}(\theta ))dF(\theta )/\lambda _{1}(%
\mathbf{A}_{1})$ for $i=1,\ldots ,q-1$ and 
\begin{equation*}
\tilde{\eta}_{q}=\frac{\int \lambda _{q}(-\mathbf{\bar{A}}(\theta
))dF(\theta )-(\int \lambda _{1}(\mathbf{\bar{A}}(\theta ))dF(\theta
))-\lambda _{1}(\mathbf{A}_{1}))}{\lambda _{1}(\mathbf{A}_{1})}
\end{equation*}%
we have $(\bar{\eta}_{1},\ldots ,\bar{\eta}_{q})\prec (\tilde{\eta}%
_{1},\ldots ,\tilde{\eta}_{q})$, so that by (\ref{eq:BS}) and Lemma \ref%
{lm:schurcncve}, $\mathbb{P}\left( Z_{0}^{2}\geq \sum_{i=1}^{q}\bar{\eta}%
_{i}Z_{i}^{2}\right) \leq \mathbb{\mathbb{P}}\left( Z_{0}^{2}\geq
\sum_{i=1}^{q}\tilde{\eta}_{i}Z_{i}^{2}\right) $.

Now applying (\ref{lam1ineq}) 
\begin{equation*}
\tilde{\eta}_{i}^{\ast }=\int \lambda _{i}(-\mathbf{\bar{A}}(\theta
))dF(\theta )/\int \lambda _{1}(\mathbf{\bar{A}}(\theta ))dF(\theta )\leq 
\tilde{\eta}_{i}
\end{equation*}%
for $i=1,\ldots ,q-1$, and since from Lemma \ref{lm:lam1Abar}, $\lambda _{1}(%
\mathbf{A}_{1})\geq 1$, also 
\begin{equation*}
\tilde{\eta}_{q}^{\ast }=\frac{\int \lambda _{q}(-\mathbf{\bar{A}}(\theta
))dF(\theta )-(\int \lambda _{1}(\mathbf{\bar{A}}(\theta ))dF(\theta )-1)}{%
\int \lambda _{1}(\mathbf{\bar{A}}(\theta ))dF(\theta )}\leq \tilde{\eta}_{q}
\end{equation*}%
provided 
\begin{equation}
\int \lambda _{q}(-\mathbf{\bar{A}}(\theta ))dF(\theta )-\left( \int \lambda
_{1}(\mathbf{\bar{A}}(\theta ))dF(\theta )-1\right) \geq 0.  \label{signlamq}
\end{equation}%
Since $\mathbb{P}(Z_{0}^{2}\geq \left. \sum_{i=1}^{q}\tilde{\eta}%
_{i}Z_{i}^{2}\right. )$ is a decreasing function in $\tilde{\eta}_{i}$, $%
\mathbb{P}\left( Z_{0}^{2}\geq \sum_{i=1}^{q}\tilde{\eta}_{i}Z_{i}^{2}%
\right) \leq \mathbb{P}\left( Z_{0}^{2}\geq \sum_{i=1}^{q}\tilde{\eta}%
_{i}^{\ast }Z_{i}^{2}\right) .$ By Theorem 3.A.8 of \cite{Marshall2011},
Lemma \ref{lm:schurcncve}, and (\ref{zeroprob}), it now suffices to show
that 
\begin{equation}
\sum_{i=1}^{j}\tilde{\eta}_{q+1-i}^{\ast }\geq \sum_{i=1}^{j}\eta _{q+1-i}
\label{finineq}
\end{equation}%
for all $1\leq j\leq q$, and since $\eta _{q}\geq 0$, this also ensures that
(\ref{signlamq}) holds. Condition (\ref{finineq}) may be rewritten as $%
\sum_{i=1}^{j}\int \nu _{i}(\theta )dF(\theta )\geq 0$, and the result
follows. $\blacksquare $

\newpage

\bibliographystyle{econometrica}
\bibliography{diss}

\end{document}